\documentclass[11pt]{article}
\usepackage{subfig,amsmath,amssymb,amsthm,graphicx,tikz,bbm,bm,geometry,setspace, adjustbox, lscape, caption, multirow, threeparttable, booktabs, setspace,hyperref, siunitx}
\usepackage[authoryear,longnamesfirst]{natbib}
\usepackage[figuresleft]{rotating}
\newtheorem{theorem}{Theorem}
\newtheorem{lemma}{Lemma}[section]

\newtheorem{example}{Example}

\newtheorem{remark}{Remark}

\newtheorem{assumption}{Assumption}

\DeclareMathOperator*{\argmin}{arg\,min}
\newcommand{\E}{\mathbb{E}}
\newcommand{\N}{\mathcal{N}}

\allowdisplaybreaks
\numberwithin{equation}{section}
\title{Robust Inference in Locally Misspecified Bipartite Networks\thanks{We are thankful to St\'{e}phane Bonhomme, \'{A}ureo de Paula, Michael Koles\'{a}r, and seminar participants at Edinburgh, Glasgow, Manchester, and Oxford for comments and suggestions. Yichong Zhang acknowledges the financial support from the NSFC under grant No. 72133002. Any and all errors are our own.} \\ \vspace{2mm}
}
\author{Luis E. Candelaria\thanks{University of Warwick.\ E-mail~address:  \texttt{l.candelaria@warwick.ac.uk}.} 
\and 
Yichong Zhang\thanks{Singapore Management University.\ E-mail~address: \texttt{yczhang@smu.edu.sg}.}
}
\date{March 2024}

\begin{document}
\maketitle
\begin{abstract}
This paper introduces a methodology to conduct robust inference in bipartite networks under local misspecification. We focus on a class of dyadic network models with misspecified conditional moment restrictions.
The framework of misspecification is local, as the effect of misspecification varies with the sample size. We utilize this local asymptotic approach to construct a robust estimator that is minimax optimal for the mean square error within a neighborhood of misspecification. Additionally, we introduce bias-aware confidence intervals that account for the effect of the local misspecification. These confidence intervals have the correct asymptotic coverage for the true parameter of interest under sparse network asymptotics. Monte Carlo experiments demonstrate that the robust estimator performs well in finite samples and sparse networks. As an empirical illustration,  we study the formation of a scientific collaboration network among economists.
\end{abstract}

\newpage
\section{Introduction}
Bipartite networks are embedded in a wide range of economic interactions that entail bilateral relations, for example, those between exporters and importers, buyers and sellers, and scientists and research teams. Dyadic regression with two-sided heterogeneity represents a leading approach for analyzing bipartite networks as it can adequately account for dyad dependence and link sparsity (\citealt{graham:2020}). While under correct model specification, it yields valid econometric inference for a structural parameter of interest, its performance under local misspecification has not been explored. To the best of our knowledge, this is the first paper to study the effects of local misspecification on a dyadic regression model for bipartite networks. 

In this paper, we develop a methodology to conduct robust inference in bipartite networks under local misspecification. We focus on a class of dyadic models characterized by conditional moment restrictions when model misspecification prevents these restrictions from holding exactly.  Some situations in dyadic network models that lead to the failure of these moment restrictions include incorrect specification of parametric assumptions, ignoring preferences for latent homophily, and excluding network externalities. Our framework is one of local misspecification, as the magnitude of misspecification vanishes as the sample size grows. This modelling approach is intended as an asymptotic device, which is particularly useful to provide a tractable approximation for the effect of the model misspecification. Moreover, it enables us to perform inference on the structural parameter of interest rather than on a pseudo-true parameter.  

We propose a robust estimator that accounts for the local misspecification on the conditional moment restrictions. To construct this estimator, we consider a sieve approximation to form unconditional moment restrictions that grow in number with the sample size. The robust estimator is then computed as a one-step refinement, which improves on an initial estimator that ignores the model misspecification by incorporating an adjustment term that minimizes the sensitivity of the inferential procedure to local deviations in the unconditional moment restrictions. In particular, this adjustment term is chosen optimally to minimize the worst-case mean squared error (MSE).  As a consequence, the robust estimator is minimax optimal in the sense that given the worst-case bias in a neighbourhood prespecified by the researcher, the robust estimator attains the minimum MSE within a class of regular estimators (\citealt{newey:1990}). 

The resulting confidence intervals are `honest', as they are widened to account for the worst-case bias within a neighborhood of misspecification (\citealt{donoho:1994, armstrong/kolesar:2021, bonhomme/weidner:2022}). This construction ensures that the bias-aware confidence intervals contain the structural parameter of interest under both correct specification and local misspecification with a prespecified coverage probability. The neighborhood of misspecification is determined ex-ante by the researcher but can be adjusted to conduct sensitivity analysis on the structural parameter of interest (\citealt{conley/hansen/rossi:2012}). Additionally, we show that in our setting, the misspecification neighborhood takes on a simple form, which enables us to obtain a closed-form expression for the adjustment term. In Section \ref{sec:model}, we describe in detail the construction of the robust estimator and bias-aware confidence intervals.

To illustrate the performance of our methodology in an empirical application, we study the formation of a scientific collaboration network among economists. While collaboration among scientists is pivotal for research productivity and innovation, the factors that determine the creation of these connections have remained underexplored (\citealt{goyal/etal:l2006,anderson/richards_shubik:2022,hsieh/konig/liu/zimmerman:2022}). In addition, the existing approaches used to analyze these collaborations rely on strong parametric assumptions or limit the degree of unobserved heterogeneity, making them susceptible to model misspecification. To address this gap, we implement our robust methodology to study the factors driving the formation of scientific collaborations among economists. The web of collaborations is described as a bipartite network connecting scientists and research projects, and it is constructed using the universe of published papers in the top American economic journals of general interest during the period of 2000 to 2006. We describe in detail the dataset and network construction in Section \ref{sec:application}.

As parameters of interest, we focus on those capturing assortative matching between scientists and the attributes of a research project. Additionally, we estimate global features of the network, such as the average out-degree. Our analysis documents that homophily in the research field of expertise is a strong predictor for link formation. Moreover, highly productive scientists are, on average, more likely to collaborate on higher-impact projects, measured by aggregate citations. The results also show that, during the sample period, female authors are more likely to sort in teams with a higher share of female authors. Finally, the evidence shows that our methodology leads to a reduction on the MSEs as predicted by out theoretical results. 

Our paper is related to the literature on robust statistics (\citealt{huber:1964,huber/ronchetti:2009,rieder:1994}), sensitivity analysis (\citealt{leamer:1985,imbens:2003, chen/tamer/torgovitsky:2011, nevo/rosen:2012,masten/poirer:2020,masten/poirer:2021}), and local misspecification (\citealt{newey:1985,conley/hansen/rossi:2012,kitamura/otsu/evdokimov:2013,bugni/canay/guggenberger:2012,andrews/gentzkov/shapiro:2017, andrews/gentzkov/shapiro:2020, bugni/ura:2019,armstrong/kolesar:2021,bonhomme/weidner:2022,armstrong/weidner/zeleneev:2022, armstrong/kline/sun:2023, christensen/connault:2023}). The most closely related papers to ours are those by \citet{andrews/gentzkov/shapiro:2017,armstrong/kolesar:2021, bonhomme/weidner:2022} and \citet{christensen/connault:2023}, who rely on local misspecification to conduct sensitivity analysis. In contrast to these papers, this work focuses on a class of models characterized by \textit{conditional} moment restrictions, which is not nested within their settings. Moreover, the object of study is entirely different as we analyze dyadic models for bipartite networks. On a technical side, bipartite network exhibit patterns of dyadic dependence that preclude us from using efficiency results invoked by the previous papers (cf. \citealt{kitamura/otsu/evdokimov:2013, bonhomme/weidner:2022}). 

The paper also contributes to the literature on dyadic regression (\citealt{graham:2017, jochmans:2018, tabord:2019, davezies/dhaultfoeuille/guyonvarch:2021, menzel:2021,GNP21,graham/niu/powell:2019, graham:2022}). The closest paper to ours is the one by \citet{graham:2022}. In this paper, the author derives the asymptotic properties of a logistic regression under sparse network asymptotics and shows that under global misspecification, the pseudo-true parameter minimizes a Kullback–Leibler Information Criterion. In contrast, we construct a one-step robust estimator for a dyadic regression under local misspecification. Our methodology differs substantially as we make use of local approximations to conduct inference on the true structural parameter rather than on a pseudo-true parameter. Notably, our estimator provides robust inference in dyadic models with degree heterogeneity, even in the presence of network externalities (cf. \citealt{pelican/graham:2020}, and see \citet{depaula:2020} for a survey on network models). On a technical note, our paper relies on the Yurinskii-type coupling under dyadic data with an increasing dimension of sieve bases. Our paper complements the kernel-based nonparametric estimation methods considered by \cite{GNP21,graham/niu/powell:2019}. 

The rest of the paper is organized as follows.  In Section \ref{sec:model}, we introduce the class of dyadic regressions with local misspecification. We also show how to construct the robust estimator and the bias-aware confidence intervals. In Section \ref{sec:size}, we prove that the bias-aware confidence intervals have the correct asymptotic size control. In Section \ref{sec:optimality}, we show that the robust estimator is minimax-MSE optimal. Section \ref{sec:simulations} provides Monte Carlo experiments. Section \ref{sec:application} presents an empirical application. Additional results and proofs are collected in the appendices.

\textbf{Notation.} Throughout the paper, we use $||\cdot||_2$, $||\cdot||_{op}$, $||\cdot||_F$ to denote the $\ell_2$-norm of a vector, the operator norm of a matrix, and the Frobenius norm of a matrix, respectively. For an positive integer $N$, we denote $[N]$  as $1,\cdots,N$. 

\section{Model}\label{sec:model}
We consider a bipartite network of scientific collaborations with $N$ scientists and $M$ projects. 
For any $i\in [N]$, let $X_i$ and $A_i$ denote vectors of scientist $i$-specific observed and unobserved attributes. Similarly, for any $j\in [M]$, let $W_j$ and $B_j$ denote vectors of project $j$-specific observed and unobserved attributes. For any pair $i \in [N]$ and $j \in [M]$, $U_{ij}$ denotes a scientist-project-specific unobserved attribute. The total sample size is denoted by $n = N+M$. 

For any finite $n$, the bipartite network $\left[ Y_{ij} \right]_{i\in \left[ N \right], j\in \left[ M \right]}$ is determined according to the $\text{X-W}$-exchangeable graphon
\begin{align}\label{eq:graphon}
Y_{ij} &= h_n\left(v_{ij}, X_i, W_j, A_i, B_j, U_{ij} \right),
\end{align} 
where $h_n$ is a measurable function that maps $(v_{ij}, X_i, W_j, A_i, B_j, U_{ij}) \mapsto \left\{ 0,1 \right\}$, and $v_{ij}$ represents an i.i.d. latent mixing factor that is unidentified, as discussed in \citet{graham:2022}. In our framework, $v_{ij}$ represents a misspecification component that will affect the conditional probability of establishing a collaboration link as specified by \eqref{eq:cprob}. The graphon in \eqref{eq:graphon} is used as a nonparametric data generating process for the bipartite network of scientific collaborations with $N$ scientists and $M$ projects. Intuitively, \eqref{eq:graphon} states that $Y_{ij}=1$ if scientist $i$ collaborates in project $j$, and $0$ otherwise. 

We assume that the conditional probability of establishing a collaboration link is given by 
\begin{equation}\label{eq:cprob}
\mathbb{P}\left[ Y_{ij}=1 \mid X_i, W_j \right] 
=
\int 
\Lambda\left( \alpha_{0,n} + Z_{ij}^\top\beta_{0} + n^{-1/2} v \right)
\pi_{ij}(v) dv,
\end{equation} where $\Lambda(u) = \exp(u)/(1+\exp(u))$ is the standard logistic CDF, $Z_{ij}=z(X_i, W_j)\in \Re^{d_z}$ is a vector-valued distance of the observed attributes $X_i$ and $W_j$, and $v_{ij}$ is a random variable with PDF given by $\pi_{ij}(\cdot)$. The distribution $\pi_{ij}(\cdot)$ can depend on $(X_i,W_j)$, it integrates to one  $\int  \pi_{ij}(v) dv = 1$, and it has mean equal to $\int v \pi_{ij}(v) dv = \eta_{ij}$. Note that $\eta_{ij}$ is a function of $(X_i,W_j)$, and thus, it is indexed by $(i,j)$. 

The coefficient $\alpha_{0,n}$ depends on the sample size $n$ to accommodate for sparse network asymptotics. As in \citet{graham:2022}, we set $\alpha_{0,n} \equiv \log(\alpha_0/n)$, which jointly with the specification of the conditional distribution in  \eqref{eq:cprob} and Assumption \ref{ass:sampling} below, ensure that in the limit the network is sparse. Consequently, the average in-degree and out-degree of this network will be bounded. That is, for any $i\in [N]$ and $j \in [M]$, let $\rho_n \equiv \mathbb{E} \left[ Y_{ij} \right]$ denote the marginal probability of forming a link, then $\lim_{n\rightarrow \infty}M \rho_n = \lambda_0^a < \infty$ and $\lim_{n\rightarrow \infty} N \rho_n = \lambda_0^p < \infty$.

The linear index $Z_{ij}^\top\beta_{0}$ in  \eqref{eq:cprob} captures the contribution that sharing similar observed attributes has on a scientist $i$'s decision to collaborate in the project $j$. In other words, this component represents an assortative matching mechanism for the formation of a scientific collaboration network. Consequently, the coefficient $\beta_0$ is interpreted as a homophily parameter. 

Finally, notice that the latent component $v_{ij}$ affects the probability of establishing scientific collaboration. In this setting, the component $n^{-1/2} v$ represents a source of misspecification. The misspecification design is local as it is indexed by $n$ and vanishes away as the sample size increases at a $\sqrt{n}$-rate. In Section \ref{sec:size}, we show that at this rate, the local misspecification induces a bias term into the asymptotic distribution of interest. Moreover, the effect of misspecification is determined by the first moment of the distribution $\pi_{ij}(\cdot)$ of the latent factor $v_{ij}$.

\begin{remark}[Conditional probability]
Alternatively, we also consider the equivalent specification for the true conditional probability
\begin{equation}\label{eq:cprob2}
\mathbb{P}\left[ Y_{ij}=1 \mid X_i, W_j \right] 
=
\Lambda\left( \alpha_{0,n} + Z_{ij}^\top\beta_{0} + n^{-1/2} \eta_{ij} \right),
\end{equation} where $\Lambda(\cdot)$ and $Z_{ij}\in \Re^{d_z}$ are as defined above. The factor  $\eta_{ij} = \eta(X_i,Z_j)$ is a scalar random variable.  In fact, \eqref{eq:cprob} and \eqref{eq:cprob2} lead to the same local misspecification in our setting because in both cases, we have
\begin{align*}
    \sqrt{n}(\mathbb{P}\left[ Y_{ij}=1 \mid X_i, W_j \right]  - \Lambda\left( \alpha_{0,n} + Z_{ij}^\top\beta_{0} \right)) = \eta_{ij} + o_p(1). 
\end{align*}
Therefore, all results in this paper hold without any change if the local misspecification is formulated as \eqref{eq:cprob2}. 
\end{remark}

We now provide four empirical illustrations that can be analyzed using the local misspecification framework characterized by \eqref{eq:graphon} and \eqref{eq:cprob}.

\begin{example}[Latent homophily]
Suppose the bipartite network $\left[ Y_{ij} \right]_{i\in \left[ N \right], j\in \left[ M \right]}$ is generated according to 
\begin{equation}\label{eq:NFM_example1}
    Y_{ij} = \mathbf{1} \left[ \alpha_{0,n} + Z_{ij}^\top\beta_{0} + n^{-1/2} v_{ij} + U_{ij} \geq 0\right]
\end{equation} 
where $v_{ij} = \psi(v_i, v_j)$,  $v_i, v_j, \in \Re^{d_v}$  are unobserved fixed effects that are potentially correlated with $X_i$ and $W_j$, and $\psi(\cdot, \cdot)$ is a known function. Common choices of $\psi(\cdot, \cdot)$  include (i) nonparametric homophily,  $\psi(v_i, v_j) = -\left\| v_i - v_j \right\|_2$, and (ii) interactive fixed effects, $\psi(v_i, v_j) = v_i^\top v_j$.
Note that if $U_{ij}$ has a Logistic CDF, then the conditional probability of establishing a link is equal to \eqref{eq:cprob}.
In this setting, the local misspecification in \eqref{eq:cprob} is induced by the presence of latent homophily $v_{ij}$.
\end{example}
    
\begin{example}[Semiparametric distribution]
Consider the graphon $\left[ Y_{ij} \right]_{i\in \left[ N \right], j\in \left[ M \right]}$ is generated according to \eqref{eq:NFM_example1}, where $U_{ij}$ is an i.i.d. random term with Logistic CDF and $v_{ij}$ is a latent mixing component with unknown distribution $\pi_{ij}(\cdot)$ that has mean $\eta_{ij}$, which can depend on $(X_i, W_j)$. In this case, $n^{-1/2} v_{ij} + U_{ij}$ represents a composite error term with an unknown distribution that renders the estimation of $(\alpha_{0,n}, \beta_0)$ nonparametric. In this setting, the presence of local misspecification is due to specifying incorrectly the distribution of $U_{ij}$.
\end{example}

\begin{example}[Functional form misspecification]
Suppose the graphon $\left[ Y_{ij} \right]_{i\in \left[ N \right], j\in \left[ M \right]}$ is generated according to 
\begin{align*}
    Y_{ij} = \mathbf{1} \left[ U_{ij} \leq f(X_i, W_j)\right]
\end{align*}
where $U_{ij}$ is an i.i.d. random variable with Logistic CDF and $f(\cdot, \cdot)$ is an unknown measurable function of $(X_i, W_j)$. Assuming that $U_{i,j}$ has a Logistic distribution is without loss of generality if we allow the functional form of $f(\cdot,\cdot)$ to be flexible. The latent component $v_{ij}$ in \eqref{eq:cprob} can be interpreted as a local specification error from specifying a linear index model. That is, $f(X_i, W_j)=\alpha_{0,n} + Z_{ij}^\top\beta_{0} + n^{-1/2} v_{ij}$, where the distribution of $v_{ij}$ given by $\pi_{ij}$ depends on $(X_i, W_j)$.
\end{example}

\begin{example}[Network externalities]
Suppose that the graphon $\left[ Y_{ij} \right]_{i \in \left[ N \right], j\in \left[ M \right]}$ is an equilibrium outcome from a strategic network formation model with incomplete information as in \citet{leung:2015} or \citet{ridder/sheng:2020}. Under Assumption \ref{ass:sampling}, and regularity conditions in \citet{leung:2015}, there exists a symmetric  Bayesian Nash equilibrium (BNE) under which edges form independently conditionally on the observed attributes $\mathcal X_N$ and $\mathcal W_M$, and the conditional linking probabilities are $X$-$W$- exchangeable. In particular, the symmetric BNE yields the following network formation equation 
\begin{equation*} 
Y_{ij} = 
\mathbf{1}\left[ \alpha_{0,n} + Z_{ij}^\top\beta_{0} 
+  
\gamma_n  
\left( \E\left[ \mathcal{E}_{ij} \mid \mathcal X_N, \mathcal W_M, \tau^\ast \right] \right)
\geq U_{ij} \right]
\end{equation*}
where $\mathcal{E}_{ij} \equiv \mathcal{E}_{ij}\left( Y_n,  \mathcal X_N, \mathcal W_M \right)$,  $Y_{n} = \left[ Y_{ij} \right]_{i \in \left[ N \right], j\in \left[ M \right]} $, $\mathcal X_N = (X_1, \cdots, X_N)$, $\mathcal W_M = (W_1, \cdots, W_M)$, and $\tau^\ast$ denotes the equilibrium belief profile. Here,  $\mathcal{E}_{ij}$ represents an endogenous network statistic such as (i) in-degree of project $j$, $\frac{1}{N} \sum_{k \neq i} Y_{kj}$;  (ii) out-degree of scientist $i$, $\frac{1}{M} \sum_{k \neq j} Y_{ik}$; or (iii) average transitivity, $\frac{1}{NM} \sum_{k \neq i} \sum_{l \neq j}  Y_{kj} Y_{kl}$. In this setting, by assuming $\gamma_n$ scales with $n^{-1/2}$ (i.e., $\gamma_n = n^{-1/2} \gamma$), the misspecification component $v_{ij} \equiv \gamma \E\left[ \mathcal{E}_{ij} \mid \mathcal X_N, \mathcal W_M, \tau^\ast \right]$, and we assume it has a conditional density $\pi_{ij}(\cdot)$ given $(X_i,W_j)$. Then, we have
\begin{align*}
    \mathbb P\left( Y_{ij} = 1 \mid X_i, W_j\right) & = \mathbb E \left[\mathbb P\left( Y_{ij} = 1 \mid \mathcal X_N, \mathcal W_M\right) \mid X_i, W_j \right]  \\
    & = \mathbb E \left[\Lambda( \alpha_{0,n} + Z_{ij}^\top\beta_{0} +  n^{-1/2} v_{ij}  ) \mid X_i, W_j \right] \\
    & = \int \Lambda(\alpha_{0,n} + Z_{ij}^\top\beta_{0} +  n^{-1/2} v  ) \pi_{ij}(v)dv, 
\end{align*}
which is just \eqref{eq:cprob}. 
\end{example}

As parameter of interest, we focus on the scalar $\Psi_{0,n}$ that is defined as 
\begin{eqnarray}\label{eq:psi_0}
    \Psi_{0,n} = \mathbb{E} \int  \gamma_n\left( D_{ij}, n^{-1/2} v, \theta_{0,n} \right) \pi_{ij}(v) dv, 
\end{eqnarray}
where $\gamma_n$ is a known measurable function that maps $\left( D_{ij}, n^{-1/2} v, \theta_{0,n} \right) \mapsto \Re$, with $D_{ij} \equiv (X_i^\top,W_j^\top)^\top$ and $\theta_{0,n} \equiv (\alpha_{0,n}, \beta_0)$. Notice that the misspecification term $n^{-1/2} v$ can affect $\Psi_{0,n}$ directly, and thus, $\Psi_{0,n}$ is allowed to be misspecified itself (cf. \citealt{armstrong/kolesar:2021}). Next, we provide some examples of the parameter of interest as described by \eqref{eq:psi_0}.
 
\begin{example}[Individual Parameters]\label{eg:4} 
If the research interest is to make inference on the homophily parameters $\beta_{0}$ or the intercept $\alpha_{0,n}$, the parameter of interest can be defined as $\Psi_{0,n} = \theta_{0,n}^{(k)}$, where $\theta_{0,n}^{(k)}$ denotes the $k$th element in $\theta_{0,n} \in \Re^{d_z + 1}$.
\end{example}

\begin{example}[Network Statistics]     
The parameter of interest $\Psi_{0,n}$ can be defined to capture global features of the network, such as the average in-degree of the network. In this case, 
\begin{equation*}
\Psi_{0,n} 
=
N \mathbb{E}
\int 
\left[  
\Lambda\left(\alpha_{0,n} + Z_{ij}^\top\beta_{0} + n^{-1/2}v \right)
\right] \pi_{ij}(v) dv,
\end{equation*}
which represents the average number of scientists that participate in a randomly selected project. Similarly, the parameter of interest could be the average out-degree of the network, i.e,
\begin{equation*}
\Psi_{0,n} 
=
M \mathbb{E}
\int 
\left[  
    \Lambda\left(\alpha_{0,n} + Z_{ij}^\top\beta_{0} + n^{-1/2}v \right) 
\right] \pi_{ij}(v) dv,
\end{equation*}
which captures the average number of projects in which a randomly selected scientist participates.
\end{example}

\begin{example}[Average Effects]\label{eg:6}
The parameter of interest $\Psi_{0,n}$ could capture the average marginal effect. For example, if $Z_{ij}^{(1)}$ is assumed to be continuous, then $\Psi_{0,n}$ can be defined as 
\begin{eqnarray*}
\Psi_{0,n} 
&=& 
n \mathbb{E} 
\int 
\frac{\partial}{\partial Z_{ij}^{(1)}} 
\Lambda\left( \alpha_{0,n} + Z_{ij}^\top\beta_{0} + n^{-1/2}v  \right)
\pi_{ij}(v) dv.
\end{eqnarray*}
Alternatively, if the object of interest is to do inference on an average partial effect, then $\Psi_{0,n}$ can be defined as  
\begin{eqnarray*}
\Psi_{0,n} = 
n
\left. \mathbb{E}
\int 
\frac{\partial}{\partial \beta} \Lambda\left(\alpha_{0,n} + Z_{ij}^\top\beta_{0} + n^{-1/2}v \right) \right|_{\beta = \beta_{0}}
\pi_{ij}(v) dv. 
\end{eqnarray*}
\end{example}

\subsection{Minimum-MSE Estimator}\label{sec:minimum-MSE}
The robust estimator $\hat \Psi(\hat \kappa)$ of $\Psi_{0,n}$ is defined as the following one-step regular estimator 
\begin{align}
\hat \Psi(\hat \kappa) 
&= 
\frac{1}{NM}\sum_{i=1}^N \sum_{j=1}^M \gamma_n(D_{ij}, 0, \hat \theta_{\text{initial}}) 
+ 
\hat \kappa^\top  \left[\frac{n}{NM}\sum_{i=1}^N \sum_{j=1}^M (Y_{ij} - \Lambda(R_{ij}^\top \hat \theta_{\text{initial}}))H_{ij}\right],
\label{eq:robust_estimator}
\end{align}
where $R_{ij} \equiv (1, Z_{ij}^\top)^\top \in \Re^{d_z +1}$, $\hat \theta_{\text{initial}}$ is a $\sqrt{n}$-consistent initial estimator of $\theta_{0,n}$, i.e., $\hat \theta_{\text{initial}} - \theta_{0,n} = O_P(n^{-1/2})$, $H_{ij}$ is a $k_n$-dimensional vector of sieve basis functions of $X_i$ and $W_j$, and $\hat \kappa$ is a $k_n$-dimensional estimator of the parameter $\kappa$. The dimension $k_n$ can be either fixed or increasing with $n$, given that its growth rate satisfies Assumption \ref{ass:kn} below. We note that restricting $H_{ij}$ to be a function of $(X_i, W_j)$ only is necessary because we focus on the conditional (and locally misspecified) moment\footnote{Moreover, it follows from the random sampling of nodes in  Assumption \ref{ass:sampling} that the conditional moment satisfies $ \mathbb E(Y_{ij} - \Lambda(R_{ij}^\top \theta) \mid \mathcal X_N, \mathcal W_M) = \mathbb E(Y_{ij} - \Lambda(R_{ij}^\top \theta) \mid  X_i, W_j)$. Hence, it is sufficient to define $H_{ij}$ as a function of $(X_i, W_j)$.} 
\begin{align}\label{eq:conditional_moment}
    \mathbb E(Y_{ij} - \Lambda(R_{ij}^\top \theta) \mid  X_i,  W_j)= 0.
\end{align}

The robust estimator $\hat \Psi(\hat \kappa)$ is composed of two components. The first term in the right-hand side of  \eqref{eq:robust_estimator} represents a plug-in estimator of $\Psi_{0,n}$, which ignores the effect of the local misspecification. The plug-in estimator is evaluated at the initial estimator $\hat \theta_{\text{initial}}$. As initial estimator, we consider the logistic maximum likelihood estimator (MLE), i.e., 
\begin{align}\label{eq:theta_initial}
    \hat \theta_{\text{initial}} = (\hat \alpha_n, \hat{\beta}_n) \equiv \arg \max_{\theta} L_n(\theta),
\end{align}
where 
\begin{align*}
    L_n(\theta)     =& \frac{1}{NM} \sum_{i=1}^N\sum_{j=1}^M \ell_{ij}(\theta), \quad \ell_{ij}(\theta)  = Y_{ij}R_{ij}^\top \theta - \log(1+\exp(R_{ij}^\top \theta)). 
\end{align*}
Notice that under local misspecification, the asymptotic distribution of the plug-in estimator will be centered at $\mathbb{E} \gamma_n\left( D_{ij}, 0, \theta_{0,n} \right)$ rather than at $\Psi_{0,n}$ as it assumes away the presence of the local misspecification.

The second component in  \eqref{eq:robust_estimator} corresponds to a one-step adjustment that accounts for the asymptotic bias generated by the local misspecification. The term in brackets is the sample analogue of the unconditional moments $\mathbb{E}\left[ \left(Y_{ij}- \Lambda\left( R_{ij}^\top\theta \right)\right) H_{ij} \right]$, which are scaled by $n$ to account for sparse network asymptotics (cf. \citealt{graham:2022}). Notice that the following unconditional moment restrictions will fail to hold exactly
\begin{eqnarray}\label{eq:unconditional_moment}
   \mathbb{E}\left[ \left(Y_{ij}- \Lambda\left( R_{ij}^\top\theta \right)\right) H_{ij} \right]=0,
\end{eqnarray}
due to the effect of the local misspecification in the conditional probability given by \eqref{eq:cprob}.
The vector $\hat \kappa$ is an estimate of $\kappa$, which captures the response of the estimator $\hat \Psi(\hat \kappa)$ to the misspecification of the $k_n \times 1$ moment conditions in \eqref{eq:unconditional_moment}. Thus, following \citet{andrews/gentzkov/shapiro:2020,armstrong/kolesar:2021,bonhomme/weidner:2022}, we refer to $\kappa$ as a sensitivity parameter. Below, we discuss how to optimally select $\kappa$ to minimize the worst-mean squared error (MSE).

In contrast to the existing literature on sensitivity analysis that relies on maximum likelihood methods or method of moments, the dimension of the sensitivity parameter $\kappa$ grows with the dimension of the sieve space. This is an implication that the underlying moment restrictions in \eqref{eq:unconditional_moment} are constructed based on the locally misspecified conditional moments in \eqref{eq:conditional_moment}. The growing sieve space will determine the limiting distribution of the one-step robust estimator. 

Next, we describe the limiting distribution of the robust estimator, which will be central to performing robust inference on $\Psi_{0,n}$. Following \cite{andrews/gentzkov/shapiro:2017, armstrong/kolesar:2021,bonhomme/weidner:2022}, we focus on a class of regular estimators. As in \citet{newey:1990}, a necessary and sufficient condition for the one-step robust estimator to be regular is 
\begin{align}\label{eq:regularity}
    G^\top_0 \kappa = \Gamma,
\end{align}
where $\Gamma \equiv \Gamma(0)$ with $\Gamma(t)$ being defined as $\sup_{t \in \mathcal{N}(0)}||\partial_t \mathbb{E} \gamma_n(D_{ij},0,\theta_{0,n}+t) - \Gamma(t)||_2 \rightarrow 0$ and $G_0 \equiv \mathbb{E}\alpha_0 \exp(Z_{ij}^\top \beta_0)H_{ij}R_{ij}^\top$. Under this requirement and provided that $\hat \kappa$ is a consistent estimator of $\kappa$, i.e., $||\hat \kappa - \kappa||_2 = o_P(1)$ and the initial estimator is $\sqrt{n}$-consistent, i.e., $\hat \theta_{\text{initial}} - \theta_{0,n} = O_P(n^{-1/2})$, we have the next strong approximation of the distribution of the robust estimator $\hat \Psi(\hat \kappa)$ which does not depend on the asymptotic distribution of the initial estimator: 
\begin{align}\label{eq:psi_asyexpansion}
\sqrt{n}(\hat \Psi(\hat \kappa) - \Psi_{0,n}) - \omega_n(\kappa) = o_P(1),
\end{align}
where 
\begin{align}\label{eq:psi_limitdist}
\omega_n(\kappa) 
\stackrel{d}{=} 
\N\left( 
\mathbb E\left[ \eta_{ij} (\kappa^{\top} \alpha_0 \exp(Z_{ij}^\top \beta_{0})H_{ij} -\Delta_{v,ij})\right],  
\Sigma_{\gamma} + 2\Sigma_{H,\gamma}^\top \kappa + \kappa^\top \Sigma_H \kappa 
\right),
\end{align}
$\Delta_{v,ij} \equiv \partial_{v} \gamma_n(D_{ij},v,\theta_{0,n})|_{v=0}$ and $(\Sigma_H,\Sigma_{H,\gamma},\Sigma_{\gamma})$ are three submatrices of a $(k_n+1) \times (k_n+1)$ covariance matrix
\begin{align*}
    \Sigma_n = \begin{pmatrix}
        \Sigma_H & \Sigma_{H,\gamma} \\
        \Sigma_{H,\gamma}^\top & \Sigma_{\gamma}
    \end{pmatrix}.   
\end{align*}
The covariance matrix $\Sigma_n$ is defined and characterized in Section \ref{sec:VB_est}. Here, $\omega_n(\kappa)$ is referred to as the strong approximation of $\sqrt{n}(\hat \Psi(\hat \kappa) - \Psi_{0,n})$ rather than the limit under weak convergence because the dimensions of $\kappa$, $\Sigma_H$, $\Sigma_{H,\gamma}$, and thus themselves, implicitly depend on the sample size $n$. 

The strong approximation result in \eqref{eq:psi_asyexpansion} and \eqref{eq:psi_limitdist} provides two main insights. First, the one-step robust estimator $\hat \Psi(\hat \kappa)$ is asymptotically biased. This asymptotic bias is multifaceted: the term $-\E\left[\eta_{ij}\Delta_{v,ij} \right]$ reflects the bias induced by the plug-in estimator in \eqref{eq:robust_estimator} which is the result of having a limiting distribution centered around $\mathbb{E} \gamma_n\left( D_{ij}, 0, \theta_{0,n} \right)$ rather than at $\Psi_{0,n}$. While the term $\E\left[ \eta_{ij} \kappa^{\top} \alpha_0 \exp(Z_{ij}^\top \beta_{0})H_{ij} \right]$ stems from the adjustment term in  \eqref{eq:robust_estimator} and represents the asymptotic bias due to using $e(R_{ij}^\top \theta)$ as the identification condition in \eqref{eq:unconditional_moment} as opposed to $\mathbb{P}\left[ Y_{ij} =1 \mid X_i, W_j \right]$. Second, the strong approximation result underscores that the local misspecification does not have a higher-order effect, as its complete impact concentrates on the asymptotic bias. In other words, the asymptotic variance does not depend on the local misspecification parameter $\eta_{ij}$.

Importantly, notice that the direction of the asymptotic bias is determined by $\eta_{ij}$, which corresponds to the mean of the latent distribution. Because this direction is unidentified, it renders the asymptotic distribution in \eqref{eq:psi_asyexpansion} not immediately useful to conduct inference on $\Psi_{0,n}$. Nevertheless, the magnitude of the asymptotic bias can be bounded if an upper bound on $\E (\eta_{ij}^2)$ is imposed. Let $\overline{\sigma}^2$ denote this upper bound, which is specified by the researcher. It follows then that 
\begin{align*}
    \E\left[ \eta_{ij} (\kappa^{\top} \alpha_0 \exp(Z_{ij}^\top \beta_{0})H_{ij} -\Delta_{v,ij})\right]    
    \leq 
    \overline \sigma (\E (\kappa^{\top} \alpha_0 \exp(Z_{ij}^\top \beta_{0})H_{ij} -\Delta_{v,ij})^2)^{1/2}.
\end{align*}
We can use this inequality to define the worst-case bias within the neighborhood of misspecification characterized by $\E (\eta_{ij}^2) \leq \overline{\sigma}^2$ as 
\begin{align*}
    \mathcal{B}(\kappa) 
    &\equiv  
    \overline \sigma (\mathbb E (\kappa^{\top} \alpha_0 \exp(Z_{ij}^\top \beta_{0})H_{ij} -\Delta_{v,ij})^2)^{1/2} \\
    & = \left\{ \overline{\sigma}^2 \left[\kappa^\top B_H \kappa - 2B_{H,\gamma}^\top \kappa + B_\gamma \right] \right\}^{1/2},
\end{align*}   
where 
\begin{align*}
    B_H = \alpha_0^2 \mathbb E \exp(2Z_{ij}^\top \beta_0) H_{ij} H_{ij}^\top, \quad B_{H,\gamma} = \alpha_0 \mathbb E \exp(Z_{ij}^\top \beta_0) H_{ij} \Delta_{v,ij}, \quad \text{and} \quad B_{\gamma} = \mathbb E \Delta_{v,ij}^2.
\end{align*}

The optimal choice of the sensitivity parameter $\kappa$ is set to minimize the worst-case MSE defined as 
\begin{align}\label{eq:MSE_H}
    MSE_H(\kappa) = \mathcal B (\kappa)^2 + ( \Sigma_{\gamma} + 2 \Sigma_{H,\gamma}^\top \kappa +  \kappa^\top  \Sigma_H \kappa)
\end{align}
subject to the regularity condition \eqref{eq:regularity}. It follows from the method of Lagrange multipliers that the optimal $\kappa$, denoted by $\kappa^*$, is defined as the solution to this convex optimization 
\begin{align*}
    \kappa^* = \argmin_{\kappa} \left\{ MSE(\kappa) + \lambda^\top (G^\top_0 \kappa - \Gamma)\right\},
\end{align*}
where $\lambda$ denotes the Lagrange multiplier associated to regularity condition \eqref{eq:regularity}. Moreover,  $\kappa^\ast$ can be characterized in closed-form, which is given by 
\begin{align}\label{eq:kappa_star}
    \kappa^* = -\Phi_1^{-1} \left\{\Phi_2 - G_0 \left[G^\top_0 \Phi_1^{-1} G_0 \right]^{-1} \left[ G^\top_0 \Phi_1^{-1}\Phi_2+\Gamma \right]  \right\},
\end{align}
where
\begin{align*}
    \Phi_1 =  \Sigma_H + \overline \sigma^2 B_H  \quad \text{and} \quad \Phi_2 = \Sigma_{H,\gamma} - \overline \sigma^2 B_{H,\gamma}.
\end{align*}

In Section \ref{sec:VB_est}, we provide consistent estimators $\hat G$, $\hat \Gamma$, $\hat B_H$, $\hat B_{H,\gamma}$, $\hat \Sigma_H$,  $\hat \Sigma_{H,\gamma}$, $\hat \Sigma_{\gamma}$ for their population counterparts. Then, the optimal sensitivity parameter $\kappa^*$ is estimated by 
\begin{align*}
    \hat \kappa^* =  -\hat \Phi_1^{-1} \left\{\hat \Phi_2 - \hat G \left[\hat G^\top \hat \Phi_1^{-1} \hat G \right]^{-1} \left[ \hat G^\top \hat \Phi_1^{-1}\hat \Phi_2+ \hat \Gamma \right]  \right\},
\end{align*}
where 
\begin{align*}
\hat  \Phi_1 =  \hat \Sigma_H  + \overline \sigma^2 \hat B_H \quad \text{and} 
\quad 
\hat \Phi_2 =  \hat \Sigma_{H,\gamma} - \overline \sigma^2 \hat B_{H,\gamma}.
\end{align*}

The approach of selecting $\kappa^\ast$ to minimize the worst-case MSE is similar in spirit to that in \citet{bonhomme/weidner:2022} and has a number of antecedents in robust statistics (\citealt{huber/ronchetti:2009, rieder:1994}). It is important to emphasize that a unique feature of our setting is that the dimension $\kappa^\ast$ is determined by the sieve basis and can increase with the sample size. Consequently, our methodology to select and consistently estimate $\kappa^\ast$ complements the existing approaches in the literature.

\subsection{Robust Confidence Interval}\label{sec:robust_CI}

The bias-aware confidence interval for $\Psi_{0,n}$ is designed to account for the effect of the local misspecification. This is obtained by enlarging the $1-\alpha$ quantile value by the worst-case bias in the strong approximation of $\sqrt{n}(\hat \Psi(\hat \kappa) - \Psi_{0,n})$ (cf. \citealt{armstrong/kolesar:2021}). In particular, we define the feasible optimal bias-aware confidence interval for $\Psi_{0,n}$ as 
\begin{align}\label{eq:robust_CI_feasible}
\widehat{CI}_{1-\alpha}(\hat \kappa^*) \equiv \hat \Psi(\hat \kappa^*)  
    \pm 
CV_\alpha \left(    \frac{ \hat {\mathcal{B}}(\hat \kappa^*)}{(\hat \Sigma_{\gamma} + 2 \hat \Sigma_{H,\gamma}^\top \hat \kappa^\ast + \hat \kappa^{*,\top} \hat \Sigma_H \hat \kappa^*)^{1/2}}  \right) 
        \frac{(\hat \Sigma_{\gamma} + 2 \hat \Sigma_{H,\gamma}^\top \hat \kappa^\ast + \hat \kappa^{*,\top} \hat \Sigma_H \hat \kappa^*)^{1/2} }{\sqrt{n}},
\end{align}
where feasible worst-case bias $\hat {\mathcal{B}}(\hat \kappa^*)$ is defined as 
\begin{align*}
\hat {\mathcal{B}}(\hat \kappa^*) =   \left\{ \overline{\sigma}^2 \left[\hat \kappa^{*,\top} \hat B_H \hat \kappa^* - 2\hat B_{H,\gamma}^\top \hat \kappa^* + \hat B_\gamma \right] \right\}^{1/2}
\end{align*}
and $CV_\alpha(t)$ denotes the $1-\alpha$ quantile of $\mathcal{Z}$ where $\mathcal{Z}\sim\N\left( t, 1 \right)$. 

\begin{remark}[Neighborhood of misspecification]
The choice of  $\overline{\sigma}^2$ cannot be determined in a data-driven way. However, this upper bound can be varied to conduct sensitivity analysis on the magnitude of the local misspecification (cf. \citealt{conley/hansen/rossi:2012}). The values of $\overline{\sigma}^2$ should reflect the potential forms of misspecification that are admissible in a given application and thus should be established on an application-specific basis by the researcher. In likelihood-based models, \citet{bonhomme/weidner:2022} offer an interpretation of  $\overline{\sigma}^2$ as a lower bound on the local power of a likelihood ratio specification test. In a setting of overidentified unconditional moments, \citealt{armstrong/kolesar:2021} suggest using a $J$-test of overidentifying restrictions to obtain a lower bound on $\overline\sigma^2$. For additional discussions on the interpretation of the neighborhood of misspecification, we refer the reader to \citet{andrews/gentzkov/shapiro:2017,noack/rothe:2019, armstrong/weidner/zeleneev:2022}. 
\end{remark}

\begin{remark}[Other performance criteria]
The researcher might be interested in a robust estimate for $\Psi_{0,n}$ that minimizes the length of the bias-aware CIs as in \citet{armstrong/kolesar:2021}. The optimal sensitivity parameter $\kappa^\ast$ can be set to minimize the asymptotic variance in the strong approximation \eqref{eq:psi_limitdist} subject to an upper bound on the worst-case bias and condition \eqref{eq:regularity}. In extended simulation experiments, we verify that the resulting bias-aware CIs also have the correct asymptotic coverage. 
\end{remark}

\subsection{Variance and Bias Components}\label{sec:VB_est}
In this section, we propose consistent estimators for $G$, $ \Gamma$, $B_H$, $B_{H,\gamma}$, $\Sigma_H$,  $\Sigma_{H,\gamma}$, and $ \Sigma_{\gamma}$. Recall that the initial estimator $\hat \theta_{\text{initial}} = (\hat \alpha_n, \hat{\beta}_n)$ is defined in \eqref{eq:theta_initial}. Then, we define the estimators 
\begin{align*}
   & \hat G =  \frac{1}{NM}\sum_{i=1}^N \sum_{j=1}^M n \exp(R_{ij}^\top \hat \theta_{\text{initial}})H_{ij}R_{ij}^\top,  \\
   & \hat \Gamma = \frac{1}{NM}\sum_{i=1}^N \sum_{j=1}^M \partial_\theta \gamma_n(D_{ij},0,\hat \theta_{\text{initial}}), \\
   & \hat B_H = \frac{1}{NM}\sum_{i=1}^N \sum_{j=1}^M n^2 \exp(2R_{ij}^\top \hat \theta_{\text{initial}}) H_{ij} H_{ij}^\top, \\ 
   & \hat B_{H,\gamma} = \frac{1}{NM}\sum_{i=1}^N \sum_{j=1}^M n^2 \exp(2R_{ij}^\top  \hat \theta_{\text{initial}}) H_{ij}\hat \Delta_{v,ij}, \quad \text{and} \quad \hat B_{\gamma} = \frac{1}{NM}\sum_{i=1}^N \sum_{j=1}^M  \hat \Delta_{v,ij}^2,
\end{align*}
where 
\begin{align*}
 \hat   \Delta_{v,ij} = \partial_{v} \gamma_n(D_{ij},v,\hat \theta_{\text{initial}})|_{v=0}.
\end{align*}

To estimate the components of the covariance matrix in \eqref{eq:psi_limitdist}, recall that 
\begin{align*}
    \Sigma_n = \begin{pmatrix}
        \Sigma_H & \Sigma_{H,\gamma} \\
        \Sigma_{H,\gamma}^\top & \Sigma_{\gamma}
    \end{pmatrix},
\end{align*}
where $\Sigma_n$ is the covariance matrix of the joint scores
\begin{align*}
    s_{ij} = & \begin{pmatrix}
        & n\left(Y_{ij} - \int \Lambda(R_{ij}^\top \theta_{0,n}+n^{-1/2}v)\pi_{ij}(v)dv\right)H_{ij} \\
        & \gamma_n(D_{ij},0,\theta_{0,n}) -  \mathbb{E}\gamma_n(D_{ij},0, \theta_{0,n}) 
    \end{pmatrix} \equiv (s_{Y,ij} H_{ij}^\top,s_{\gamma,ij})^\top.
\end{align*}
Let
$\overline s_{ij} = \mathbb{E}(s_{ij}|X_i,A_i,W_j,B_j)$,  $\overline s_{1i}^a = \mathbb{E}(\overline s_{ij}|X_i,A_i)$, $\overline s_{1j}^p = \mathbb{E}(\overline s_{ij}|W_j,B_j)$, $\Sigma_{n}^a = \mathbb{E} \overline s_{1i}^a(\overline s_{1i}^a)^\top$, $\Sigma_{n}^p = \mathbb{E} \overline s_{1j}^p(\overline s_{1j}^p)^\top$, and $\Sigma_{2n} = n^{-1}\mathbb{E}(s_{ij} - \overline s_{ij})(s_{ij} - \overline s_{ij})^\top$. Because $D_{ij} = (X_i^\top,W_j^\top)$, we have 
\begin{align*}
\overline{s}_{\gamma,ij} \equiv \mathbb E (s_{\gamma,ij}|X_i,A_i,W_j,B_j) = s_{\gamma,ij}.     
\end{align*}
Then, the covariance matrix $\Sigma_n$ is given by 
\begin{align}\label{eq:Sigma_n}
\Sigma_n \equiv    \begin{pmatrix}
        \Sigma_H & \Sigma_{H,\gamma} \\
        \Sigma_{H,\gamma}^\top & \Sigma_{\gamma}
    \end{pmatrix} = \frac{\Sigma^a_{n}}{1-\phi} + \frac{\Sigma^p_{n}}{\phi} + \frac{ \Sigma_{2n}}{\phi(1-\phi)}
\end{align}
where $\phi_n \equiv M/n \rightarrow \phi \in (0,1)$, 
\begin{align*}
\Sigma_{n}^a 
= 
\mathbb E 
\begin{pmatrix}
[\mathbb E (s_{Y,ij}H_{ij}|X_i,A_i)][\mathbb E (s_{Y,ij}H_{ij}|X_i,A_i)]^\top	
&	
\mathbb E (s_{Y,ij}H_{ij}|X_i,A_i)\mathbb E (s_{\gamma,ij}|X_i,A_i)  \\
\mathbb E (s_{Y,ij}H_{ij}^\top |X_i,A_i)\mathbb E (s_{\gamma,ij}|X_i,A_i)  
& [\mathbb E (s_{\gamma,ij}|X_i,A_i)]^2
\end{pmatrix},
\end{align*}
\begin{align*}
\Sigma_{n}^p = \mathbb E \begin{pmatrix}
[\mathbb E (s_{Y,ij}H_{ij}|W_j,B_j)][\mathbb E (s_{Y,ij}H_{ij}|W_j,B_j)]^\top	&	\mathbb E (s_{Y,ij}H_{ij}|W_j,B_j)\mathbb E (s_{\gamma,ij}|W_j,B_j)  \\
\mathbb E (s_{Y,ij}H_{ij}^\top |W_j,B_j)\mathbb E (s_{\gamma,ij}|W_j,B_j)  & [\mathbb E (s_{\gamma,ij}|W_j,B_j)]^2
\end{pmatrix},
\end{align*}
and 
\begin{align*}
\Sigma_{2n} = \mathbb E \begin{pmatrix}
   n^{-1} Var(s_{Y,ij}|X_i,A_i,W_j,B_j) H_{ij}H_{ij}^\top & 0\\
    0 & 		0
\end{pmatrix}.
\end{align*}

Let $\hat s_{Y,ij} = n(Y_{ij} - \exp(R_{ij}^\top \hat \theta_{\text{initial}}))$ and $\hat s_{\gamma,ij} = \gamma_n(D_{ij},0,\hat \theta_{\text{initial}}) - \frac{1}{NM}\sum_{i=1}^N \sum_{j=1}^M\gamma_n(D_{ij},0,\hat \theta_{\text{initial}})$. Following \cite{graham:2022}, we define 
\begin{align*}
    \hat \Sigma_n \equiv   \begin{pmatrix}
      \hat  \Sigma_H & \hat \Sigma_{H,\gamma} \\
      \hat  \Sigma_{H,\gamma}^\top & \hat \Sigma_{\gamma}
    \end{pmatrix} 
    = \frac{\hat \Sigma^a_{n}}{1-\phi_n} + \frac{\hat \Sigma^p_{n}}{\phi_n} + \frac{\hat \Sigma_{2n}}{\phi_n(1-\phi_n)},
\end{align*}
where 
\begin{align*}
   \hat \Sigma^a_{n} =   \frac{1}{N}\sum_{i=1}^N
\begin{pmatrix}
\frac{2}{M(M-1)}\sum_{j=1}^{M-1} \sum_{k=j+1}^M(\hat s_{Y,ij}\hat s_{Y,ik}H_{ij} H_{ik}^\top)
&	
\frac{2}{M(M-1)}\sum_{j=1}^{M-1} \sum_{k=j+1}^M(\hat s_{Y,ij}\hat s_{\gamma,ik}H_{ij}) \\
\frac{2}{M(M-1)}\sum_{j=1}^{M-1} \sum_{k=j+1}^M(\hat s_{Y,ij}\hat s_{\gamma,ik}H_{ij}^\top)
& \frac{2}{M(M-1)}\sum_{j=1}^{M-1} \sum_{k=j+1}^M \hat s_{\gamma,ij}\hat s_{\gamma,ik}
\end{pmatrix},
\end{align*}
\begin{align*}
   \hat \Sigma^p_{n} =  \frac{1}{M}\sum_{j=1}^M
\begin{pmatrix}
\frac{2}{N(N-1)}\sum_{i=1}^{N-1} \sum_{k=i+1}^N(\hat s_{Y,ij}\hat s_{Y,ik}H_{ij} H_{ik}^\top)
&	
\frac{2}{N(N-1)}\sum_{i=1}^{N-1} \sum_{k=i+1}^N(\hat s_{Y,ij}\hat s_{\gamma,ik}H_{ij}) \\
\frac{2}{N(N-1)}\sum_{i=1}^{N-1} \sum_{k=i+1}^N(\hat s_{Y,ij}\hat s_{\gamma,ik}H_{ij}^\top)
& \frac{2}{N(N-1)}\sum_{i=1}^{N-1} \sum_{k=i+1}^N \hat s_{\gamma,ij}\hat s_{\gamma,ik}
\end{pmatrix},
\end{align*}
\begin{align*}
\hat    \Sigma_{2n} =  \frac{1}{n}\begin{pmatrix}
   \frac{1}{NM}\sum_{i=1}^N \sum_{j=1}^M \hat s_{Y,ij}^2 H_{ij}H_{ij}^\top -  \hat \Sigma^a_{YY,n} - \hat \Sigma^p_{YY,n} & 0\\
    0 & 		0
\end{pmatrix},
\end{align*}
and $\hat \Sigma^a_{YY,n}$ and $\hat \Sigma^p_{YY,n}$ are the upper-left $k_n \times k_n$ submatrices of $\hat \Sigma^a_{n}$ and $\hat \Sigma^p_{n}$, respectively. In Lemma \ref{lem:bias_variance_estimator} of the Appendix, we show that the estimators $\hat G$, $\hat \Gamma$, $\hat B_H$, $\hat B_{H,\gamma}$, $\hat \Sigma_H$,  $\hat \Sigma_{H,\gamma}$, and $\hat \Sigma_{\gamma}$ are consistent.

\section{Asymptotic Size Control}\label{sec:size}
In this section, we show the bias-aware CI proposed in \eqref{eq:robust_CI_feasible} has the correct asymptotic size control, which relies on the following assumptions.  
\begin{assumption}
\begin{enumerate}
    \item The graphon $\{Y_{ij}\}_{i \in [N],j \in [M]}$ is generated according to  \eqref{eq:graphon} with $\left\{ X_i, A_i \right\}_{i=1}^{N}$,  $\left\{ W_j, B_j \right\}_{j=1}^{M}$, $\left\{ U_{ij} \right\}_{i\in \left[ N \right], j\in \left[ M \right]}$, $\left\{ v_{ij} \right\}_{i\in \left[ N \right], j\in \left[ M \right]}$ be sequences of i.i.d. random variables that are independent, $M/n \rightarrow \phi \in (0,1)$ and $n = M+N$. 
    \item The conditional  distribution $\mathbb{P}(Y_{ij} = 1|W_i,X_j)$ is as defined in  \eqref{eq:cprob} with $\eta_{ij} = \int v \pi_{ij}(v)dv$. In addition, suppose $\mathbb E \eta_{ij}^2 \leq \overline \sigma^2$ for some known upper bound $\overline \sigma^2$ and $\max_{i\in [N],j\in [M]}\int |v|\exp(v) \pi_{ij}dv \leq C$ for some constant $C \in (0,\infty)$. 
    \item Let $\theta_0 = (\alpha_0,\beta_0^\top)^\top \in \mathbb{A} \times \mathbb{B} = \Theta$, the parameter spaces $\mathbb{A}$ and $\mathbb{B}$ are compact. 
    \item Suppose that $Z_{ij}$ is bounded, and 
    \begin{align*}
    \Lambda(R_{ij}^\top \theta_{0,n}) \in \left[\frac{c_1}{n},\frac{c_2}{n}\right]
    \end{align*}
    for some constants $0<c_1 \leq c_2<\infty$, where $R_{ij} = (1,Z_{ij}^\top)^\top$. 
    \item There exist constants $0 < \lambda_1 < \lambda_2 < \infty$ such that, with probability approaching one, 
    \begin{align*}
    \lambda_1 \leq \lambda_{\min}\left( \frac{1}{NM}\sum_{i=1}^N \sum_{j=1}^M R_{ij}R_{ij}^\top\right) \leq \lambda_{\max}\left( \frac{1}{NM}\sum_{i=1}^N \sum_{j=1}^M R_{ij}R_{ij}^\top\right) \leq \lambda_2,
    \end{align*}
    where for a symmetric matrix $A$, $\lambda_{\min}(A)$ and $\lambda_{\max}(A)$ denote the minimum and maximum singular values of $A$, respectively. 
    \item Let $G(\alpha,\beta) = \mathbb{E}\alpha \exp(Z_{ij}^\top \beta)H_{ij}R_{ij}^\top$ and $G_0 = G(\alpha_0,\beta_0)$. Then, 
    \begin{align*}
    \sup_{(\alpha,\beta^\top)^\top \in \Theta}
    \frac{\lambda_{max}(G(\alpha,\beta) - G_0)}{|\alpha - \alpha_0| + ||\beta - \beta_0||_2} \leq C    
    \end{align*}
     for a positive constant $C<\infty.$ 
\end{enumerate}\label{ass:sampling}
\end{assumption}
Condition 1 of Assumption \ref{ass:sampling} describes the graph formation and ensures that $\{Y_{ij}\}_{i \in [N],j \in [M]}$ is $\text{X-W}$ exchangeable. Condition 2 restricts the conditional probability to satisfy  \eqref{eq:cprob}.
Condition 3 is a regularity condition on the parameter space $\Theta$. Condition 4 bounds the largest magnitude of the covariates and the rate at which the probability of forming a link decays. This condition, jointly with 3, controls the rate of sparsity in the network. The boundedness of $Z_{ij}$ is for theoretical simplicity and can be relaxed a sub-exponential tail of $Z_{ij}$, which is necessary for $G_0$ in Condition 6 to be well-defined. Condition 5 is an identification condition for $\theta_{0,n}$. Finally, Condition 6 is a continuity condition of $G(\alpha, \beta)$ at $(\alpha_0, \beta_0)$, which bounds the maximum eigenvalue difference. Conditions 1-5 of Assumption \ref{ass:sampling} has been used in \citet{graham:2022}. 
Condition 6 is specific to our setting. 

The next assumption provides the smoothness conditions required for the asymptotic expansions. 
\begin{assumption}
\begin{enumerate}
\item Recall $\theta_{0,n} \equiv (\alpha_{0,n}, \beta_0)$. Suppose $\gamma_n(D_{ij},0,\theta_{0,n}+t)$ is twice-differentiable in $t$, 
\begin{align*}
\mathbb E \left\|\partial_\theta \gamma_n(D_{ij},0,\theta_{0,n}) \right\|_2^2 = O(1),    
\end{align*}
\begin{align*}
\sup_{t \in \mathcal{N}(0)}||\partial_t \mathbb{E} \gamma_n(D_{ij},0,\theta_{0,n}+t) - \Gamma(t)||_2 \rightarrow 0,
\end{align*}
and
\begin{align*}
    \sup_{t \in \mathcal{N}(0)}\left\| \frac{1}{NM}\sum_{i=1}^N \sum_{j=1}^M \partial_{t,t^\top} \gamma_n(D_{ij},0,\theta_{0,n}+t)\right\|_{op} = O_P(1), 
\end{align*}
where $\mathcal{N}(0) \in \Re^{d_z+1}$ is an arbitrary neighborhood of $0$. Also, let $\Gamma(0)$ be denoted as $\Gamma$. 
\item Let $\Delta_{v,ij} = \partial_v\gamma_n(D_{ij},v,\theta_{0,n})|_{v=0}$. Then, we have  
\begin{align}\label{eq:quadratic_expansion}
    \sqrt{n}\mathbb{E} 
\left( 
    \int \gamma_n(D_{ij},n^{-1/2}v,\theta_{0,n})\pi_{ij}(v)dv 
    -
    \gamma_n(D_{ij},0,\theta_{0,n})    
\right) =  \mathbb E \eta_{ij}\Delta_{v,ij} + o(1),
\end{align} 
\begin{align*}
  \sup_{t \in \mathcal{N}(0)}\left\| \frac{1}{NM}\sum_{i=1}^N \sum_{j=1}^M   \partial_{v,t}\gamma_n(D_{ij},0,\theta_{0,n} + t) \right\|_2 = O_P(1)
\end{align*}
and $ \mathbb E \Delta_{v,ij}^2 = O(1)$.

\end{enumerate}   \label{A2:smoothness} 
\end{assumption}
Condition 1 of Assumption \ref{A2:smoothness} states that, in the absence of local misspecification, the parameter of interest is twice differentiable at $\theta_{0,n}$ with bounded derivatives. Meanwhile, Condition 2 requires the parameter of interest to be Gateaux differentiable in the local misspecification component. This condition ensures that the bias in the parameter of interest induced by the local misspecification is non-vanishing at a rate $\sqrt{n}$. All the parameters in Examples \ref{eg:4}--\ref{eg:6} satisfy this assumption. 
  
\begin{assumption}
\begin{enumerate}
  \item Let $p_{ij} = \mathbb E (Y_{ij}|X_i,A_i,W_j,B_j)$. Then, we assume $n p_{ij}$ and  $ \gamma_n(D_{ij},0,\theta_{0,n})$ have bounded support. 
  \item Suppose $||G_0||_{op} = O(1)$, $||B_{H,\gamma}||_2 = O(1)$, $||\Sigma_{H,\gamma}||_2 = O(1)$, 
\begin{align*}
   &  \lambda_1 \leq \lambda_{\min}\left(\Sigma_H\right) \leq \lambda_{\max}\left(\Sigma_H\right) \leq \lambda_2, \quad \lambda_{\max}\left(B_H\right) \leq \lambda_2, \\
   & \lambda_1 \leq \lambda_{\min}(G_0^\top G_0) \leq \lambda_{\max}(G_0^\top G_0) \leq \lambda_2,
\end{align*}
\begin{align*}
    \lambda_{\max}\left(\frac{1}{N}\sum_{i=1}^N \left[\mathbb E (s_{Y,ij} H_{ij}|X_i,A_i)\right]  \left[\mathbb E (s_{Y,ij} H_{ij}|X_i,A_i)\right]^\top \right) = O_P(1),
\end{align*}
and 
\begin{align*}
    \Sigma_{\gamma} + 2 \Sigma_{H,\gamma}^\top \kappa^* +  \kappa^{*,\top}  \Sigma_H \kappa^* \geq \lambda_1
\end{align*}
for some constants $0 < \lambda_1 \leq \lambda_2 < \infty$.
\end{enumerate}\label{ass:Y_gamma}
\end{assumption}
Assumption \ref{ass:Y_gamma} contains the regularity conditions that  ensure $\kappa^* = O(1)$.  

\begin{assumption}
Suppose $H_{ij} \in \Re^{k_n}$, $\max_{i \in [N], j \in [M]}||H_{ij}||_\infty \leq \zeta_n$, and  
\begin{align*}
    \frac{1}{NM}\sum_{i=1}^N \sum_{j=1}^M ||H_{ij}||_2^2 = O_P(k_n), 
\end{align*}
such that $k_n = o(n^{1/5})$ and $k_n^3 \zeta_n^2 \log^2(n) = o(n)$. \label{ass:kn}
\end{assumption}
Assumption \ref{ass:kn} controls the complexity of the sieve basis by imposing a rate condition on the dimension and strength of $H_{ij}$. Specifically, we require $k_n = o(n^{1/5})$ to ensure the Yurinskii's coupling of the score with an increasing dimension.\footnote{See, for example, \citet[Theorem 10]{P02} for the statement of the Yurinskii's coupling.} This rate is typically assumed in the literature. See, for example, the discussion after \citet[Theorem 1]{li/liao:2020}. The second rate requirement is weaker than the first if $\zeta_n = o(k_n)$, which is the case if the support of $X_i$ and $W_j$ are bounded and splines are used to construct the sieve basis. We also note that Assumption \ref{ass:kn} allows for $k_n$ to be fixed. Then, Condition 2 in Assumption \ref{ass:Y_gamma} requires $k_n \geq d_z+1$. 

\begin{assumption}
The initial estimator satisfies $\hat \theta_{\text{initial}} - \theta_{0} = O_P(n^{-1/2})$.   \label{ass:regular_intitial}
\end{assumption}
The logistic MLE satisfies this requirement; see \citet{graham:2022}. The next theorem shows that our bias-aware CI has the correct asymptotic size control under sparse network asymptotics. 

\begin{theorem}\label{theo:size}
 Suppose Assumptions \ref{ass:sampling}--\ref{ass:regular_intitial} hold and the bias-aware CI, i.e., $\widehat {CI}(\hat \kappa^*)$, is constructed as \eqref{eq:robust_CI_feasible}. Then, for the significance level $\alpha \in (0,1)$, we have
 \begin{align*}
\liminf_{n \rightarrow \infty}     \mathbb P( \Psi_{0,n} \in \widehat {CI}_{1-\alpha}(\hat \kappa^*)) \geq 1-\alpha. 
 \end{align*} 
\end{theorem}

\section{Asymptotic Optimality}\label{sec:optimality}
In this section, we show that robust one-step estimator $\hat \Psi(\hat \kappa^*)$ is minimax MSE-optimal compared to any other regular estimator that is asymptotically linear.

We consider an arbitrary estimator $\breve \theta$, which has the following linear expansion: 
\begin{align*}
	\sqrt{n}(\breve \theta - \theta_{0,n}) =  \kappa^\top_\theta  \left[\frac{n^{3/2}}{NM}\sum_{i=1}^N \sum_{j=1}^M (Y_{ij} - \Lambda(R_{ij}^\top \hat \theta_{0,n}))F_{ij}\right] + o_P(1), 
\end{align*}
where $F_{ij} \in \Re^{d_F}$ is $d_F$ dimensional vector of transformations of $(X_i,W_j)$,  $\kappa_F \in \Re^{d_F \times d_R}$, and $d_F$ is fixed. We further assume the estimator $\breve \theta$ is regular. 
\begin{assumption}\label{assumption:brevetheta}
Suppose $G_F^\top \kappa_\theta = I_{d_R} $ with $G_F = \mathbb E \alpha_0 \exp(Z_{ij}^\top \beta_0) F_{ij}R_{ij}^\top$.    Further assume
\begin{align*}
    \lambda_1 \leq \lambda_{\min}(G_F^\top G_F) \leq \lambda_{\max}(G_F^\top G_F) \leq \lambda_2
\end{align*}
for some constants $0 < \lambda_1 \leq \lambda_2 < \infty$.
\end{assumption}

The estimator $\breve \theta$ can be viewed as a GMM estimator based on the following set of unconditional moments:
\begin{align*}
      \mathbb{E}\left[ \left(Y_{ij}- \Lambda\left( R_{ij}^\top\theta \right)\right) F_{ij} \right]=0.
\end{align*}
A leading example of $\breve \theta$ is our initial estimator $\hat \theta_{\text{initial}}$, which satisfies the above moment condition with $F_{ij} = R_{ij}$. 

Consider the plug-in estimator $\breve \Psi$ of $\Psi_{0,n}$ defined as 
\begin{align*}
	\breve \Psi = \frac{1}{NM}\sum_{i=1}^N \sum_{j=1}^M \gamma(D_{ij},0,\breve \theta).
\end{align*} 
Following the arguments that lead to \eqref{eq:psi_asyexpansion}, we can show that 
\begin{align*}
 \sqrt n (\breve \Psi - \Psi_{0,n}) - \omega_{F,n} = o_P(1)
\end{align*}
and
\begin{align*}
 \omega_{F,n} \stackrel{d}{=} \N\left( \mathbb E \eta_{ij} ( \kappa^{\top}_F \alpha_0 \exp(Z_{ij}^\top \beta_{0})F_{ij} -\Delta_{v,ij}),  \Sigma_{\gamma} + 2\Sigma_{F,\gamma}^\top \kappa_F + \kappa^\top_F \Sigma_F \kappa_F \right),
\end{align*}
where $\Sigma_F$ and $\Sigma_{F,\gamma}$ are similarly defined as $\Sigma_H$ and $\Sigma_{H,\gamma}$ with $H_{ij}$ replaced by $F_{ij}$, and 
\begin{align}\label{eq:kappa_F}
    \kappa_F = \kappa_\theta \Gamma. 
\end{align}
Specifically, we have
\begin{align*}
  \begin{pmatrix}
        \Sigma_F & \Sigma_{F,\gamma} \\
        \Sigma_{F,\gamma}^\top & \Sigma_{\gamma}
    \end{pmatrix} = \frac{\Sigma^a_{F,n}}{1-\phi} + \frac{\Sigma^p_{F,n}}{\phi} + \frac{ \Sigma_{2,F,n}}{\phi(1-\phi)}
\end{align*}
with 
\begin{align*}
\Sigma_{F,n}^a = \mathbb E 
\begin{pmatrix}
[\mathbb E (s_{Y,ij}F_{ij}|X_i,A_i)][\mathbb E (s_{Y,ij}F_{ij}|X_i,A_i)]^\top	
&	
\mathbb E (s_{Y,ij}F_{ij}|X_i,A_i)\mathbb E (s_{\gamma,ij}|X_i,A_i)  \\
\mathbb E (s_{Y,ij}F_{ij}^\top |X_i,A_i)\mathbb E (s_{\gamma,ij}|X_i,A_i)  
& [\mathbb E (s_{\gamma,ij}|X_i,A_i)]^2
\end{pmatrix},
\end{align*}
\begin{align*}
\Sigma_{F,n}^p = \mathbb E \begin{pmatrix}
[\mathbb E (s_{Y,ij}F_{ij}|W_j,B_j)][\mathbb E (s_{Y,ij}H_{ij}|W_j,B_j)]^\top	&	\mathbb E (s_{Y,ij}F_{ij}|W_j,B_j)\mathbb E (s_{\gamma,ij}|W_j,B_j)  \\
\mathbb E (s_{Y,ij}F_{ij}^\top |W_j,B_j)\mathbb E (s_{\gamma,ij}|W_j,B_j)  & [\mathbb E (s_{\gamma,ij}|W_j,B_j)]^2
\end{pmatrix},
\end{align*}
and 
\begin{align*}
\Sigma_{2,F,n} = \mathbb E \begin{pmatrix}
   n^{-1} Var(s_{Y,ij}|X_i,A_i,W_j,B_j) F_{ij}F_{ij}^\top & 0\\
    0 & 		0
\end{pmatrix}.
\end{align*}

Then, the worst-case mean squared error for the plug-in estimator $\breve \Psi$ is  
\begin{align}\label{eq:MSE_F}
    MSE_F(\kappa_F) = \overline{\sigma}^2 \left[\kappa^\top_F B_F \kappa_F - 2B_{F,\gamma}^\top \kappa_F + B_\gamma \right] + ( \Sigma_{\gamma} + 2 \Sigma_{F,\gamma}^\top \kappa_F +  \kappa^\top_F  \Sigma_F \kappa_F),
\end{align}
where 
\begin{align*}
    B_F = \alpha_0^2 \mathbb E \exp(2Z_{ij}^\top \beta_0) F_{ij} F_{ij}^\top \quad \text{and} \quad B_{F,\gamma} = \alpha_0 \mathbb E \exp(Z_{ij}^\top \beta_0) F_{ij} \Delta_{v,ij}.
\end{align*}

\begin{assumption}\label{ass:F_approx}
Suppose there exists a matrix $\Pi_F \in \Re^{d_F \times k_n}$ and $\delta_{ij} \equiv \delta(X_i,W_j) \in \Re^{d_F}$ for some $d_F$-dimensional function $\delta$ such that 
\begin{align*}
	F_{ij} = \Pi_F H_{ij} + \delta_{ij}, 
\end{align*}
$||\Pi_F||_{op} = O(1)$, and $\mathbb E ||\delta_{ij}||_2^2 = o(1)$. 
\end{assumption}

Assumption \ref{ass:F_approx} implies each element of $F_{ij}$ can be well approximated by the sieve bases $H_{ij}$. This condition holds given that $F_{ij}$ is a sufficiently smooth function of $(X_i,W_j)$. 
\begin{theorem}\label{theo:optimal}
 Suppose Assumptions \ref{ass:sampling}--\ref{ass:F_approx} holds. Then, we have
\begin{align*}
   MSE_H(\kappa^*) \leq MSE_F(\kappa_F) + o(1),
\end{align*}
where $MSE_H(\cdot)$, $\kappa^*$, $\kappa_F$, and $MSE_F(\cdot)$  are defined in \eqref{eq:MSE_H}, \eqref{eq:kappa_star}, \eqref{eq:kappa_F}, and \eqref{eq:MSE_F}, respectively. 
\end{theorem} 

Theorem \ref{theo:optimal} shows that the one-step robust estimator $\hat \Psi(\hat \kappa^*)$ attains the smallest worst-case MSE relative to any arbitrary plug-in estimator $\breve \Psi$ that estimates regularly the parameter $\theta$ based on unconditional moments that are constructed using a general set of covariates $F_{ij}$. That is, the one-step robust estimator $\hat \Psi(\hat \kappa^*)$ is minimax-MSE optimal within a clas of regular estimators. This result holds even when the dimension $k_n$ of $H_{ij}$ is fixed, as long as Assumption \ref{ass:F_approx} is valid.  To guarantee Assumption \ref{ass:F_approx} in the fixed dimension case, we can construct $H_{ij}$ by stacking $F_{ij}$ with other functions of $(X_i,W_j)$. Theoretically, the benefit of allowing for a diverging number of $H_{ij}$ is that it expands the range of $F_{ij}$ that can be approximated by linear functions of $H_{ij}$, thereby broadening the set of estimators to which our one-step estimator is superior. In practice, we do not suggest using a very large $k_n$ when $n$ is moderate because it may (i) violate Assumption \ref{ass:kn} and (ii) cause numerical instability.

\section{Monte Carlo Simulations}\label{sec:simulations}

This section presents simulation evidence for the finite sample performance of the one-step robust estimator introduced in Section \ref{sec:robust_CI}. We consider a wide array of Monte Carlo designs that are meant to capture differences in the specification of the local misspecification, sample size and sparsity of the network.

The bipartite network is modelled according to the following data-generating process. For any $i \in \left[ N \right]$ and $j \in \left[ M \right]$, the observed attributes $X_i$ and $W_j$ are drawn from independent and identically distributed log-normal distributions with mean $-1/4$ and variance $1/2$. The dyad-specific attributes are computed to account for assortative matching. In particular, we define $Z_{ij} = \log\left( X_i \cdot W_j \right)$, which will be distributed normally with mean equal to $-1/2$ and variance $1$. The sieve basis is computed using a Hermite polynomial approximation of $Z_{ij}$ of order $k_n$. As an alternative sieve basis, we consider the tensor product of polynomial expansions on $X_i$ and $W_j$, which yields similar qualitative results.  

The unobserved heterogeneity $A_i$ and $B_j$ are drawn from independent and identically distributed log-normal distributions with mean $-1/12$ and standard deviation given by $1/\sqrt{6}$. Finally, $U_{ij}$ is distributed as a standard exponential distribution. The bipartite graphon is simulated according to the equation 
\begin{equation}\label{eq:Y_DGP_sim}
    Y_{ij} = 1\left[ \alpha_{0} -\log(n) + Z_{ij} \beta_0 + \log(A_i) + \log(B_i) + n^{-1/2} v_{ij} + U_{ij} \geq 0 \right], 
\end{equation}
where the term $\log(n)$ will ensure that the average degree of this network is bounded. In particular, under the current setting, the probability of establishing a link decreases at a rate $1/n$.

The distribution of the misspecification component $v_{ij}$ is simulated to capture three different designs: (i) latent homophily, (ii) functional form misspecification,  and (iii) semiparametric distribution. In the first design, the misspecification component is modelled as $v_{ij}=1\left[\sigma>0 \right](1+\sigma) * \overline{Z}_{n} * (v_i \cdot  v_j)$ where $v_i = 3/4 \, \log(X_i)  + 1/4 \, \text{N}(1, \sigma)$ and $v_j = 3/4 \, \log(W_j)  + 1/4 \, \text{N}(1, \sigma)$ represent fixed effects, and $\overline{Z}_{n} = (NM)^{-1} \sum_{i,j}Z_{ij}$. Here, $\text{N}(1, \sigma)$ stands for a normal random variable with mean $1$ and standard deviation $\sigma$. In the second design, we simulate $v_{ij} = 1\left[\sigma>0 \right](1+\sigma) * \overline{Z}_{n} * \left( Z_{ij}^2 + Z_{ij}^3\right)$. This design captures the misspecification that stems from truncating the higher-order effects of the observed attributes. Finally, for the last design, the misspecification component is modelled as $v_{ij} = 1\left[\sigma>0 \right](1+\sigma) * N\left(Z_{ij}, \sigma \right)$, which represents an heterogeneous latent mixing component in \eqref{eq:Y_DGP_sim}.

Notice that for all designs, the conditional density of $v_{ij}$ given by $\pi_{ij}$ depends on $(X_i, W_j)$. 
Moreover, the magnitude of misspecification is determined by $\sigma^2$ with $\sigma^2 \in \left\{0, 1, 2, 3, 4 \right\}$. A $\sigma^2=0$ represents the absence of local misspecification. Across each of the designs, the upper bound on the magnitude of local misspecification $\mathbb E \eta_{ij}^2$ is fixed as $\overline{\sigma}^2= \max_{\sigma^2\leq 4} \mathbb E \eta_{ij}^2$ with  $\mathbb E \eta_{ij}^2= var( \mathbb E (v_{ij}|X_i, W_j))$. That is, it corresponds to the largest value of $\mathbb E \eta_{ij}^2$ that is attained when the simulation design sets $\sigma^2=4$. The true DGP design is completed by setting $\alpha_{0} = \log(2.56)$, $\beta_0 = \log(4)$, and network size equal to $n = N + M$ with $N, M =100, 200, 300, 400$. 

The one-step robust estimator is computed using the definition in \eqref{eq:robust_estimator}. As an initial estimator $\hat\theta_{\text{initial}}$, we consider a logistic maximum likelihood estimator. We compare the performance of the robust methodology with that of a plug-in estimator with two different initial $\hat\theta_{\text{initial}}$: (i) a logistic regression and (ii) a Poisson regression. Given the exponential distribution of $U_{ij}$, the logistic estimator $\hat\theta_n$ will present a bias of order $o(n^{-2})$, as discussed in \citet{graham:2022}. This bias is negligible concerning the effect generated by the local misspecification, which induces a nonvanishing bias in the asymptotic distribution of order $n^{-1/2}$. On the contrary, the initial Poisson regressor represents the Composite Maximum Likelihood Estimator when the distribution of the error term is known. We report bias-aware confidence intervals for all the estimators, that is the one-step robust estimator and both plug-in estimators.\footnote{We also study the performance of the robust methodology when the Poisson regression is set as the initial estimator $\hat\theta_{\text{initial}}$. The results are close to being numerically equal to the baseline setting with the logistic estimator as $\hat\theta_{\text{initial}}$. This is a consequence of the fact that $\hat\Psi(\hat \kappa)$ is a regular estimator, which ensures that its strong approximation does not depend on the asymptotic distribution of the initial estimator. We report these results in Appendix \ref{sec:appx:sim_poisson_init}.}

As parameter of interest $\Psi_{0,n}$, we consider: (i) the homophily parameter $\beta_0$, and (ii) the average out-degree of the network $M\mathbb{E}\int \exp( R_{ij}^\top \theta_{0,n} + n^{-1/2} v)\pi_{ij}dv$. Notice that in case (i), the parameter of interest is not affected directly by the local misspecification; while in (ii), the functional form of $\Psi_{0,n}$ is also misspecified. Importantly, in both cases, the local misspecification plays a central role due to its effect on the conditional probability \eqref{eq:cprob}.

\subsection{Homophily parameters}
Table \ref{table:beta_robust_DGP_homophily} summarizes results from computing the one-step robust estimator across 2,000 Monte Carlo replications when the parameter of interest is $\Psi_{0,n} = \beta_0$,  $v_{ij}$ accounts for latent homophily, and the sieves dimension is set to $k_n=2$.\footnote{In Appendix \ref{sec:appx:sim_larger_sieves}, we report the results using larger dimensions for the sieves approximation.} The table includes the true value of $\beta_0$, along with the mean coefficient estimates (coeff.), $\sqrt{n}$-bias, standard error (s.e.), standard error-to-standard deviation ratio (s.e./sd), root-mean-squared error (rMSE), 95\% confidence interval (conf. int.), confidence interval length (length), and 95\% coverage probabilities (95\% CP). The final column shows the average degree of the network. These values were calculated over all simulations.

The top panel in Table \ref{table:beta_robust_DGP_homophily} shows the results of computing the robust estimator in a network of size $n=200$ and across misspecification designs $\sigma^2 \in \left\{0,1, 2, 3, 4 \right\}$. The mean estimate and $\sqrt{n}$-Bias show that the robust estimator approximates well the true DGP value. Moreover, the standard error-to-standard deviation ratio (s.e./sd) indicates that the estimator's sampling variability is well approximated. Finally, the bias-aware confidence intervals have the correct asymptotic size control under sparse asymptotic designs as indicated by Theorem \ref{theo:size}. Notice that the coverage probabilities are greater than 95\%, which is expected as the bias-aware confidence intervals are conservative. 

The bottom panels of Table \ref{table:beta_robust_DGP_homophily} report the results of the robust estimator in larger networks with sizes equal to $n=400$, $600$, and $800$. As the sample size grows, the point estimates approximate more precisely the true parameter value. Note that the $\sqrt{n}$-Bias does not vanish completely. This is expected as the average bias is scaled up by a factor of $\sqrt{n}$, and thus, it reflects the first-order effect that the local misspecification has on the limiting distribution of the robust estimator. Notice also that the length of the bias-aware confidence intervals decreases uniformly as $n$ grows throughout all the misspecification designs and even when the average network degree is as low as 0.3\%. 

Tables \ref{table:beta_logistic_DGP_homophily} and \ref{table:beta_poisson_DGP_homophily} report the results of the plug-in estimators with an initial logistic and Poisson regression, respectively. Table \ref{table:beta_ratios_DGP_homophily} reports the ratios of the plug-in estimators relative to the robust estimator to facilitate the analysis. Relative to the plug-in estimator with an initial logistic estimator, it is clear that the robust estimator exhibits smaller rMSEs. In fact, the rMSE obtained under the logistic estimator is, on average, 1.089 times larger than that of the robust methodology when the network size is $n=200$.  Even as both estimators become more precise in larger networks, the logistic estimator attains larger rMSE relative to the robust estimator as supported by Theorem \ref{theo:optimal}. They are, on average, 1.054 times larger when $n=400$, 1.036 times larger when $n=600$, and 1.029 times larger when $n=800$. Notice that the robust estimator also presents smaller standard errors, significantly less $\sqrt{n}$-Bias, and tighter confidence intervals. 

The robust estimator also outperforms the plug-in estimator with an initial Poisson regression in Table \ref{table:beta_poisson_DGP_homophily}. When comparing the two, the robust estimator exhibits smaller RMSEs. Albeit, the margins are smaller as the plug-in estimator coincides with the composite MLE. We observe that the Poisson regression has rMSEs that are, on average, 1.019 times larger when $n=200$, 1.014 times larger when $n=400$, 1.010 times larger when $n=600$, and 1.008 times larger when $n=800$. Notably, both estimators are comparable in terms of standard errors, $\sqrt{n}$-Bias, confidence interval's length, and coverage probabilities. This evidence suggests that under model misspecification, the one-step robust estimator attains shorter rMSEs relative to the composite MLE.

Appendix Tables \ref{table:beta_robust_DGP_functional}-\ref{table:beta_ratios_DGP_functional}, and \ref{table:beta_robust_DGP_normal}-\ref{table:beta_ratios_DGP_normal} report the simulation results of estimating $\beta_0$ for the designs of functional form misspecification and semiparametric distribution. Qualitatively, we observe similar patterns; that is, the robust one-step estimator dominates the plug-in logistic and Poisson estimators in terms of the rMSE. Moreover, it has comparable performance in terms of  $\sqrt{n}$-Bias and coverage probabilities.
      
\subsection{Average out-degree}
Table \ref{table:psi_robust_DGP_homophily} summarizes results from computing the one-step robust estimator across 2,000 Monte Carlo replications when the parameter of interest is the average out-degree of the network, $\Psi_{0,n} = M\mathbb{E}\int \exp\left[ R_{ij}^\top \theta_{0,n} + n^{-1/2} v \right] \pi_{ij}(v) dv$, $v_{ij}$ accounts for latent homophily, and $k_n=3$.\footnote{In Appendix \ref{sec:appx:sim_larger_sieves}, we report the results using larger dimensions in the sieves approximation.} Notice that, in this case, the local misspecification directly affects the parameter of interest, and thus, the true value of $\Psi_{0,n}$ varies across different specifications of $\sigma^2\in \left\{0, 1, 2, 3, 4 \right\}$. 

The results in Table \ref{table:psi_robust_DGP_homophily} suggest that the robust estimator yields reliable inference for the parameter of interest, and its performance improves in larger networks, notwithstanding the large degree of sparsity. The $\sqrt{n}$-Bias is relatively small and shows a decaying pattern as $n$ grows. The bias-aware confidence intervals have the correct asymptotic coverage, and their length decreases uniformly as $n$ grows and throughout all the misspecification designs.

Tables \ref{table:psi_logistic_DGP_homophily} and \ref{table:psi_poisson_DGP_homophily} report the results of the plug-in estimators with an initial logistic and Poisson regression, respectively.
Table \ref{table:psi_ratios_DGP_homophily} reports the ratios of the plug-in estimators relative to the robust estimator to facilitate the analysis. When comparing the plug-in estimator with an initial logistic regression and the robust estimator, it is clear that the robust estimator improves significantly on the plug-in estimator by attaining smaller rMSEs. In particular, the Logistic plug-in estimator has rMSE that are, on average, 1.084 times larger when $n=200$, 1.035 times when $n=400$, 1.023 times larger when $n=600$, and 1.016 times larger when $n=800$.  The same pattern is observed when comparing the robust estimator with the plug-in estimator with an initial Poisson regression in Table \ref{table:psi_poisson_DGP_homophily}. Although this second plug-in estimator provides better asymptotic guarantees, the results show that using the Poisson estimator yields similar rMSE when $n=200$, 1.003 times larger rMSE when $n=400$, 1.002 times larger for $n=600$, and 1.001 times larger when $n=800$.

Appendix Tables \ref{table:psi_robust_DGP_functional}-\ref{table:psi_ratios_DGP_functional} and \ref{table:psi_robust_DGP_normal}-\ref{table:psi_ratios_DGP_normal} report the simulation results of estimating $\Psi_{0,n}$ for the designs of functional form misspecification and semiparametric distribution. Qualitatively, we observe similar patterns. This evidence suggests that computing the optimal confidence intervals of the robust estimator might lead to a significant improvement when the parameter of interest is affected directly by the local misspecification.

\section{Empirical Application}\label{sec:application}

In this section, we implement the methodology developed in Section \ref{sec:size} to study a network of scientific collaborations among economists. We utilize articles published in leading American journals of general interest to construct a bipartite network connecting authors and articles. The goal is to estimate the factors influencing an author's decision to participate in a paper, thereby establishing a scientific collaboration with other authors who are also involved in that project.

This study contributes to the growing literature on scientific collaborations in economics. Most of the existing literature has focused on describing stylized features of this network (\citealt{newman:2001, goyal/etal:l2006} and \citealt{boschini/sjogren:2007}) or measuring the effect that these collaborations have on research productivity (\citealt{ductor/etal:2014, ductor:2015} and \citealt{colussi:2018}). Fewer studies have analyzed the mechanisms that drive the formation of these scientific collaborations (\citealt{fafchamps/vaderleij/goyal:2010,anderson/richards_shubik:2022} and \citealt{hsieh/konig/liu/zimmerman:2022}).\footnote{Relatively, the literature on the economics of innovation aims to identify individual contributions in teams production (see e.g. \citealt{bonhomme:2021})} We extend these studies by considering a larger number of factors that can influence an author's decision to collaborate on a project while using a computationally tractable method that accounts for two-sided heterogeneity. More importantly, this is the first paper to study the effects that local misspecification has on a bipartite network of scientific collaboration and conduct robust inference on the parameters of interest.

The data include all papers published in top American economic journals from 2000 to 2006: the American Economic Review, Econometrica, Journal of Political Economy, and Quarterly Journal of Economics. It contains author's information, such as gender, university granting PhD degree, graduation year, research fields, institution of employment and position. For articles, the dataset includes citation count since the year of publication, authors' names, publication issue, number of pages, references, three-digit JEL code, and keywords.\footnote{We thank Tommaso Colussi for sharing this data with us.} We used this data to define a single, static bipartite network of collaborations that includes a total of 1776 authors and 1600 articles.\footnote{Our final sample includes only unique observations of matched authors and articles.} 

We examine the formation probability of a collaboration link using \eqref{eq:cprob} and compare the one-step robust estimator's performance with plug-in estimators based on initial logistic and Poisson regressions. The parameters of interest include assortative matching parameters and the network's average out-degree. We control for author- and project-specific characteristics, along with author-project pair-specific observed attributes. In terms of individual author's attributes, we include binary variables for the categories of female gender and junior economist position and control for the average number of citations received by the author.\footnote{The category of junior economists includes assistant professors and economists in government or research institutions.}$^{,}$\footnote{The average citation count serves as a proxy measure for the author's productivity.}  As project-specific attributes, we include the number of authors in the project, indices for the share of female and senior authors in the project, and binary variables indicating whether some authors share the same institution or obtained their PhD from the same university.\footnote{We use information on the authors collaborating in a given project to construct these project-specific controls.}$^{,}$\footnote{The category of senior economists includes associate and full professors and senior economists in government or research institutions.} Table \ref{table:descriptive_stats} provides summary statistics on the network and author-specific and project-specific observed attributes. 

As main mechanisms for assortative matching, we consider: (i) similarity between author $i$' research fields and project $j$'s classification topic (\texttt{jelcode}), (ii) sorting of more productive authors into projects with higher impact (\texttt{citations}), (iii) gender sorting across collaborations (\texttt{gender}), and (iv) collaboration between junior and senior scientists (\texttt{junior\_senior}).\footnote{Ideally, the author's productivity should be measured using the citations count prior to year 2000. However, this is not possible because of data limitations. Alternatively, a different measure of productivity could be considered, such as the average number of citations received by the institution granting the PhD to the author.} Table \ref{table:robust_estimates} collects the estimation results using the robust methodology with an initial logistic estimator and sieves dimension $k_n=3$. It reports the coefficient estimates, standard errors, theoretical worst-case rMSEs,  confidence intervals, confidence intervals' length and the adjustment term of the one-step robust estimator. The results computed impose an upper bound on $\overline{\sigma}^2 = 4$. In Appendix Tables \ref{table:robust_appx_sigma_1},  \ref{table:robust_appx_sigma_2}, and \ref{table:robust_appx_sigma_3}, we report the results for the robust estimator where the upper bound  $\overline{\sigma}^2$ is specified equal to $1$, $2$, and $3$, respectively. Figures \ref{fig:CIs_robust_jelcode}--\ref{fig:CIs_robust_psi} depict the bias-aware confidence intervals for the parameters of interest across different levels $\overline{\sigma}^2$.

The results in Table \ref{table:robust_estimates} indicate that similarity in the research field of expertise is a strong and positive factor influencing the decision to collaborate on a research project. This result aligns with the findings in \citealt{ductor:2015} and \citealt{hsieh/konig/liu/zimmerman:2022}. 
Additionally, highly productive authors tend to participate in projects with higher impact with a positive probability. This outcome is consistent with the role of ``star'' economists discussed in \citet{goyal/etal:l2006}. 
The evidence also suggests that, during the sample period, sorting of female authors is a strong positive predictor for the establishment of a collaboration network. This pattern is also observed by \citet{boschini/sjogren:2007} during the 1991-2002 time period. On the contrary, this analysis suggests that junior scientists are less likely to participate in a project with a large share of senior authors. One likely explanation is that successful collaborations tend to perdure,  and it is costly to establish new connections (cf. \citealt{hsieh/konig/liu/zimmerman:2022}). Notice that the estimates predict a negative constant coefficient, which is expected from a sparse bipartite network. The analysis indicates that the average number of authors per project is 1.74.

When comparing these results to those of the plug-in estimator based on initial logistic and Poisson regressions in Tables \ref{table:logit_estimates} and \ref{table:poisson_estimates}, we observe that the plug-in estimators fail to capture the effects predicted on citations. Moreover, the one-step robust estimator provides substantial improvements in predicting the homophily coefficients and the average out-degree statistics. In particular, the robust estimator attains smaller theoretical rMSEs for almost all the estimates except for the gender attribute. Notice that the estimates associated with gender present the largest theoretical rMSEs across all the methodologies. This might be related to the fact that the shares of female authors observed in the sample and in a given project are small. The methodologies are comparable in terms of standard errors and confidence intervals' length. 

\section{Conclusion}\label{sec:conclusion}
We study the effects of local misspecification on a bipartite network. We focus on a class of dyadic network models characterized by conditional moment restrictions that are locally misspecified. The magnitude of misspecification is indexed by the sample size and vanishes at a rate $n^{-1/2}$.  We utilize this local asymptotic approach to construct a robust estimator that is MSE-minimax optimal within a prespecified neighborhood of misspecification. Additionally, we introduce bias-aware confidence intervals that account for the effect of the local misspecification. These confidence intervals are asymptotically valid for the structural parameters of interest under sparse network asymptotics, both in the correctly-specified and locally-misspecified case. In an empirical application, we study the formation of a scientific collaboration network among economists. Our analysis documents that homophily in the research field of expertise, sorting of highly productive scientists into higher-impact projects, and participation of female authors in teams with a higher share of female authors are all strong statistical factors that explain the formation of a collaboration network.

\clearpage
\newpage
%--------------------------------------------------------------------------------------------
% Simulation Tables 
%--------------------------------------------------------------------------------------------
% Set DGP Homophily
\begin{landscape}
\begin{sidewaystable}
\centering    
\sisetup{round-mode=places}
\caption{$\hat \beta_n$ Robust Estimator}
\label{table:beta_robust_DGP_homophily}
\resizebox{\textwidth}{!}{%
\begin{tabular}{ l S[round-precision=3] S[round-precision=3] S[round-precision=3] S[round-precision=3] S[round-precision=3] S[round-precision=3]
>{{[}} % Add square bracket before column
S[round-precision=3, table-format = -1.3,table-space-text-pre={[}]
@{,\,} % Add comma and thin-space between the columns
S[round-precision=3, table-format = -1.3,table-space-text-post={]}]
<{{]}} % Add square bracket after column
S[round-precision=3] S[round-precision=3] S[round-precision=3]
}  
\toprule
$\overline{\sigma}^2$ & {True} &  {coeff.} & {$\sqrt{n}$-Bias} & {s.e.} & {s.e./sd} & {rMSE} & \multicolumn{2}{c}{{[conf. int.]}} & {length} & {95\% CP} & {Degree} \\
\midrule
        \multicolumn{12}{c}{$n=200$} \\
        0.0    &  1.386294 &  1.384936 &  -0.064024 &  0.146513 &  0.884916 &  2.164864 &  0.651498 &  2.112036 &   1.460538 &  0.9995 &  0.010359 \\
1.0    &  1.386294 &  1.382533 &  -0.074874 &  0.146392 &  0.884879 &  2.165600 &  0.650950 &  2.111050 &   1.460100 &  0.9995 &  0.010359 \\
2.0    &  1.386294 &  1.381577 &  -0.073572 &  0.146885 &  0.884407 &  2.166627 &  0.650806 &  2.111378 &   1.460572 &  0.9995 &  0.010359 \\
3.0    &  1.386294 &  1.380808 &  -0.070893 &  0.146731 &  0.883814 &  2.162847 &  0.651332 &  2.111231 &   1.459898 &  1.0000 &  0.010371 \\
4.0    &  1.386294 &  1.380891 &  -0.066203 &  0.146241 &  0.882415 &  2.155220 &  0.652574 &  2.110652 &   1.458078 &  1.0000 &  0.010416 \\
        \multicolumn{12}{c}{$n=400$} \\
\bottomrule
\end{tabular}
}
\begin{flushleft}
\footnotesize{$^{1}$ Number of Monte Carlo simulations is $2,000$. $^{2}$ DGP Local misspecification: Latent Homophily. $^{3}$ Sieves dimension $k_n=2$.}    
\end{flushleft}
\end{sidewaystable}
\end{landscape}

\begin{landscape}
\begin{sidewaystable}
\centering
\sisetup{round-mode=places}
\caption{$\hat \beta_n$ Logistic Estimator}
\label{table:beta_logistic_DGP_homophily}
\resizebox{\textwidth}{!}{%
\begin{tabular}{ l S[round-precision=3] S[round-precision=3] S[round-precision=3] S[round-precision=3] S[round-precision=3] S[round-precision=3]
>{{[}} % Add square bracket before column
S[round-precision=3, table-format = -1.3,table-space-text-pre={[}]
@{,\,} % Add comma and thin-space between the columns
S[round-precision=3, table-format = -1.3,table-space-text-post={]}]
<{{]}} % Add square bracket after column
S[round-precision=3] S[round-precision=3] S[round-precision=3]
}  
\toprule
$\overline{\sigma}^2$ & {True} &  {coeff.} & {$\sqrt{n}$-Bias} & {s.e.} & {s.e./sd} & {rMSE} & \multicolumn{2}{c}{{[conf. int.]}} & {length} & {95\% CP} & {Degree} \\
\midrule
        \multicolumn{12}{c}{$n=200$} \\
        0.0    &  1.386294 &  1.418125 &   0.439142 &  0.159662 &  0.898730 &  2.357572 &  0.661026 &  2.173667 &   1.512641 &     1.0 &  0.010359 \\
1.0    &  1.386294 &  1.417213 &   0.427703 &  0.159093 &  0.896176 &  2.358537 &  0.660416 &  2.172659 &   1.512244 &     1.0 &  0.010359 \\
2.0    &  1.386294 &  1.417730 &   0.429232 &  0.159020 &  0.892221 &  2.359296 &  0.660326 &  2.172965 &   1.512638 &     1.0 &  0.010359 \\
3.0    &  1.386294 &  1.417209 &   0.432657 &  0.159411 &  0.894535 &  2.356015 &  0.660819 &  2.172957 &   1.512138 &     1.0 &  0.010371 \\
4.0    &  1.386294 &  1.415494 &   0.439686 &  0.158495 &  0.890720 &  2.348501 &  0.662208 &  2.172562 &   1.510354 &     1.0 &  0.010416 \\
        \multicolumn{12}{c}{$n=400$} \\
\bottomrule
\end{tabular}
}
\begin{flushleft}
\footnotesize{$^{1}$ Number of Monte Carlo simulations is $2,000$. $^{2}$ DGP Local misspecification: Latent Homophily.} 
\end{flushleft}
\end{sidewaystable}
\end{landscape}

\begin{landscape}
\begin{sidewaystable}
\centering
\sisetup{round-mode=places}
\caption{$\hat \beta_n$ Poisson Estimator}
\label{table:beta_poisson_DGP_homophily}
\resizebox{\textwidth}{!}{%
\begin{tabular}{ l S[round-precision=3] S[round-precision=3] S[round-precision=3] S[round-precision=3] S[round-precision=3] S[round-precision=3]
>{{[}} % Add square bracket before column
S[round-precision=3, table-format = -1.3,table-space-text-pre={[}]
@{,\,} % Add comma and thin-space between the columns
S[round-precision=3, table-format = -1.3,table-space-text-post={]}]
<{{]}} % Add square bracket after column
S[round-precision=3] S[round-precision=3] S[round-precision=3]
}  
\toprule
$\overline{\sigma}^2$ & {True} &  {coeff.} & {$\sqrt{n}$-Bias} & {s.e.} & {s.e./sd} & {rMSE} & \multicolumn{2}{c}{{[conf. int.]}} & {length} & {95\% CP} & {Degree} \\
\midrule
        \multicolumn{12}{c}{$n=200$} \\
        0.0    &  1.386294 &  1.383860 &  -0.081643 &  0.149187 &  0.905361 &  2.205843 &  0.660497 &  2.100546 &   1.440049 &  0.9995 &  0.010359 \\
1.0    &  1.386294 &  1.381215 &  -0.092474 &  0.149330 &  0.906950 &  2.206735 &  0.659904 &  2.099607 &   1.439704 &  0.9995 &  0.010359 \\
2.0    &  1.386294 &  1.380485 &  -0.091225 &  0.149550 &  0.904805 &  2.207491 &  0.659826 &  2.099862 &   1.440036 &  0.9995 &  0.010359 \\
3.0    &  1.386294 &  1.380167 &  -0.088597 &  0.149865 &  0.907074 &  2.203764 &  0.660324 &  2.099736 &   1.439412 &  1.0000 &  0.010371 \\
4.0    &  1.386294 &  1.379858 &  -0.084023 &  0.148426 &  0.899944 &  2.196045 &  0.661612 &  2.099094 &   1.437481 &  1.0000 &  0.010416 \\
        \multicolumn{12}{c}{$n=400$} \\
\bottomrule
\end{tabular}
}
\begin{flushleft}
\footnotesize{$^{1}$ Number of Monte Carlo simulations is $2,000$. $^{2}$ DGP Local misspecification: Latent Homophily.}     
\end{flushleft}
\end{sidewaystable}
\end{landscape}

% Ratios Beta DGP Homophily
\begin{table}[ht!]
\centering
\sisetup{round-mode=places}
\begin{threeparttable}
\caption{$\hat \beta_n$ Ratios}
\label{table:beta_ratios_DGP_homophily}
\begin{tabular}{ l 
S[round-precision=3] S[round-precision=3] S[round-precision=3] S[round-precision=3] 
S[round-precision=3] S[round-precision=3] S[round-precision=3] S[round-precision=3]
}  
\toprule
& \multicolumn{4}{c}{Logit/Robust} & \multicolumn{4}{c}{Poisson/Robust} \\
\cline{2-5}\cline{6-9}
$\overline{\sigma}^2$ & {s.e} & {bias} &  {rMSE} &  {length} & {s.e} & {bias} & {rMSE} &  {length} \\
\midrule
        \multicolumn{9}{c}{$n=200$} \\
        0.0 &  1.089744 &   6.859015 &  1.089016 &  1.035674 &  1.018252 &  1.275187 &  1.018929 &  0.985972 \\
1.0 &  1.086759 &   5.712281 &  1.089092 &  1.035712 &  1.020069 &  1.235056 &  1.018995 &  0.986031 \\
2.0 &  1.082610 &   5.834211 &  1.088926 &  1.035648 &  1.018143 &  1.239947 &  1.018861 &  0.985940 \\
3.0 &  1.086419 &   6.102993 &  1.089312 &  1.035783 &  1.021358 &  1.249734 &  1.018918 &  0.985967 \\
4.0 &  1.083796 &   6.641520 &  1.089681 &  1.035852 &  1.014940 &  1.269182 &  1.018943 &  0.985874 \\
\bottomrule
\end{tabular}
\begin{tablenotes}
\item[$^{1}$] \footnotesize{Number of Monte Carlo simulations is $2,000$.} 
\item[$^{2}$] \footnotesize{DGP Local misspecification: Latent Homophily.} 
\item[$^{3}$] \footnotesize{Sieves dimension $k_n=2$.}     
\end{tablenotes}
\end{threeparttable}
\end{table}

%% Psi Estimates DGP Homophily
\begin{landscape}    
\begin{sidewaystable}
\centering
\sisetup{round-mode=places}
\caption{$\hat \Psi_n$ Robust Estimator}
\label{table:psi_robust_DGP_homophily}
\resizebox{\textwidth}{!}{%
\begin{tabular}{ l S[round-precision=3] S[round-precision=3] S[round-precision=3] S[round-precision=3] S[round-precision=3] S[round-precision=3]
>{{[}} % Add square bracket before column
S[round-precision=3, table-format = -1.3,table-space-text-pre={[}]
@{,\,} % Add comma and thin-space between the columns
S[round-precision=3, table-format = -1.3,table-space-text-post={]}]
<{{]}} % Add square bracket after column
S[round-precision=3] S[round-precision=3] S[round-precision=3]
}  
\toprule
$\overline{\sigma}^2$ & {True} &  {coeff.} & {$\sqrt{n}$-Bias} & {s.e.} & {s.e./sd} & {rMSE} & \multicolumn{2}{c}{{[conf. int.]}} & {length} & {95\% CP} & {Degree} \\
\midrule
        \multicolumn{12}{c}{$n=200$} \\
        0.0    &  1.04078 &  1.021123 &  -0.118292 &  0.161829 &  0.950031 &  2.405740 &  0.696576 &  1.368255 &   0.671678 &  0.9310 &  0.010348 \\
1.0    &  1.04130 &  1.020989 &  -0.133929 &  0.161802 &  0.949840 &  2.405466 &  0.695999 &  1.367660 &   0.671661 &  0.9295 &  0.010342 \\
2.0    &  1.04112 &  1.020364 &  -0.126456 &  0.161972 &  0.950020 &  2.406518 &  0.696266 &  1.368090 &   0.671824 &  0.9280 &  0.010346 \\
3.0    &  1.04286 &  1.020130 &  -0.134546 &  0.162343 &  0.954913 &  2.408571 &  0.697140 &  1.369552 &   0.672412 &  0.9270 &  0.010357 \\
4.0    &  1.04840 &  1.025599 &  -0.162622 &  0.162867 &  0.956951 &  2.413504 &  0.700003 &  1.373799 &   0.673797 &  0.9255 &  0.010393 \\
        \multicolumn{12}{c}{$n=400$} \\
\bottomrule
\end{tabular}
}
\begin{flushleft}
\footnotesize{$^{1}$ Number of Monte Carlo simulations is $2,000$. $^{2}$ DGP Local misspecification: Latent Homophily. $^{3}$ Sieves dimension $k_n=3$.}    
\end{flushleft}
\end{sidewaystable}
\end{landscape}

\begin{landscape}
\begin{sidewaystable}
\centering
\sisetup{round-mode=places}
\caption{$\hat \Psi_n$ Logistic Estimator}
\label{table:psi_logistic_DGP_homophily}
\resizebox{\textwidth}{!}{%
\begin{tabular}{ l S[round-precision=3] S[round-precision=3] S[round-precision=3] S[round-precision=3] S[round-precision=3] S[round-precision=3]
>{{[}} % Add square bracket before column
S[round-precision=3, table-format = -1.3,table-space-text-pre={[}]
@{,\,} % Add comma and thin-space between the columns
S[round-precision=3, table-format = -1.3,table-space-text-post={]}]
<{{]}} % Add square bracket after column
S[round-precision=3] S[round-precision=3] S[round-precision=3]
}  
\toprule
$\overline{\sigma}^2$ & {True} &  {coeff.} & {$\sqrt{n}$-Bias} & {s.e.} & {s.e./sd} & {rMSE} & \multicolumn{2}{c}{{[conf. int.]}} & {length} & {95\% CP} & {Degree} \\
\midrule
        \multicolumn{12}{c}{$n=200$} \\
        0.0    &  1.04078 &  1.053372 &   0.358522 &  0.171938 &  0.936239 &  2.605881 &  0.715451 &  1.416811 &   0.701360 &  0.9375 &  0.010348 \\
1.0    &  1.04130 &  1.052440 &   0.342782 &  0.171838 &  0.935518 &  2.605433 &  0.714948 &  1.416128 &   0.701180 &  0.9360 &  0.010342 \\
2.0    &  1.04112 &  1.052864 &   0.350466 &  0.172216 &  0.936738 &  2.606327 &  0.715257 &  1.416546 &   0.701289 &  0.9345 &  0.010346 \\
3.0    &  1.04286 &  1.053157 &   0.343080 &  0.172451 &  0.940377 &  2.609870 &  0.716012 &  1.418226 &   0.702214 &  0.9385 &  0.010357 \\
4.0    &  1.04840 &  1.056593 &   0.316758 &  0.173149 &  0.943374 &  2.615160 &  0.718871 &  1.422725 &   0.703855 &  0.9360 &  0.010393 \\
        \multicolumn{12}{c}{$n=400$} \\
\bottomrule
\end{tabular}
}
\begin{flushleft}
\footnotesize{$^{1}$ Number of Monte Carlo simulations is $2,000$. $^{2}$ DGP Local misspecification: Latent Homophily.} 
\end{flushleft}
\end{sidewaystable}
\end{landscape}

\begin{landscape}
\begin{sidewaystable}
\centering
\sisetup{round-mode=places}
\caption{$\hat \Psi_n$ Poisson Estimator}
\label{table:psi_poisson_DGP_homophily}
\resizebox{\textwidth}{!}{%
\begin{tabular}{ l S[round-precision=3] S[round-precision=3] S[round-precision=3] S[round-precision=3] S[round-precision=3] S[round-precision=3]
>{{[}} % Add square bracket before column
S[round-precision=3, table-format = -1.3,table-space-text-pre={[}]
@{,\,} % Add comma and thin-space between the columns
S[round-precision=3, table-format = -1.3,table-space-text-post={]}]
<{{]}} % Add square bracket after column
S[round-precision=3] S[round-precision=3] S[round-precision=3]
}  
\toprule
$\overline{\sigma}^2$ & {True} &  {coeff.} & {$\sqrt{n}$-Bias} & {s.e.} & {s.e./sd} & {rMSE} & \multicolumn{2}{c}{{[conf. int.]}} & {length} & {95\% CP} & {Degree} \\
\midrule
        \multicolumn{12}{c}{$n=200$} \\
        0.0    &  1.04078 &  1.02 &  -0.084570 &  0.163157 &  0.959931 &  2.406032 &  0.707590 &  1.362010 &   0.654421 &  0.9260 &  0.010348 \\
1.0    &  1.04130 &  1.02 &  -0.099843 &  0.163427 &  0.961454 &  2.405642 &  0.707091 &  1.361389 &   0.654297 &  0.9240 &  0.010342 \\
2.0    &  1.04112 &  1.02 &  -0.092490 &  0.163241 &  0.959592 &  2.405830 &  0.707408 &  1.361752 &   0.654344 &  0.9210 &  0.010346 \\
3.0    &  1.04286 &  1.02 &  -0.100975 &  0.163488 &  0.963803 &  2.408407 &  0.708189 &  1.363251 &   0.655063 &  0.9235 &  0.010357 \\
4.0    &  1.04840 &  1.03 &  -0.129188 &  0.163727 &  0.964163 &  2.413648 &  0.710960 &  1.367570 &   0.656611 &  0.9220 &  0.010393 \\
        \multicolumn{12}{c}{$n=400$} \\
\bottomrule
\end{tabular}
}
\begin{flushleft}
\footnotesize{$^{1}$ Number of Monte Carlo simulations is $2,000$. $^{2}$ DGP Local misspecification: Latent Homophily.} 
\end{flushleft}
\end{sidewaystable}
\end{landscape}

% Ratios Psi DGP Homophily
\begin{table}[ht!]
\centering
\sisetup{round-mode=places}
\begin{threeparttable}
\caption{$\hat \Psi_n$ Ratios}
\label{table:psi_ratios_DGP_homophily}
\begin{tabular}{ l 
S[round-precision=3] S[round-precision=3] S[round-precision=3] S[round-precision=3] 
S[round-precision=3] S[round-precision=3] S[round-precision=3] S[round-precision=3]
}  
\toprule
& \multicolumn{4}{c}{Logit/Robust} & \multicolumn{4}{c}{Poisson/Robust} \\
\cline{2-5}\cline{6-9}
$\overline{\sigma}^2$ & {s.e} & {bias} &  {rMSE} &  {length} & {s.e} & {bias} & {rMSE} &  {length} \\
\midrule
        \multicolumn{9}{c}{$n=200$} \\
        0.0 &  1.062470 &   3.030830 &  1.083193 &  1.044190 &  1.008209 &  0.714928 &  1.000121 &  0.974306 \\
1.0 &  1.062024 &   2.559426 &  1.083130 &  1.043949 &  1.010044 &  0.745494 &  1.000073 &  0.974147 \\
2.0 &  1.063246 &   2.771448 &  1.083028 &  1.043858 &  1.007831 &  0.731398 &  0.999714 &  0.973981 \\
3.0 &  1.062267 &   2.549901 &  1.083576 &  1.044321 &  1.007056 &  0.750483 &  0.999932 &  0.974199 \\
4.0 &  1.063133 &   1.947820 &  1.083553 &  1.044610 &  1.005281 &  0.794411 &  1.000060 &  0.974494 \\
\bottomrule
\end{tabular}
\begin{tablenotes}
\item[$^{1}$] \footnotesize{Number of Monte Carlo simulations is $2,000$.} 
\item[$^{2}$] \footnotesize{DGP Local misspecification: Latent Homophily.} 
\item[$^{3}$] \footnotesize{Sieves dimension $k_n=3$.}     
\end{tablenotes}
\end{threeparttable}
\end{table}

%----------------------------------------------------------------------------------------%
% Empirical Application
%----------------------------------------------------------------------------------------%
\clearpage 
\newpage 

\begin{table}[ht!]
\centering
\sisetup{round-mode=places}
\begin{threeparttable}
\caption{Descriptive Statistics}
\label{table:descriptive_stats}
\begin{tabular}{ l S[round-precision=3] S[round-precision=3] S[round-precision=3]
    S[round-precision=3] S[round-precision=3] 
  }    
\toprule
                    & {mean} & {median} & {std} & {min} & {max} \\
\midrule
network\_links & 1.584459 & 1.000000 & 1.160860 & 0.000000 & 14.000000 \\
\multicolumn{6}{c}{Author's attributes} \\
gender & 0.105856 & 0.000000 & 0.307653 & 0.000000 & 1.000000 \\
junior dummy & 0.329955 & 0.000000 & 0.470196 & 0.000000 & 1.000000 \\
citations & 45.700621 & 35.000000 & 36.065130 & 0.000000 & 163.000000 \\
\multicolumn{6}{c}{Project's attributes} \\
share\_gender & 0.095333 & 0.000000 & 0.233227 & 0.000000 & 1.000000 \\
number\_authors & 1.958125 & 2.000000 & 0.775481 & 1.000000 & 5.000000 \\
share\_seniors & 0.485833 & 0.500000 & 0.395037 & 0.000000 & 1.000000 \\
same\_institution & 0.181250 & 0.000000 & 0.385225 & 0.000000 & 1.000000 \\
same\_PhD & 0.179375 & 0.000000 & 0.383666 & 0.000000 & 1.000000 \\
\bottomrule
\end{tabular}    
\begin{tablenotes}
\item[$^{1}$] \footnotesize{Total sample includes $N=1776$ authors and $M=1600$ articles.}
\end{tablenotes}
\end{threeparttable}
\end{table}

\clearpage
\newpage

\begin{table}[ht!]
\centering
\sisetup{round-mode=places}
\begin{threeparttable}
\caption{Results Robust Estimator}
\label{table:robust_estimates}
\begin{tabular}{ l S[round-precision=3] S[round-precision=3] S[round-precision=3]
    >{{[}} % Add square bracket before column
    S[round-precision=3, table-format = -1.3,table-space-text-pre={[}]
    @{,\,} % Add comma and thin-space between the columns
    S[round-precision=3, table-format = -1.3,table-space-text-post={]}]
    <{{]}} % Add square bracket after column
    S[round-precision=3]
    S[round-precision=3] 
  }    
\toprule
                    & {coeff.} & {s.e.} & {max rMSE} & \multicolumn{2}{c}{{[conf. int.]}} & {length} & {adjm} \\
\midrule
\multicolumn{8}{c}{Homophily parameters} \\
jelcode       &  1.633287 &  0.044595 &  10.243663 &  1.375316 &  1.891259 &  0.515943 &  0.126587 \\
citations     &  0.391383 &  0.014632 &  10.253093 &  0.186848 &  0.595917 &  0.409069 &  0.087766 \\
gender        &  5.937756 &  0.180474 &  69.119570 &  4.408206 &  7.467306 &  3.059100 &  0.481103 \\
junior\_senior & -3.671167 &  0.130258 &  23.323543 & -4.306160 & -3.036173 &  1.269987 & -0.623165 \\
constant      & -0.833779 &  0.079618 &  15.207093 & -1.239148 & -0.428410 &  0.810738 & -0.068691 \\
\multicolumn{8}{c}{Average out-degree} \\
Psi           &  1.739303 &  0.031193 &   2.063477 &  1.661189 &  1.817418 &  0.156229 & -0.031920 \\
\bottomrule
\end{tabular}    
\begin{tablenotes}
\item[$^{1}$] \footnotesize{Total sample includes $N=1776$ authors and $M=1600$ articles.}
\end{tablenotes}
\end{threeparttable}
\end{table}

\begin{table}[ht!]
\centering
\sisetup{round-mode=places}
\begin{threeparttable}
\caption{Results Logit Estimator}
\label{table:logit_estimates}
\begin{tabular}{ l S[round-precision=3] S[round-precision=3] S[round-precision=3]
    >{{[}} % Add square bracket before column
    S[round-precision=3, table-format = -1.3,table-space-text-pre={[}]
    @{,\,} % Add comma and thin-space between the columns
    S[round-precision=3, table-format = -1.3,table-space-text-post={]}]
    <{{]}} % Add square bracket after column
    S[round-precision=3]
    S[round-precision=3] 
  }    
\toprule
                    & {coeff.} & {s.e.} & {max rMSE} & \multicolumn{2}{c}{{[conf. int.]}} & {length} & {adjm} \\
\midrule
\multicolumn{8}{c}{Homophily parameters} \\
jelcode       &  1.506700 &  0.042679 &  11.772679 &  1.224981 &  1.788420 &  0.563440 &   0.0 \\
citations     &  0.303617 &  0.009477 &  21.632682 & -0.087151 &  0.694385 &  0.781536 &   0.0 \\
gender        &  5.456654 &  0.174605 &  71.557285 &  3.895323 &  7.017984 &  3.122660 &   0.0 \\
junior\_senior & -3.048002 &  0.132489 &  28.519568 & -3.780298 & -2.315705 &  1.464592 &   0.0 \\
constant      & -0.765088 &  0.079949 &  16.870465 & -1.200913 & -0.329262 &  0.871651 &   0.0 \\
\multicolumn{8}{c}{Average out-degree} \\
Psi           &  1.771223 &  0.036332 &   2.111029 &  1.700013 &  1.842433 &  0.142420 &   0.0 \\
\bottomrule
\end{tabular}
\begin{tablenotes}
\item[$^{1}$] \footnotesize{Total sample includes $N=1776$ authors and $M=1600$ articles.}
\end{tablenotes}
\end{threeparttable}
\end{table}

\begin{table}[ht!]
\centering
\sisetup{round-mode=places}
\begin{threeparttable}
\caption{Results Poisson Estimator}
\label{table:poisson_estimates}
\begin{tabular}{ l S[round-precision=3] S[round-precision=3] S[round-precision=3]
    >{{[}} % Add square bracket before column
    S[round-precision=3, table-format = -1.3,table-space-text-pre={[}]
    @{,\,} % Add comma and thin-space between the columns
    S[round-precision=3, table-format = -1.3,table-space-text-post={]}]
    <{{]}} % Add square bracket after column
    S[round-precision=3]
    S[round-precision=3] 
  }    
\toprule
                    & {coeff.} & {s.e.} & {max rMSE} & \multicolumn{2}{c}{{[conf. int.]}} & {length} & {adjm} \\
\midrule
\multicolumn{8}{c}{Homophily parameters} \\
jelcode       &  1.499157 &  0.042438 &  10.271102 &  1.244377 &  1.753936 &  0.509559 &   0.0 \\
citations     &  0.287122 &  0.009269 &  20.407781 & -0.082155 &  0.656399 &  0.738554 &   0.0 \\
gender        &  5.415339 &  0.172948 &  66.372159 &  3.947223 &  6.883456 &  2.936233 &   0.0 \\
junior\_senior & -3.036208 &  0.132070 &  26.362914 & -3.729138 & -2.343278 &  1.385860 &   0.0 \\
constant      & -0.756427 &  0.079655 &  15.914533 & -1.174610 & -0.338245 &  0.836365 &   0.0 \\
\multicolumn{8}{c}{Average out-degree} \\
Psi           &  1.758750 &  0.035621 &   2.069708 &  1.688934 &  1.828566 &  0.139632 &   0.0 \\
\bottomrule
\end{tabular}
\begin{tablenotes}
\item[$^{1}$] \footnotesize{Total sample includes $N=1776$ authors and $M=1600$ articles.}
\end{tablenotes}
\end{threeparttable}
\end{table}

% Figures
%----------------------------------------------------------------------------------------%
\begin{figure}[ht!]
\centering
\caption{Confidence Intervals: \texttt{jelcode}}
\label{fig:CIs_robust_jelcode}
\includegraphics[width=0.6\textwidth]{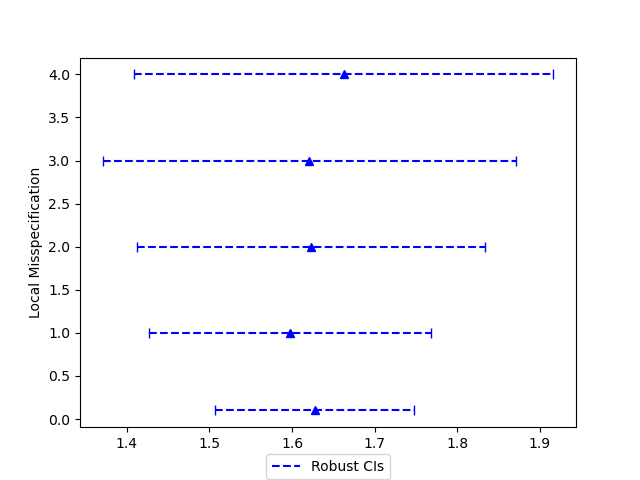}
\end{figure}

\begin{figure}[ht!]
\centering
\caption{Confidence Intervals: \texttt{citations}}
\label{fig:CIs_robust_citations}
\includegraphics[width=0.6\textwidth]{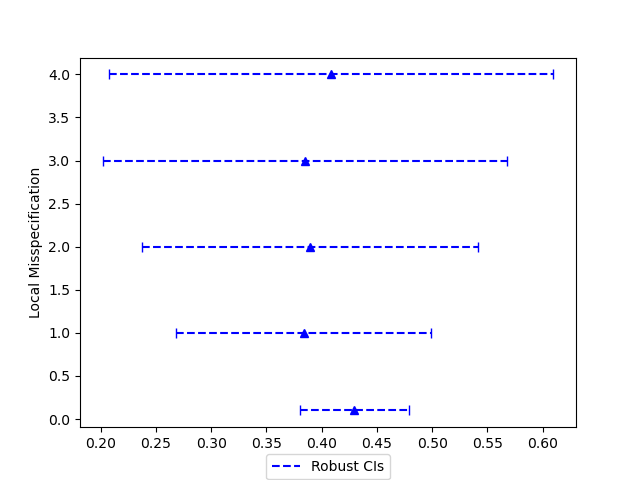}
\end{figure}

\begin{figure}[ht!]
\centering
\caption{Confidence Intervals: \texttt{gender}}
\label{fig:CIs_robust_gender}
\includegraphics[width=0.6\textwidth]{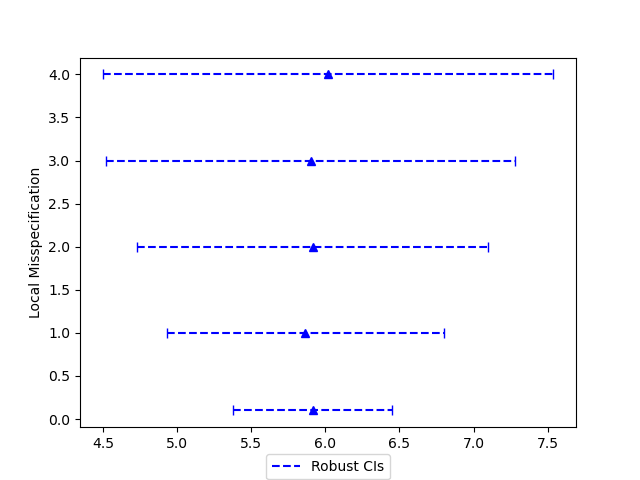}
\end{figure}

\begin{figure}[ht!]
\centering
\caption{Confidence Intervals: \texttt{junior-senior}}
\label{fig:CIs_robust_juniorsenior}
\includegraphics[width=0.6\textwidth]{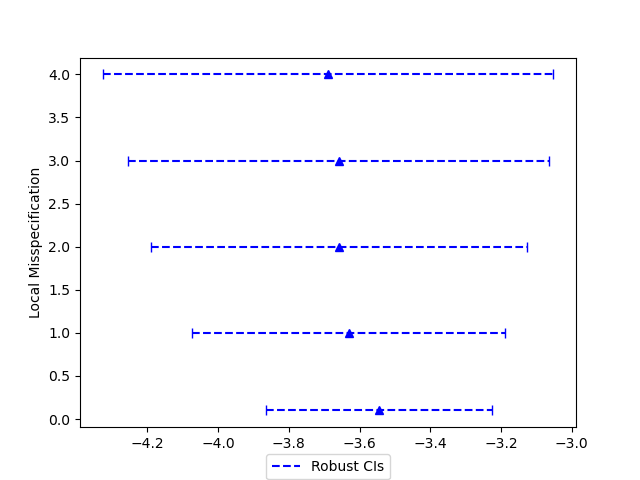}
\end{figure}

\begin{figure}[ht!]
\centering
\caption{Confidence Intervals: \texttt{constant}}
\label{fig:CIs_robust_constant}
\includegraphics[width=0.6\textwidth]{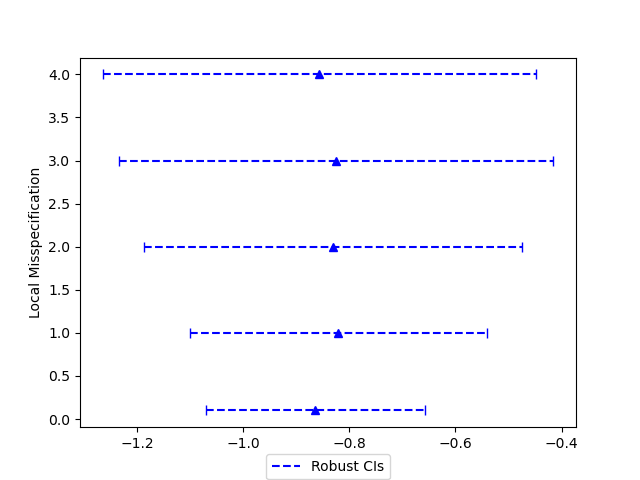}
\end{figure}

\begin{figure}[ht!]
\centering
\caption{Confidence Intervals: $\Psi_{0,n}$}
\label{fig:CIs_robust_psi}
\includegraphics[width=0.6\textwidth]{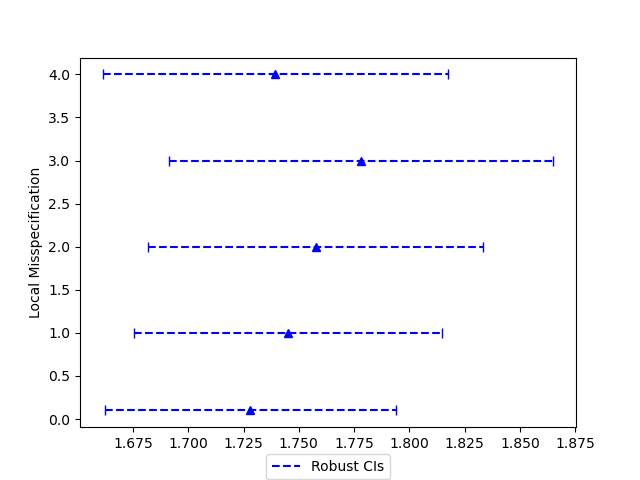}
\end{figure}
%--------------------------------------------------------------------------------------------------------------------------
\clearpage
\newpage
\appendix
\section*{Mathematical Appendix}
Throughout the appendix, we define $(X_1,\cdots,X_N)$, $(A_1,\cdots,A_N)$, $(W_1,\cdots,W_M)$, and $(B_1,\cdots,B_M)$ as $\mathcal X_N$, $\mathcal A_N$, $\mathcal W_M$, and $\mathcal B_M$, respectively. 

\section{Proof of Theorem \ref{theo:size}}\label{appendix:proof_size}
We divide the proof into two steps. In the first step, we establish the strong approximation
\begin{align}\label{eq:psi_approx_appendix}
\sqrt{n}(\hat \Psi(\hat \kappa^*) - \Psi_{0,n}) - \omega_n(\kappa^*) = o_P(1),
\end{align}
where 
\begin{align*}
\omega_n(\kappa^*) \stackrel{d}{=} \N\left( \mathbb E \eta_{ij} (\kappa^{*,\top} \alpha_0 \exp(Z_{ij}^\top \beta_{0})H_{ij} -\Delta_{v,ij}),  \Sigma_{\gamma} + 2\Sigma_{H,\gamma}^\top \kappa^* + \kappa^{*,\top} \Sigma_H \kappa^* \right).
\end{align*}

In the second step, we prove the desired result. The proof relies on the following lemmas, which are proved in Section \ref{sec:lem}.

\begin{lemma}{\label{lem:bias_variance_estimator}}
Suppose the Assumptions in Theorem \ref{theo:size} hold. Define $\hat G (\theta) = \frac{1}{NM}\sum_{i=1}^N \sum_{j=1}^M n \exp(R_{ij}^\top \theta) H_{ij}R_{ij}^\top$ so that $\hat G = \hat G (\theta_{\text{initial}})$. 
Then, we have that
\begin{align*}
&	\sup_{\theta \in \mathbb B(\theta_{0,n},C/\sqrt{n})}\left\| \hat G(\theta) - G_0\right\|_{op} = O_P\left(\sqrt{\frac{k_n^2 \zeta_n^2 \log^2(n)}{n}}\right) \quad \text{for any constant $C>0$,}\\
&  \left\| \hat B_H - B_H\right\|_{op} = O_P\left(\sqrt{\frac{k_n^2 \zeta_n^2 \log^2(n)}{n}}\right),\quad \left\| \hat B_{H,\gamma} - B_{H,\gamma}\right\|_{2} = O_P\left(\sqrt{\frac{k_n^2 \zeta_n^2 \log^2(n)}{n}}\right), \\
& \sup_{\theta \in \mathbb B(\theta_{0,n},C/\sqrt{n})}\left\|\frac{1}{NM}\sum_{i=1}^N \sum_{j=1}^M \partial_\theta \gamma_n(D_{ij},0, \theta) - \Gamma \right\|_{op} = O_P\left(\frac{1}{\sqrt{n}}\right), \quad ||\hat B_{\gamma} - B_{\gamma}||_2 =O_P\left(\frac{1}{\sqrt{n}}\right), \\
& \left\| \hat \Sigma_H - \Sigma_H \right\|_{op} = O_P\left(  \sqrt{\frac{k_n^2 \zeta_n^2}{n} } \right), \quad \left\| \hat \Sigma_{H,\gamma} - \Sigma_{H,\gamma} \right\|_{2} = O_P\left(  \sqrt{\frac{k_n^2 \zeta_n^2}{n} } \right), \\
& | \hat \Sigma_{\gamma} - \Sigma_{\gamma} | =O_P(\zeta_n n^{-1/2}), \quad \|\hat \kappa^* - \kappa^\ast \|_2 = o_P( k_n^{-1/2}), \quad \text{and} \\
& ||\kappa^*||_2 \leq C<\infty \quad \text{for some constant $C>0$},
\end{align*}
where $\mathbb{B}(\theta_{0,n},C/\sqrt{n})$ is a ball centered at $\theta_{0,n}$ with radius $C/\sqrt{n}$. 
\end{lemma}

\begin{lemma}\label{lem:strong_approx}
Suppose the Assumptions in Theorem \ref{theo:size} hold. Then there exists a $(k_n+1) \times 1$ Gaussian vector $g_n$ such that $g_n \stackrel{d}{=} \N(0,\Sigma_n)$ and  
\begin{align*}
& \left\| 
    \sqrt{n}
    \begin{pmatrix}
    & \frac{1}{NM}\sum_{i=1}^N \sum_{j=1}^M n \left(Y_{ij} - \int \Lambda(R_{ij}^\top \theta_{0,n}+n^{-1/2}v)\pi_{ij}(v)dv\right)H_{ij} \\
    & \frac{1}{NM}\sum_{i=1}^N \sum_{j=1}^M   \left[\gamma_n(D_{ij},0,\theta_{0,n}) -  \mathbb{E}\gamma_n(D_{ij},0, \theta_{0,n})\right] 
\end{pmatrix}  - g_n \right\|_2 = O_P(\sqrt{k_n^5/n}),
\end{align*}
where $\Sigma_n$ is defined in \eqref{eq:Sigma_n}.
\end{lemma}

\subsection{Step 1: Proof of \eqref{eq:psi_approx_appendix}}
We have
\begin{align}\label{eq:psihat}
	\sqrt{n}(\hat \Psi(\hat \kappa^*) - \Psi_{0,n}) & = \frac{\sqrt{n}}{NM}\sum_{i=1}^N \sum_{j=1}^M (\gamma(D_{ij},0,\hat \theta_{\text{initial}}) - \mathbb{E}\gamma(D_{ij},0,\theta_{0,n}) ) - \mathbb E \eta_{ij}\Delta_{v,ij}  \notag \\
	& + \hat \kappa^{*,\top}  \left[\frac{n^{3/2}}{NM}\sum_{i=1}^N \sum_{j=1}^M (Y_{ij} - \Lambda(R_{ij}^\top \hat \theta_{\text{initial}}))H_{ij}\right] + o(1)  \notag  \\ 
	& = \frac{\sqrt{n}}{NM}\sum_{i=1}^N \sum_{j=1}^M (\gamma(D_{ij},0, \theta_{0,n}) - \mathbb{E}\gamma(D_{ij},0,\theta_{0,n}) ) - \mathbb E \eta_{ij}\Delta_{v,ij} \notag  \\
	& + \left[\frac{1}{NM}\sum_{i=1}^N \sum_{j=1}^M \partial_\theta \gamma_n(D_{ij},0,\tilde \theta)^\top-\hat \kappa^{*,\top} \hat G (\tilde \theta) \right] \sqrt{n}(\hat \theta_{\text{initial}} - \theta_{0,n}) \notag \\
 & + \hat \kappa^{*,\top} \left[\frac{n^{3/2}}{NM}\sum_{i=1}^N \sum_{j=1}^M (Y_{ij} - \Lambda(R_{ij}^\top \theta_{0,n}))H_{ij}\right]+ o_P(1) \notag \\
 	& = \frac{\sqrt{n}}{NM}\sum_{i=1}^N \sum_{j=1}^M (\gamma(D_{ij},0, \theta_{0,n}) - \mathbb{E}\gamma(D_{ij},0,\theta_{0,n}) ) - \mathbb E \eta_{ij}\Delta_{v,ij} \notag \\
  &+ \hat \kappa^{*,\top} \left[\frac{n^{3/2}}{NM}\sum_{i=1}^N \sum_{j=1}^M (Y_{ij} - \int \Lambda(R_{ij}^\top \theta_{0,n} + n^{-1/2}v)d\pi_{ij} H_{ij}\right] \notag \\
  & + \hat \kappa^{*,\top}\left[\frac{n^{3/2}}{NM}\sum_{i=1}^N \sum_{j=1}^M \int (\Lambda(R_{ij}^\top \theta_{0,n}+ n^{-1/2}v) - \Lambda(R_{ij}^\top \theta_{0,n}))d\pi_{ij} H_{ij}\right] +  o_P(1) \notag \\
 	& = \frac{\sqrt{n}}{NM}\sum_{i=1}^N \sum_{j=1}^M (\gamma(D_{ij},0, \theta_{0,n}) - \mathbb{E}\gamma(D_{ij},0,\theta_{0,n}) ) + \mathbb E \eta_{ij} (\kappa^{*,\top} \alpha_0 \exp(Z_{ij}^\top \beta_{0})H_{ij} -\Delta_{v,ij}) \notag \\
  &+ \kappa^{*,\top} \left[\frac{n^{3/2}}{NM}\sum_{i=1}^N \sum_{j=1}^M (Y_{ij} - \int \Lambda(R_{ij}^\top \theta_{0,n} + vn^{-1/2})d\pi_{ij})H_{ij}\right]+ o_P(1),
\end{align}
where 
the first equality is by \eqref{eq:quadratic_expansion}, 
the second equality is by the mean-value theorem where $\tilde \theta$ is between $\hat \theta_{\text{initial}}$ and $\theta_{0,n}$ and the facts that 
\begin{align*}
&    \hat \kappa^{*,\top} \left[\frac{n^{3/2}}{NM}\sum_{i=1}^N \sum_{j=1}^M (\Lambda(R_{ij}^\top \hat \theta_{\text{initial}}))H_{ij} - \Lambda(R_{ij}^\top \hat \theta_{0,n}))H_{ij}\right] \\
&    = \hat \kappa^{*,\top}  \left[\frac{n^{3/2}}{NM}\sum_{i=1}^N \sum_{j=1}^M (\Lambda(R_{ij}^\top \tilde \theta))(1-\Lambda(R_{ij}^\top \tilde \theta))H_{ij}R_{ij}^\top\right] (\hat \theta_{\text{initial}} - \theta_{0,n}) \\
& = \kappa^{*,\top} \left[\hat G(\tilde \theta) - \frac{n}{NM}\sum_{i=1}^N \sum_{j=1}^M \frac{\exp(2R_{ij}^\top \hat \theta_{\text{initial}}) (2+\exp(R_{ij}^\top \hat \theta_{\text{initial}}))}{(1+\exp(R_{ij}^\top \hat \theta_{\text{initial}}))^2} \right] H_{ij}R_{ij}\sqrt{n}(\hat \theta_{\text{initial}} - \theta_{0,n})
\end{align*}
and 
\begin{align*}
& \left|    \kappa^{*,\top} \left[\frac{n}{NM}\sum_{i=1}^N \sum_{j=1}^M \frac{\exp(2R_{ij}^\top \hat \theta_{\text{initial}}) (2+\exp(R_{ij}^\top \hat \theta_{\text{initial}}))}{(1+\exp(R_{ij}^\top \hat \theta_{\text{initial}}))^2} \right] \sqrt{n}(\hat \theta_{\text{initial}} - \theta_{0,n}) \right| \\
& \leq 2||\kappa^{*,\top}||_2 \left\| \frac{n}{NM}\sum_{i=1}^N \sum_{j=1}^M \exp(2R_{ij}^\top \hat \theta_{\text{initial}})H_{ij}R_{ij}\right\|_{op} \left\Vert_2\sqrt{n}(\hat \theta_{\text{initial}} - \theta_{0,n}) \right\Vert_2 \\
& \leq O_P(n^{-1}) \times  \left\| \frac{1}{NM}\sum_{i=1}^N \sum_{j=1}^M \exp\left(2(\hat \alpha_n - \log n) + 2Z_{ij}^\top \hat \beta_{n}\right) H_{ij}R_{ij}\right\|_{op} = O_P(n^{-1}), 
\end{align*}
and the third equality is by Lemma \ref{lem:bias_variance_estimator},  
\begin{align*}
   & \left\| \hat G^\top(\tilde \theta) \hat \kappa^* -  \frac{1}{NM}\sum_{i=1}^N \sum_{j=1}^M \partial_\theta \gamma_n(D_{ij},0,\tilde \theta) \right\|_{op} \\
   & \leq    \left\| (\hat G(\tilde \theta) - G_0) \right\|_{op} ||\hat \kappa^*||_2 + ||G_0||_{op} ||\hat \kappa^* - \kappa^*||_2 + \left\|\frac{1}{NM}\sum_{i=1}^N \sum_{j=1}^M \partial_\theta \gamma_n(D_{ij},0,\tilde \theta) - \Gamma \right\|_{op}\\
  & =  O_P\left(\sqrt{\frac{k_n^2 \zeta_n^2 \log^2(n)}{n}}\right) +o_P(1)= o_P(1) 
\end{align*}
and $\sqrt{n}(\hat \theta_{\text{initial}} - \theta_{0,n}) = O_P(1)$. To see the last equality of \eqref{eq:psihat}, we note that   
\begin{align*}
    & \left| (\kappa^* - \hat \kappa^*)^\top  \left[\frac{n^{3/2}}{NM}\sum_{i=1}^N \sum_{j=1}^M (Y_{ij} - \int \Lambda(R_{ij}^\top \theta_{0,n} + vn^{-1/2})d\pi_{ij} )H_{ij}\right] \right| \\
    & \leq ||\kappa^* - \hat \kappa^*||_2 \left\| \frac{n^{3/2}}{NM}\sum_{i=1}^N \sum_{j=1}^M (Y_{ij} - \int \Lambda(R_{ij}^\top \theta_{0,n}+vn^{-1/2})d\pi_{ij})H_{ij} \right\|_2 \\
    & \leq ||\kappa^* - \hat \kappa^*||_2 (\left\|g_n \right\|_2 + o_P(1)) = o_P(1),
\end{align*}
where the second inequality is by Lemma \ref{lem:strong_approx} and the facts that $||\kappa^* - \hat \kappa^*||_2 = o_P(k_n^{-1/2})$ as shown in Lemma \ref{lem:bias_variance_estimator} and $||g_n||_2 = O_P(k_n^{1/2})$.\footnote{When $k_n$ is fixed, we interpret $o_P(k_n^{-1/2})$ and $O_P(k_n^{1/2})$ as $o_P(1)$ and $O_P(1)$, respectively.}

In addition, we have 
\begin{align*}
&   \biggl\| \hat \kappa^{*,\top}\left[\frac{n^{3/2}}{NM}\sum_{i=1}^N \sum_{j=1}^M \int (\Lambda(R_{ij}^\top \theta_{0,n}+ n^{-1/2}v) - \Lambda(R_{ij}^\top \theta_{0,n}))d\pi_{ij} H_{ij}\right] - \kappa^{*,\top}\mathbb E \alpha_0 \exp(Z_{ij}^\top \beta_{0})H_{ij}\eta_{ij}\biggr\|_2 \\
   & \leq \biggl\| (\hat \kappa^{*,\top}- \kappa^{*,\top})\mathbb E \alpha_0 \exp(Z_{ij}^\top \beta_{0})H_{ij}\eta_{ij}\biggr\|_2 \\
   & + \left\| \hat \kappa^{*,\top} \left[ \frac{n}{NM}\sum_{i=1}^N \sum_{j=1}^M \Lambda(R_{ij}^\top \theta_{0,n})(1-\Lambda(R_{ij}^\top \theta_{0,n})) H_{ij}\eta_{ij} -\mathbb E \alpha_0 \exp(Z_{ij}^\top \beta_{0})H_{ij}\eta_{ij} \right]\right\|_2 \\
   & + \left\|\hat \kappa^{*,\top}\frac{n^{1/2}}{NM}\sum_{i=1}^N \sum_{j=1}^M \int \left( \int_0^v \Lambda(R_{ij}^\top \theta_{0,n}+ \tilde v)(1-\Lambda(R_{ij}^\top \theta_{0,n}+ \tilde v))(1-2\Lambda(R_{ij}^\top \theta_{0,n} + \tilde v ))\tilde v d \tilde v\right) d\pi_{ij}\right\|_2 \\
   & = o_P(1),
\end{align*}
where the inequality is by Taylor's expansion with the integral remainder term and the equality is by the fact that 
\begin{align*}
& \left|    n \left( \int \Lambda(R_{ij}^\top \theta_{0,n}+ \tilde v)(1-\Lambda(R_{ij}^\top \theta_{0,n}+ \tilde v))(1-2\Lambda(R_{ij}^\top \theta_{0,n} + \tilde v ))\tilde v d \tilde v\right) d\pi_{ij} \right| \\
& \leq C \int \int_0^{|v|} \exp(\tilde v) \tilde v d\tilde v \pi_{ij}(v) dv \leq C
\end{align*}
by Assumption \ref{ass:sampling}. Combining the above two bounds, we can establish the last equality in \eqref{eq:psihat}. 

Then, by Lemma \ref{lem:strong_approx}, we have
\begin{align*}
& \sqrt{n}(\hat \Psi(\hat \kappa^*) - \Psi_{0,n}) \\
& = (1,\kappa^*)^\top g_n + \mathbb E \eta_{ij} (\kappa^{*,\top} \alpha_0 \exp(Z_{ij}^\top \beta_{0})H_{ij} -\Delta_{v,ij})\\
& + (1,\kappa^*)^\top \left( \sqrt{n}
    \begin{pmatrix}
    & \frac{1}{NM}\sum_{i=1}^N \sum_{j=1}^M n \left(Y_{ij} - \int \Lambda(R_{ij}^\top \theta_{0,n}+n^{-1/2}v)\pi_{ij}(v)dv\right)H_{ij} \\
    & \frac{1}{NM}\sum_{i=1}^N \sum_{j=1}^M   \left[\gamma_n(D_{ij},0,\theta_{0,n}) -  \mathbb{E}\gamma_n(D_{ij},0, \theta_{0,n})\right] 
\end{pmatrix}  - g_n\right) + o_P(1) \\
& = (1,\kappa^*)^\top g_n + \mathbb E \eta_{ij} (\kappa^{*,\top} \alpha_0 \exp(Z_{ij}^\top \beta_{0})H_{ij} -\Delta_{v,ij})+ (1+||\kappa^*||_2) O_P\left(\sqrt{k_n^5/n}\right) + 
o_P(1) \\
& = (1,\kappa^*)^\top g_n + \mathbb E \eta_{ij} (\kappa^{*,\top} \alpha_0 \exp(Z_{ij}^\top \beta_{0})H_{ij} -\Delta_{v,ij})+ o_P(1),
\end{align*}
where the last equality holds because $||\kappa^*||_2 = O(1)$ and $k_n^5/n = o(1)$. 

In addition, we note that 
\begin{align*}
|\mathbb E \eta_{ij} (\kappa^{*,\top} \alpha_0 \exp(Z_{ij}^\top \beta_{0})H_{ij} -\Delta_{v,ij})| & \leq (\mathbb E \eta_{ij}^2)^{1/2} (\mathbb E((\kappa^{*,\top} \alpha_0 \exp(Z_{ij}^\top \beta_{0})H_{ij} -\Delta_{v,ij}))))^{1/2} \\
& \leq \overline{\sigma}(\mathbb E((\kappa^{*,\top} \alpha_0 \exp(Z_{ij}^\top \beta_{0})H_{ij} -\Delta_{v,ij})))^{1/2} = \mathcal B(\kappa^*).
\end{align*}

\subsection{Step 2: Proof of the Desired Result.}
We note that 
\begin{align*}
& \mathbb P( \Psi_{0,n} \in \widehat {CI}_{1-\alpha}(\hat \kappa^*)) \\
& = \mathbb P\left( \frac{| (1,\kappa^*)^\top g_n + \mathbb E \eta_{ij} (\kappa^{*,\top} \alpha_0 \exp(Z_{ij}^\top \beta_{0})H_{ij} -\Delta_{v,ij})+ o_P(1)|}{( \hat \Sigma_{\gamma} + 2 \hat \Sigma_{H,\gamma}^\top \hat \kappa^* +  \hat \kappa^{*,\top} \hat \Sigma_H \hat \kappa^*)^{1/2} } \leq   CV_\alpha \left(    \frac{ \hat {\mathcal{B}}(\hat \kappa^*)}{(\hat \Sigma_{\gamma} + 2 \hat \Sigma_{H,\gamma}^\top \hat \kappa^* + \hat \kappa^{*,\top} \hat \Sigma_H \hat \kappa^*)^{1/2}}  \right)  \right) \\
& \geq \mathbb P\begin{pmatrix}
    \frac{| (1,\kappa^*)^\top g_n + \mathbb E \eta_{ij} (\kappa^{*,\top} \alpha_0 \exp(Z_{ij}^\top \beta_{0})H_{ij} -\Delta_{v,ij})|}{( \Sigma_{\gamma} + 2 \Sigma_{H,\gamma}^\top  \kappa^* +   \kappa^{*,\top}  \Sigma_H  \kappa^*)^{1/2} } \\
  \leq   CV_\alpha \left(    \frac{ {\mathcal{B}}( \kappa^*)}{( \Sigma_{\gamma} + 2  \Sigma_{H,\gamma}^\top \kappa^* +  \kappa^{*,\top} \Sigma_H \kappa^*)^{1/2}}  \right) - o_P(1)    
\end{pmatrix} \\
& = \mathbb P\begin{pmatrix}
    \frac{| (1,\kappa^*)^\top g_n + \mathbb E \eta_{ij} (\kappa^{*,\top} \alpha_0 \exp(Z_{ij}^\top \beta_{0})H_{ij} -\Delta_{v,ij})|}{( \Sigma_{\gamma} + 2 \Sigma_{H,\gamma}^\top  \kappa^* +   \kappa^{*,\top}  \Sigma_H  \kappa^*)^{1/2} } \\
  \leq   CV_\alpha \left(    \frac{ {\mathcal{B}}( \kappa^*)}{( \Sigma_{\gamma} + 2  \Sigma_{H,\gamma}^\top \kappa^* +  \kappa^{*,\top} \Sigma_H \kappa^*)^{1/2}}  \right)    
\end{pmatrix}  - \sup_{t \in \Re} \mathbb P\begin{pmatrix}
    \frac{| (1,\kappa^*)^\top g_n + t|}{( \Sigma_{\gamma} + 2 \Sigma_{H,\gamma}^\top  \kappa^* +   \kappa^{*,\top}  \Sigma_H  \kappa^*)^{1/2} } \leq  o_P(1)    
\end{pmatrix} \\
& = \mathbb P\begin{pmatrix}
    \frac{| (1,\kappa^*)^\top g_n + \mathbb E \eta_{ij} (\kappa^{*,\top} \alpha_0 \exp(Z_{ij}^\top \beta_{0})H_{ij} -\Delta_{v,ij})|}{( \Sigma_{\gamma} + 2 \Sigma_{H,\gamma}^\top  \kappa^* +   \kappa^{*,\top}  \Sigma_H  \kappa^*)^{1/2} } \\
  \leq   CV_\alpha \left(    \frac{ {\mathcal{B}}( \kappa^*)}{( \Sigma_{\gamma} + 2  \Sigma_{H,\gamma}^\top \kappa^* +  \kappa^{*,\top} \Sigma_H \kappa^*)^{1/2}}  \right)    
\end{pmatrix} - \sup_{t \in \Re} \mathbb P(|\N(0,1)+t| \leq o_P(1)) \\
& \geq 1-\alpha - o(1),
\end{align*}
where the first inequality holds because by Lemma \ref{lem:bias_variance_estimator}, we have $ |\hat {\mathcal{B}}(\hat \kappa^*) -  \mathcal{B}(\hat \kappa^*)| = o_P(1)$,  
\begin{align*}
 |   (\hat \Sigma_{\gamma} + 2 \hat \Sigma_{H,\gamma}^\top \hat \kappa^* + \hat \kappa^{*,\top} \hat \Sigma_H \hat \kappa^*)^{1/2} - ( \Sigma_{\gamma} + 2  \Sigma_{H,\gamma}^\top \kappa^* +  \kappa^{*,\top} \Sigma_H \kappa^*)^{1/2}| = o_P(1),
\end{align*}
and 
\begin{align*}
  \left|  CV_\alpha \left(    \frac{ \hat {\mathcal{B}}(\hat \kappa^*)}{(\hat \Sigma_{\gamma} + 2 \hat \Sigma_{H,\gamma}^\top \hat \kappa^* + \hat \kappa^{*,\top} \hat \Sigma_H \hat \kappa^*)^{1/2}}  \right) -  CV_\alpha \left(    \frac{ {\mathcal{B}}( \kappa^*)}{( \Sigma_{\gamma} + 2  \Sigma_{H,\gamma}^\top \kappa^* +  \kappa^{*,\top} \Sigma_H \kappa^*)^{1/2}}  \right)  \right| = o_P(1). 
  \end{align*}
The last inequality follows from the fact that the univariate standard normal density is bounded and the construction of the critical value function $CV_\alpha(\cdot)$. By taking $\liminf_{n \rightarrow \infty}$ on both sides of the above display, we obtain the desired result.

\section{Proof of Theorem \ref{theo:optimal}}
Because $\kappa_\theta$ satisfies Assumption \ref{assumption:brevetheta}, we can rewrite $\kappa_\theta$ as
\begin{align*}
  \kappa_\theta =  G_F(G_F^\top G_F)^{-1} (I_{d_R}-G_F^\top \Xi)+\Xi,
\end{align*}
where $\Xi$ is some deterministic $d_F \times d_R$ matrix that does not depend on $n$. Then, we have 
\begin{align*}
    \kappa_F = \left[G_F(G_F^\top G_F)^{-1} (I_{d_R}-G_F^\top \Xi)+\Xi\right]\Gamma.
\end{align*}
Further define 
\begin{align*}
\tilde    \kappa_F = \left[\Pi_F G_0(G_0^\top\Pi_F^\top \Pi_F G_0)^{-1}(I_{d_R} - G_0^\top \Pi_F^\top \Xi) + \Xi\right]\Gamma
\end{align*}
and 
\begin{align*}
    \kappa_H = \Pi_F^\top \tilde    \kappa_F. 
\end{align*}
We rely the following three claims to establish the desired result: (1) $||\kappa_F - \tilde \kappa_F||_2 = o(1)$, (2) $MSE_F(\kappa_F) = MSE_F(\tilde \kappa_F)+o(1)$, and (3) $MSE_F(\tilde \kappa_F) = MSE_H(\kappa_H) + o(1)$. Claims (2) and (3) imply that 
\begin{align*}
    MSE_F(\kappa_F) = MSE_H(\kappa_H) + o(1).
\end{align*}
In addition, we note that $G_0^\top \kappa_H = \Gamma$, implying that $\kappa_H$ satisfies \eqref{eq:regularity}. Therefore, by the definition of $\kappa^*$, we have
\begin{align*}
     MSE_H(\kappa_H) \geq  MSE_H(\kappa^*), 
\end{align*}
which is the desired result. 

Next, we prove the three claims. 

\textbf{Proof of Claim (1).} 
We have
\begin{align*}
	G_F = \mathbb E \alpha_0 \exp(Z_{ij}^\top \beta_0) (\Pi_F H_{ij} + \delta_{ij})R_{ij}^\top = \Pi_F G_0 + G_\delta,
\end{align*}
where $\|G_\delta\|_{op} = o(1)$. Then, by Assumption \ref{assumption:brevetheta}, we have
\begin{align*}
  \left\|  (G_F^\top G_F)^{-1} - (G_0^\top\Pi_F^\top \Pi_F G_0)^{-1} \right\|_{op} = o(1),
\end{align*}
which implies 
\begin{align*}
 ||\kappa_F - \tilde \kappa_F||_2 & \leq C    \left\|G_F(G_F^\top G_F)^{-1} (I_{d_R}-G_F^\top \Xi) -  \Pi_F G_0(G_0^\top\Pi_F^\top \Pi_F G_0)^{-1}(I_{d_R} - G_0^\top \Pi_F^\top \Xi) \right\|_2 \\
 & \leq C\left\|G_F(G_F^\top G_F)^{-1}  - \Pi_F G_0(G_0^\top\Pi_F^\top \Pi_F G_0)^{-1} \right\|_2 \\
 & + C\left\|G_F(G_F^\top G_F)^{-1}  G_F^\top - \Pi_F G_0(G_0^\top\Pi_F^\top \Pi_F G_0)^{-1} G_0^\top\Pi_F^\top \right\|_2 = o(1).
\end{align*}

\textbf{Proof of Claim (2).} Given Claim (1), Claim (2) holds by the continuity of the quadratic function $MSE_F(\kappa_F)(\cdot)$. 

\textbf{Proof of Claim (3).} We note that 
\begin{align*}
    & \left|(\tilde \kappa_F^\top B_F \tilde \kappa_F)^{1/2} - (\kappa_H^\top B_H \kappa_H)^{1/2}\right| \\
    & \leq C \left| \left(\tilde \kappa_F^\top \mathbb E \exp(2Z_{ij}^\top \beta_0) (\Pi_F H_{ij}+\delta_{ij}) (\Pi_F H_{ij}+\delta_{ij})^\top \tilde 
 \kappa_F\right)^{1/2} - \left(\kappa_F^\top \mathbb E \exp(2Z_{ij}^\top \beta_0) (\Pi_F H_{ij}) (\Pi_F H_{ij})^\top \kappa_F\right)^{1/2} \right| \\
    & \leq C \left[\mathbb E  \exp(2Z_{ij}^\top \beta_0) (\tilde \kappa_F^\top \delta_{ij})^2 \right]^{1/2} \leq C \left(\mathbb E \left\|\delta_{ij}\right\|_2^2\right)^{1/2} = o(1). 
\end{align*}
This implies 
\begin{align*}
    \left|\tilde \kappa_F^\top B_F \tilde \kappa_F -  \kappa_H^\top B_H \kappa_H\right| = o(1),
\end{align*}
given that $\tilde \kappa_F^\top B_F \tilde \kappa_F$ is fixed and 
\begin{align*}
\kappa_H^\top B_H \kappa_H =  ( (\tilde \kappa_F^\top B_F \tilde \kappa_F)^{1/2} +o(1))^2 = O(1).  
\end{align*}
In the same manner, we can show that  
\begin{align*}
&   B^\top_{F,\gamma} \tilde \kappa_F = B^\top_{H,\gamma} \kappa_H + o(1), \quad   \tilde \kappa_F^\top \Sigma_F \tilde \kappa_F = \kappa_H^\top \Sigma_H \kappa_H + o(1), \quad \text{and}  \quad \Sigma^\top_{F,\gamma} \kappa_F = \Sigma^\top_{H,\gamma} \kappa_H + o(1).
\end{align*}
This concludes the proof.

\newpage
\section{Proof of Lemmas \ref{lem:bias_variance_estimator} and \ref{lem:strong_approx}}\label{sec:lem}

We use the following notation throughout this section. For a scalar function $f(U_i,V_j)$ for dyadic data $(U_i,V_j)$, we denote it as $f_{ij}$ and write $\mathbb P_{N,M}f_{ij} = \frac{1}{NM} \sum_{i=1}^N \sum_{j=1}^M f_{ij}$, $\mathbb P_{i,\cdot} f_{ij} = \mathbb E (f_{ij}|U_i) $, $\mathbb P_{\cdot,j} f_{ij} =  \mathbb E (f_{ij}|V_j)$, $\mathbb P f_{ij} = \mathbb E f_{ij}$, $\mathbb P_N f_{ij} = \frac{1}{N}\sum_{i=1}^Nf_{ij}$, $ \mathbb P_M f_{ij} = \frac{1}{M}\sum_{j=1}^M f_{ij}$, and $\mathbb U_{N,M}f_{ij} = \mathbb P_{N,M}f_{ij} - \mathbb P_N \mathbb P_{i,}f_{ij}- \mathbb P_M \mathbb P_{,j}f_{ij} + \mathbb P f_{ij}$. Based on these definitions, we have $\mathbb P_{N,M}f_{ij} = \mathbb P_N \mathbb P_M f_{ij} =  \mathbb P_M \mathbb P_N f_{ij}$. 

\subsection{Proof of Lemma \ref{lem:bias_variance_estimator}}
\textbf{Proof of $\sup_{\theta \in \mathbb B(\theta_{0,n},C/\sqrt{n})}\left\| \hat G(\theta) - G_0\right\|_{op}  = O_P\left(\sqrt{\frac{k_n^2 \zeta_n^2 \log^2(n)}{n}}\right)$.} 

Note that 
\begin{align*}
    & \sup_{\theta \in \mathbb B(\theta_{0,n},C/\sqrt{n})}\left\| \hat G(\theta) - G_0\right\|_{op} \\
    & \leq 
    \sup_{\theta \in \mathbb B(\theta_{0,n},C/\sqrt{n})}\left\| \hat G(\theta) - G(\theta)\right\|_{op} + \left\| G(\theta_{0,n}) - G_0 \right\| 
    \\
    & 
    =  
    \sup_{v_1 \in \Re^{k_n}, ||v_1||_2 = 1, v_2 \in \Re^{d_R}, ||v_2||_2 = 1, (\alpha,\beta) \in \mathbb B((\alpha_0,\beta_0),C/\sqrt{n}) }
    | (\mathbb P_{N,M} - \mathbb P) \exp(\alpha+Z_{ij}^\top \beta)v_1^\top H_{ij}R_{ij}^\top v_2| + O_P(n^{-1/2}) 
    \\
    & 
    \leq 
    \sup_{v_1 \in \Re^{k_n}, ||v_1||_2 = 1, v_2 \in \Re^{d_R}, ||v_2||_2 = 1, (\alpha,\beta) \in \mathbb B((\alpha_0,\beta_0),C/\sqrt{n}) }
    | \mathbb P_N (\mathbb P_{M} - \mathbb P_{i,\cdot}) \exp(\alpha+Z_{ij}^\top \beta)v_1^\top H_{ij}R_{ij}^\top v_2|  
    \\
    & + \sup_{v_1 \in \Re^{k_n}, ||v_1||_2 = 1, v_2 \in \Re^{d_R}, ||v_2||_2 = 1, (\alpha,\beta) \in \mathbb B((\alpha_0,\beta_0),C/\sqrt{n})}| (\mathbb P_N - \mathbb P) \left[\mathbb P_{i,\cdot} \exp(\alpha+Z_{ij}^\top \beta)v_1^\top H_{ij}R_{ij}^\top v_2\right]| + O_P(n^{-1/2})
    \\
    & 
    \leq 
    \sup_{i = 1,\cdots,N, v_1 \in \Re^{k_n}, ||v_1||_2 = 1, v_2 \in \Re^{d_R}, ||v_2||_2 = 1, (\alpha,\beta) \in \mathbb B((\alpha_0,\beta_0),C/\sqrt{n})}
    | (\mathbb P_{M} - \mathbb P_{i,\cdot}) \exp(\alpha+Z_{ij}^\top \beta)v_1^\top H_{ij}R_{ij}^\top v_2|  
    \\
    & + \sup_{v_1 \in \Re^{k_n}, ||v_1||_2 = 1, v_2 \in \Re^{d_R}, ||v_2||_2 = 1, (\alpha,\beta) \in \mathbb B((\alpha_0,\beta_0),C/\sqrt{n})}
    | (\mathbb P_N - \mathbb P) \left[\mathbb P_{i,\cdot} \exp(\alpha+Z_{ij}^\top \beta)v_1^\top H_{ij}R_{ij}^\top v_2\right]| + O_P(n^{-1/2}).
\end{align*}
We note that $\exp(\alpha + Z_{ij}^\top \beta) v_1^\top H_{ij} R_{ij}^\top v_2$ is independent across $j = 1,\cdots,M$ and 
\begin{align*}
& 
\sup_{i=1,\cdots,N, v_1 \in \Re^{k_n}, ||v_1||_2 = 1, v_2 \in \Re^{d_R}, ||v_2||_2 = 1, (\alpha,\beta) \in \mathbb B((\alpha_0,\beta_0),C/\sqrt{n})}
\left| \exp(\alpha + Z_{ij}^\top \beta) v_1^\top H_{ij} R_{ij}^\top v_2 \right|  \\
& \leq \max_{i \in [N]}  \exp(\alpha_0 + Z_{ij}^\top \beta_0+ C||Z_{ij}||_2/\sqrt{n}) ||H_{ij}||_2 ||R_{ij}||_2 \leq C k_n^{1/2} \zeta_n \equiv F.
\end{align*}

Then, the class of functions
\begin{align*}
&     
\mathcal F = \cup_{i = 1}^N \mathcal F_i 
\\ 
& 
\mathcal F_i = \left\{ \exp(\alpha + Z_{ij}^\top \beta) v_1^\top H_{ij} R_{ij}^\top v_2: v_1 \in \Re^{k_n}, ||v_1||_2 = 1, v_2 \in \Re^{d_R}, ||v_2||_2 = 1, (\alpha,\beta) \in \mathbb B((\alpha_0,\beta_0),C/\sqrt{n}) \right\}
\end{align*}
is of the VC type in the sense that 
\begin{align*}
    \sup_{Q}N(\mathcal{F},||\cdot||_{Q,2}, \epsilon ||F||_{Q,2})\leq \sum_{i=1}^N \sup_{Q}N(\mathcal{F}_j,||\cdot||_{Q,2}, \epsilon ||F||_{Q,2})\leq N\left(\frac{a}{\epsilon}\right)^{ck_n} \leq \left(\frac{a}{\epsilon}\right)^{c(k_n \vee \log n)},~0<\forall \epsilon \leq 1,
\end{align*}
where $\sup_Q$ means the supremum over all finitely discrete distributions. Therefore, by \citet[Corollary 5.1]{CCK14} we have
\begin{align*}
    & 
    \mathbb E \sup_{i = 1,\cdots,N, v_1 \in \Re^{k_n}, ||v_1||_2 = 1, v_2 \in \Re^{d_R}, ||v_2||_2 = 1, (\alpha,\beta) \in \mathbb B((\alpha_0,\beta_0),C/\sqrt{n})}
    | (\mathbb P_{M} - \mathbb P_{i,\cdot}) \exp(\alpha+Z_{ij}^\top \beta)v_1^\top H_{ij}R_{ij}^\top v_2| 
    \\
    & \leq C\left(\sqrt{\frac{k_n^2 \zeta_n^2 \log^2(n)}{n}} + \frac{k_n^{3/2} \zeta_n \log^2(n)}{n} \right) \leq C \sqrt{\frac{k_n^2 \zeta_n^2 \log^2(n)}{n}}, 
\end{align*}
which implies 
\begin{align*}
    & 
    \sup_{i = 1,\cdots,N, v_1 \in \Re^{k_n}, ||v_1||_2 = 1, v_2 \in \Re^{d_R}, ||v_2||_2 = 1, (\alpha,\beta) \in \mathbb B((\alpha_0,\beta_0),C/\sqrt{n})} \left| (\mathbb P_{M} - \mathbb P_{i,\cdot}) \exp(\alpha+Z_{ij}^\top \beta)v_1^\top H_{ij}R_{ij}^\top v_2 \right| \\
    & = O_P\left(\sqrt{\frac{k_n^2 \zeta_n^2 \log^2(n)}{n}}\right).
\end{align*}

Similarly, we can show that 
\begin{align*}
& \sup_{v_1 \in \Re^{k_n}, ||v_1||_2 = 1, v_2 \in \Re^{d_R}, ||v_2||_2 = 1, (\alpha,\beta) \in \mathbb B((\alpha_0,\beta_0),C/\sqrt{n})} \left| (\mathbb P_N - \mathbb P) \left[\mathbb P_{i,\cdot} \exp(\alpha+Z_{ij}^\top \beta)v_1^\top H_{ij}R_{ij}^\top v_2\right]\right| \\
 & = O_P\left(\sqrt{\frac{k_n^2 \zeta_n^2 \log^2(n)}{n}}\right), 
\end{align*}
which implies the desired result. 

We can establish $\left\| \hat B_H - B_H \right\|_{op} = O_P\left(\sqrt{\frac{k_n^2 \zeta_n^2 \log^2(n)}{n}}\right)$ and $\left\| \hat B_{H,\gamma} - B_{H,\gamma}\right\|_{2} = O_P\left(\sqrt{\frac{k_n^2 \zeta_n^2 \log^2(n)}{n}}\right)$  in the same manner.

\textbf{Proof of $\sup_{\theta \in \mathbb B(\theta_{0,n},C/\sqrt{n})}\left\| \frac{1}{NM}\sum_{i=1}^N \sum_{j=1}^M \partial_\theta \gamma_n(D_{ij},0, \theta) - \Gamma \right\|_2  =  O_P(n^{-1/2})$.} We have 
\begin{align*}
    & \left\| \frac{1}{NM}\sum_{i=1}^N \sum_{j=1}^M \partial_\theta \gamma_n(D_{ij},0, \theta) - \Gamma \right\|_2 \\
    & \leq  \left\|\frac{1}{NM}\sum_{i=1}^N \sum_{j=1}^M \partial_{\theta,\theta^\top} \gamma_n(D_{ij},0,\tilde \theta) ( \theta  - \theta_{0,n}) \right\|_2 + \left\| \frac{1}{NM}\sum_{i=1}^N \sum_{j=1}^M \partial_\theta \gamma_n(D_{ij},0, \theta_{0,n}) -\Gamma \right\|_2 \\
    & \leq \left\|\frac{1}{NM}\sum_{i=1}^N \sum_{j=1}^M \partial_{\theta,\theta^\top} \gamma_n(D_{ij},0,\tilde \theta)\right\|_{op} \left\|(\theta  - \theta_{0,n}) \right\|_2 + O_P(n^{-1/2}),
\end{align*}
where $\tilde \theta$ is between $\theta$ and $\theta_{0,n}$ and thus belongs to $\mathbb B(\theta_{0,n},C/\sqrt{n})$. Then, by Assumption \ref{A2:smoothness}.2 and 
taking $\sup_{\theta \in \mathbb B(\theta_{0,n},C/\sqrt{n})}$ on both sides, we obtain the desired result. 

\textbf{Proof of $||\hat \Sigma_H - \Sigma_H||_{op} = O_P\left(  \sqrt{\frac{k_n^2 \zeta_n^2}{n} } \right)$ and $\left\| \hat \Sigma_{H,\gamma} - \Sigma_{H,\gamma} \right\|_{2} = O_P\left(  \sqrt{\frac{k_n^2 \zeta_n^2}{n} } \right)$.} We note that 
\begin{align*}
    \hat \Sigma_H = \hat \Sigma^a_{YY,n} + \hat \Sigma^p_{YY,n} + \frac{1}{n}\left[ \frac{1}{NM}\sum_{i=1}^N \sum_{j=1}^M \hat s_{Y,ij}^2 H_{ij}H_{ij}^\top -  \hat \Sigma^a_{YY,n} - \hat \Sigma^p_{YY,n} \right].
\end{align*}

By Lemma \ref{lem:Sigma^a_YY}, we have  
\begin{align}\label{eq:Sigma^a_YY}
    || \hat \Sigma^a_{YY,n} -  \Sigma^a_{YY,n}||_{op} = O_P\left(  \sqrt{\frac{k_n^2 \zeta_n^2}{n} } \right)
\end{align}
and 
\begin{align}\label{eq:Sigma^p_YY}
    || \hat \Sigma^p_{YY,n} -  \Sigma^p_{YY,n}||_{op} = O_P\left(  \sqrt{\frac{k_n^2 \zeta_n^2}{n} } \right).
\end{align}

Next, we focus on $\frac{1}{NM}\sum_{i=1}^N \sum_{j=1}^M \hat s_{Y,ij}^2 H_{ij}H_{ij}^\top$. Recall $s_{Y,ij} = n \left(Y_{ij} - \int \Lambda(R_{ij}^\top \theta_{0,n} + n^{-1/2} v) \pi_{ij}dv\right)H_{ij}$ and $p_{ij} = \mathbb E (Y_{ij}|X_i,A_i,W_j,B_j)$, and define 
\begin{align*}
  \overline  s_{Y,ij} = \mathbb E(s_{Y,ij}|X_i,A_i,W_j,B_j) = n( p_{ij} - \mathbb E(p_{ij}|X_i,W_j)). 
\end{align*}

Then, we have
\begin{align}\label{eq:hat_s^2}
    \frac{1}{NM}\sum_{i=1}^N \sum_{j=1}^M \hat s_{Y,ij}^2 H_{ij}H_{ij}^\top = \frac{1}{NM}\sum_{i=1}^N \sum_{j=1}^M  s_{Y,ij}^2 H_{ij}H_{ij}^\top + \frac{1}{NM}\sum_{i=1}^N \sum_{j=1}^M (\hat s_{Y,ij}+s_{Y,ij})(\hat s_{Y,ij}-s_{Y,ij}) H_{ij}H_{ij}^\top. 
\end{align}
For the first term on the RHS of \eqref{eq:hat_s^2}, we have
\begin{align}\label{eq:hat_s^2_1}
\frac{1}{NM}\sum_{i=1}^N \sum_{j=1}^M  s_{Y,ij}^2 H_{ij}H_{ij}^\top  & = \frac{1}{NM}\sum_{i=1}^N \sum_{j=1}^M  \overline s_{Y,ij}^2 H_{ij}H_{ij}^\top  + \frac{2}{NM}\sum_{i=1}^N \sum_{j=1}^M  \overline s_{Y,ij}(s_{Y,ij}-\overline s_{Y,ij}) H_{ij}H_{ij}^\top \notag \\
& + \frac{1}{NM}\sum_{i=1}^N \sum_{j=1}^M  (s_{Y,ij}-\overline s_{Y,ij})^2 H_{ij}H_{ij}^\top \notag \\
& = \frac{1}{NM}\sum_{i=1}^N \sum_{j=1}^M  \overline s_{Y,ij}^2 H_{ij}H_{ij}^\top  + \frac{2}{NM}\sum_{i=1}^N \sum_{j=1}^M  n \overline s_{Y,ij}(Y_{ij}-p_{ij}) H_{ij}H_{ij}^\top \notag \\
& + \frac{1}{NM}\sum_{i=1}^N \sum_{j=1}^M  n^2(Y_{ij} - p_{ij}) H_{ij}H_{ij}^\top - \frac{2}{NM}\sum_{i=1}^N \sum_{j=1}^M  n^2(Y_{ij} - p_{ij})p_{ij} H_{ij}H_{ij}^\top  \notag \\
& + \frac{1}{NM}\sum_{i=1}^N \sum_{j=1}^M  n^2p_{ij}(1-p_{ij})H_{ij}H_{ij}^\top \notag \\
& \equiv I+2II+III-2IV + V.
\end{align}

For term $I$ in the RHS of \eqref{eq:hat_s^2_1}, we have
\begin{align*}
    \overline s_{Y,ij}^2 = n^2(p_{ij}-\mathbb E(p_{ij}|X_i,W_j))^2 \leq C,
\end{align*}
for some constant $C<\infty$, and thus, 
\begin{align}\label{eq:I_bd}
   \left\| I \right\|_{op} \leq C \left\| \frac{1}{NM}\sum_{i=1}^N \sum_{j=1}^M  H_{ij}H_{ij}^\top  \right\|_{op} = O_P(1). 
\end{align}

For term $II$ in the RHS of \eqref{eq:hat_s^2_1}, we have
\begin{align*}
    & \mathbb E \left( \left\| II \right\|_{F}^2 \mid \mathcal X_N, \mathcal A_N, \mathcal W_M, \mathcal B_M\right) \\
    &= \sum_{(\ell,\ell')\in [k_n] \times [k_n]} \mathbb E \left\{ \left[\frac{1}{NM}\sum_{i=1}^N \sum_{j=1}^M  n \overline s_{Y,ij}(Y_{ij}-p_{ij}) H_{ij,\ell}H_{ij,
    \ell'} \right]^2  \mid \mathcal X_N, \mathcal A_N, \mathcal W_M, \mathcal B_M\right\} \\
    & =\sum_{(\ell,\ell')\in [k_n] \times [k_n]} \frac{1}{N^2M^2}\sum_{i=1}^N \sum_{j=1}^M  \mathbb E \left\{ \left[n \overline s_{Y,ij}(Y_{ij}-p_{ij}) H_{ij,\ell}H_{ij,
    \ell'} \right]^2  \mid \mathcal X_N, \mathcal A_N, \mathcal W_M, \mathcal B_M\right\}\\
    & \leq C \frac{n }{NM} \mathbb P_{N,M} ||H_{ij}||_2^4 = O_P\left(\frac{k_n^2 \zeta_n^2}{n}\right),
\end{align*}
which implies 
\begin{align}\label{eq:II_bd}
   \left\| II \right\|_{op}  \leq   \left\| II \right\|_{F}    =  O_P\left(\sqrt{ \frac{k_n^2 \zeta_n^2}{n} }\right).
\end{align}

In the same manner, we can show that 
\begin{align}\label{eq:III_bd}
     \left\| n^{-1}III \right\|_{op} = O_P\left(\sqrt{ \frac{k_n^2 \zeta_n^2}{n} }\right)
\end{align}
and 
\begin{align}\label{eq:IV_bd}
     \left\| IV \right\|_{op} = O_P\left(\sqrt{ \frac{k_n^2 \zeta_n^2}{n} }\right).
\end{align}
Last, for term $V$ in the RHS of \eqref{eq:hat_s^2_1}, by Hoeffding decomposition, we have 
\begin{align}\label{eq:V_hoeffding}
    & n^{-1} V - \mathbb E \left( n p_{ij}(1-p_{ij})H_{ij}H_{ij}^\top \right) \notag \\
    & = (\mathbb P_{N,M} - \mathbb P)\left( n p_{ij}(1-p_{ij})H_{ij}H_{ij}^\top \right) \notag \\
    & = (\mathbb P_{N}- \mathbb P) \left[ \mathbb P_{i,} \left(n p_{ij}(1-p_{ij})H_{ij}H_{ij}^\top\right) \right] + (\mathbb P_{M}- \mathbb P) \left[ \mathbb P_{,j} \left(n p_{ij}(1-p_{ij})H_{ij}H_{ij}^\top\right) \right] \notag \\
    & + \mathbb U_{N,M} \left(n p_{ij}(1-p_{ij})H_{ij}H_{ij}^\top\right).
\end{align}

We have
\begin{align*}
&    \left\| (\mathbb P_{N}- \mathbb P) \left[ \mathbb P_{i,} \left(n p_{ij}(1-p_{ij})H_{ij}H_{ij}^\top\right) \right] \right\|_{op} \\
 &   \leq \left\| (\mathbb P_{N}- \mathbb P) \left[ \mathbb P_{i,} \left(n p_{ij}(1-p_{ij})H_{ij}H_{ij}^\top\right) \right] \right\|_{F} = O_P\left(\sqrt{ \frac{k_n^2 \zeta_n^2}{n} }\right)
\end{align*}
and 
\begin{align*}
&    \left\|(\mathbb P_{M}- \mathbb P) \left[ \mathbb P_{,j} \left(n p_{ij}(1-p_{ij})H_{ij}H_{ij}^\top\right) \right]\right\|_{op} \\
& \leq \left\|(\mathbb P_{M}- \mathbb P) \left[ \mathbb P_{,j} \left(n p_{ij}(1-p_{ij})H_{ij}H_{ij}^\top\right) \right]\right\|_{F} = O_P\left(\sqrt{ \frac{k_n^2 \zeta_n^2}{n} }\right).
\end{align*}
To analyze the last term in the RHS of \eqref{eq:V_hoeffding}, we note that 
\begin{align*}
    & \mathbb E\left\|    \mathbb U_{N,M}  \left(n p_{ij}(1-p_{ij})H_{ij}H_{ij}^\top\right) \right\|_{F}^2 \\
    & = \sum_{(\ell, \ell') \in [k_n] \times [k_n]} \mathbb E \left[\mathbb U_{N,M} \left(n p_{ij}(1-p_{ij})H_{ij,\ell}H_{ij,\ell'}\right) \right] ^2 \\
    & \leq \frac{C}{(NM)^2}    \sum_{(\ell, \ell') \in [k_n] \times [k_n]} \sum_{i = 1}^N \sum_{j=1}^M \mathbb E \left(n p_{ij}(1-p_{ij})H_{ij,\ell}H_{ij,\ell'}\right)^2 \leq \frac{ C k_n^2 \zeta_n^2}{NM}, 
\end{align*}
which implies 
\begin{align*}
& \left\|    \mathbb U_{N,M}  \left(n p_{ij}(1-p_{ij})H_{ij}H_{ij}^\top\right) \right\|_{op} \leq \left\|    \mathbb U_{N,M}  \left(n p_{ij}(1-p_{ij})H_{ij}H_{ij}^\top\right) \right\|_{F} = O_P\left(\frac{k_n \zeta_n}{n}\right).
\end{align*}

Therefore, following \eqref{eq:V_hoeffding}, we have
\begin{align}\label{eq:V_bd}
   \left\| n^{-1} V - \mathbb E \left( n p_{ij}(1-p_{ij})H_{ij}H_{ij}^\top \right) \right\|_{op} = O_P\left(\sqrt{ \frac{k_n^2 \zeta_n^2}{n} }\right).
\end{align}

Combining \eqref{eq:hat_s^2_1}, \eqref{eq:I_bd}, \eqref{eq:II_bd}, \eqref{eq:III_bd}, \eqref{eq:IV_bd}, and \eqref{eq:V_bd}, we have
\begin{align}\label{eq:hat_s^2_1B}
 & \left\|  \frac{1}{n NM}\sum_{i=1}^N \sum_{j=1}^M  s_{Y,ij}^2 H_{ij}H_{ij}^\top -  \mathbb E \left( n p_{ij}(1-p_{ij})H_{ij}H_{ij}^\top \right) \right\|_{op}  \notag \\
 & = \left\|  \frac{1}{n NM}\sum_{i=1}^N \sum_{j=1}^M  s_{Y,ij}^2 H_{ij}H_{ij}^\top -  n^{-1}  \mathbb E \left[ Var(s_{Y,ij}\mid X_i,A_i,W_j,B_j) H_{ij}H_{ij}^\top \right] \right\|_{op} =  O_P\left(\sqrt{ \frac{k_n^2 \zeta_n^2}{n} }\right).
\end{align}

For the second term on the RHS of \eqref{eq:hat_s^2}, we have
\begin{align}\label{eq:hat_s^2_2}
    & \left\| \frac{1}{NM}\sum_{i=1}^N \sum_{j=1}^M (\hat s_{Y,ij}+s_{Y,ij})(\hat s_{Y,ij}-s_{Y,ij}) H_{ij}H_{ij}^\top \right\|_{op} \notag \\
    & \leq \left\|\frac{2}{NM}\sum_{i=1}^N \sum_{j=1}^M n^2 Y_{ij} \exp(R_{ij}^\top \tilde \theta)  ||R_{ij}||_2 H_{ij}H_{ij}^\top\right\|_{op} ||\hat \theta_{\text{initial}} - \theta_{0,n}||_2 \notag \\
    & \leq \left\|\frac{1}{NM}\sum_{i=1}^N \sum_{j=1}^M n Y_{ij} H_{ij}H_{ij}^\top\right\|_{op} \times O_P(n^{-1/2}) \notag \\
    & \leq \left\|\frac{1}{NM}\sum_{i=1}^N \sum_{j=1}^M \left[n (Y_{ij}-p_{ij}) + n p_{ij} \right] H_{ij}H_{ij}^\top\right\|_{op} \times O_P(n^{-1/2}) \notag \\
    & \leq \left[ O_P\left(\sqrt{ \frac{k_n^2 \zeta_n^2}{n} }\right) + \left\|\frac{1}{NM}\sum_{i=1}^N \sum_{j=1}^M H_{ij}H_{ij}^\top\right\|_{op}\right] \times O_P(n^{1/2}) = O_P(n^{1/2}). 
\end{align}
where $\tilde \theta = (\tilde \alpha_n, \tilde \beta_n)$ is between $\hat \theta_{\text{initial}}$ and 0 so that $\log n + \tilde \alpha_n = O_P(1)$, $\tilde \beta_n = O_P(1)$, and 
\begin{align*}
    n \exp(R_{ij}^\top \tilde \theta) \leq \exp(|\log n + \tilde \alpha_n| + ||Z_{ij}||_2 ||\tilde \beta||_2) \leq \exp(|\log n + \tilde \alpha_n| + C ||\tilde \beta_n||_2), 
\end{align*}
and the last inequality is due to \eqref{eq:III_bd}. Further note that 
\begin{align*}
\left\|     \hat \Sigma^a_{YY,n} \right\|_{op} \leq \left\|   \Sigma^a_{YY,n} \right\|_{op} + \left\|     \hat \Sigma^a_{YY,n} - \Sigma^a_{YY,n} \right\|_{op} = O_P(1) 
\end{align*}
and, similarly, $\left\|     \hat \Sigma^p_{YY,n} \right\|_{op} = O_P(1)$.

Combining \eqref{eq:hat_s^2}, \eqref{eq:hat_s^2_1B}, \eqref{eq:hat_s^2_2}, and above fact, we obtain the desired result that 
\begin{align*}
& \left\|     \frac{1}{n}\left[ \frac{1}{NM}\sum_{i=1}^N \sum_{j=1}^M \hat s_{Y,ij}^2 H_{ij}H_{ij}^\top -  \hat \Sigma^a_{YY,n} - \hat \Sigma^p_{YY,n} \right] -   n^{-1}  \mathbb E \left[ Var(s_{Y,ij}\mid X_i,A_i,W_j,B_j) H_{ij}H_{ij}^\top \right] \right\|_{op} \\
& = O_P\left(\sqrt{ \frac{k_n^2 \zeta_n^2}{n} }\right). 
\end{align*}

We can show $\left\| \hat \Sigma_{H,\gamma} - \Sigma_{H,\gamma} \right\|_{2} = O_P\left(  \sqrt{\frac{k_n^2 \zeta_n^2}{n} } \right)$ in the same manner. 

\textbf{Proof of $| \hat \Sigma_{\gamma} - \Sigma_{\gamma} | = O_P(\zeta_n n^{-1/2})$.}
We note that 
\begin{align}
   \hat \Sigma_{\gamma} = \frac{\hat \Sigma^a_{\gamma\gamma,n}}{1-\phi} +  \frac{\hat \Sigma^p_{\gamma\gamma,n}}{\phi},
\end{align}
where $\hat \Sigma^a_{\gamma\gamma,n}$ and $\hat \Sigma^p_{\gamma\gamma,n}$ are the bottom-right elements of $\hat \Sigma^a_n$ and $\hat \Sigma^p_n$, respectively. Specifically, we have 
\begin{align*}
    \hat \Sigma^a_{\gamma\gamma,n} = \frac{1}{N} \sum_{i=1}^N \frac{2}{M(M-1)}\sum_{j=1}^{M-1} \sum_{k=j+1}^M \hat s_{\gamma,ij}\hat s_{\gamma,ik}
\end{align*}
and 
\begin{align*}
\hat \Sigma^p_{\gamma\gamma,n} = \frac{1}{M} \sum_{j=1}^M \frac{2}{N(N-1)}\sum_{i=1}^{N-1} \sum_{k=i+1}^N \hat s_{\gamma,ij}\hat s_{\gamma,ik}.
\end{align*}
Then, similar to the argument in the proof of Lemma \ref{lem:Sigma^a_YY} (with the dimension $k_n=1$), we can show that 
\begin{align*}
\left\|        \hat \Sigma^a_{\gamma\gamma,n} -     \Sigma^a_{\gamma\gamma,n} \right\|_{op} = O_P\left(\sqrt{\frac{\zeta_n^2}{n}} \right) \quad \text{and} \quad \left\|        \hat \Sigma^p_{\gamma\gamma,n} -     \Sigma^p_{\gamma\gamma,n} \right\|_{op} = O_P\left(\sqrt{\frac{\zeta_n^2}{n}} \right),
\end{align*}
where 
\begin{align*}
\Sigma^a_{\gamma\gamma,n}  = \mathbb E    [\mathbb E (s_{\gamma,ij}|X_i,A_i)]^2  \quad \text{and} \quad \Sigma^p_{\gamma\gamma,n}  = \mathbb E [\mathbb E (s_{\gamma,ij}|W_j,B_j)]^2.
\end{align*}
This is the desired result. 

\textbf{Proof of $\|\hat \kappa^* - \kappa^\ast \|_2 = o_P( k_n^{-1/2})$.} Based on the previous results, we have 
\begin{align*}
\left\|    \hat \Phi_1 -  \Phi_1\right\|_{op} \leq \overline \sigma^2 \left\|\hat B_H - B_H \right\|_{op} + \left\|\hat \Sigma_H - \Sigma_H \right\|_{op} = O_P\left(\sqrt{\frac{k_n^2 \zeta_n^2 \log^2(n)}{n}}\right),
\end{align*}
\begin{align*}
\left\|  \hat \Phi_2 - \Phi_2   \right\|_2 \leq  \overline \sigma^2 \left\|\hat B_{H,\gamma} - B_{H,\gamma} \right\|_{2} + \left\|\hat \Sigma_{H,\gamma} - \Sigma_{H,\gamma} \right\|_{2} = O_P\left(\sqrt{\frac{k_n^2 \zeta_n^2 \log^2(n)}{n}}\right), 
\end{align*}
\begin{align*}
\left\|    \hat G - G_0 \right\|_{op} = O_P\left(\sqrt{\frac{k_n^2 \zeta_n^2 \log^2(n)}{n}}\right), 
\end{align*}
and $||\hat \Gamma - \Gamma||_2 = O_P(n^{-1/2})$. By Assumption \ref{ass:Y_gamma}.2, we have 
\begin{align*}
C \geq \lambda_{\max}(\Phi_1) \geq \lambda_{\min}(\Phi_1)\geq c \quad \text{and} \quad   C \geq \lambda_{\max}(G^\top_0 \Phi_1^{-1} G_0 ) \geq \lambda_{\min}(G^\top_0 \Phi_1^{-1} G_0 )\geq c  
\end{align*}
for some constants $(C,c)$ such that $\infty>C\geq c>0$. In addition, we have $||G_0||_{op} = O(1)$, $||\Phi_2||_2 = O(1)$, and $||\Gamma||_2 = O(1)$.  

Now we are ready to bound $\|       \hat \kappa^*  - \kappa^\ast\|_2$. Note that  
\begin{align}\label{eq:kappa*}
& \|       \hat \kappa^*  - \kappa^\ast\|_2 \notag \\
& = \left\|\hat \Phi_1^{-1} \left\{\hat \Phi_2 - \hat G \left[\hat G^\top \hat \Phi_1^{-1} \hat G \right]^{-1} \left[ \hat G^\top \hat \Phi_1^{-1}\hat \Phi_2+ \hat \Gamma \right]  \right\} - \Phi_1^{-1} \left\{ \Phi_2 -  G_0 \left[ G_0^\top  \Phi_1^{-1}  G_0 \right]^{-1} \left[ G_0^\top \Phi_1^{-1} \Phi_2+  \Gamma \right]  \right\} \right\|_2  \notag \\
& \leq \left\|\hat \Phi_1^{-1} \hat \Phi_2 - \Phi_1^{-1} \Phi_2 \right\|_2 + \left\|\hat \Phi_1^{-1} \hat G \left[\hat G^\top \hat \Phi_1^{-1} \hat G \right]^{-1}  \hat G^\top \hat \Phi_1^{-1}\hat \Phi_2  - \Phi_1^{-1}  G_0 \left[  G_0^\top \Phi_1^{-1} G_0 \right]^{-1}  G_0^\top \Phi_1^{-1} \Phi_2 \right\|_2  \notag \\
& + \left\|\hat \Phi_1^{-1}\hat G \left[\hat G^\top \hat \Phi_1^{-1} \hat G \right]^{-1}\hat \Gamma -  \Phi_1^{-1} G_0 \left[ G_0^\top \Phi_1^{-1} G_0 \right]^{-1} \Gamma   \right\|_2  \notag \\
& \equiv I+II+III. 
\end{align}

For term $I$ on the RHS of \eqref{eq:kappa*}, we have
\begin{align*}
\left\|\hat \Phi_1^{-1} \hat \Phi_2 - \Phi_1^{-1} \Phi_2 \right\|_2 & \leq \left\|\hat \Phi_1^{-1} (\hat \Phi_2 - \Phi_2) \right\|_2  + \left\|(\hat \Phi_1^{-1} -  \Phi_1^{-1}) \Phi_2 \right\|_2    \\
& \leq \left\|\hat \Phi_1^{-1} \right\|_{op} \left\| \hat \Phi_2 - \Phi_2\right\|_2 + \left\|\hat \Phi_1^{-1} -  \Phi_1^{-1}\right\|_{op} \left\|\Phi_2 \right\|_2 = O_P\left(\sqrt{\frac{k_n^2 \zeta_n^2 \log^2(n)}{n}}\right). 
\end{align*}

For term $II$ on the RHS of \eqref{eq:kappa*}, we have
\begin{align*}
&     \left\|\hat \Phi_1^{-1} \hat G \left[\hat G^\top \hat \Phi_1^{-1} \hat G \right]^{-1}  \hat G^\top \hat \Phi_1^{-1}\hat \Phi_2  - \Phi_1^{-1}  G_0 \left[  G_0^\top \Phi_1^{-1} G_0 \right]^{-1}  G_0^\top \Phi_1^{-1} \Phi_2 \right\|_2 \\
& \leq \left\| \hat \Phi_1^{-1} \hat G - \Phi_1^{-1} G_0 \right\|_{op} \left\| \left[\hat G^\top \hat \Phi_1^{-1} \hat G \right]^{-1}  \hat G^\top \hat \Phi_1^{-1}\hat \Phi_2 \right\|_2 + \left\|\Phi_1^{-1}  G_0 \right\|_{op} \left\|\left[\hat G^\top \hat \Phi_1^{-1} \hat G \right]^{-1} -  \left[  G_0^\top \Phi_1^{-1} G_0 \right]^{-1} \right\| ||\hat \Phi_2||_2\\
& + \left\| \Phi_1^{-1}  G_0 \left[  G_0^\top \Phi_1^{-1} G_0 \right]^{-1}  G_0^\top \Phi_1^{-1} \right\|_{op} \left\|\hat \Phi_2 - \Phi_2 \right\|_2 \\
& = O_P\left(\sqrt{\frac{k_n^2 \zeta_n^2 \log^2(n)}{n}}\right). 
\end{align*}
Similarly, we have $||III||_2 = O_P\left(\sqrt{\frac{k_n^2 \zeta_n^2 \log^2(n)}{n}}\right)$, which implies 
\begin{align*}
    \|\hat \kappa^* - \kappa^\ast \|_2 = O_P\left(\sqrt{\frac{k_n^2 \zeta_n^2 \log^2(n)}{n}}\right)= o_P( k_n^{-1/2}),
\end{align*}
where the last equality holds by Assumption \ref{ass:kn}. 

\textbf{Proof of $\|\kappa^\ast \|_2 = O(1)$.} 
Note that 
\begin{align*}
  \|\kappa^\ast \|_2 & =  \left\| \Phi_1^{-1} \left\{ \Phi_2 -  G_0 \left[ G_0^\top  \Phi_1^{-1}  G_0 \right]^{-1} \left[ G_0^\top \Phi_1^{-1} \Phi_2+  \Gamma \right]  \right\}  \right\|_2 \\
  & \leq \left\| \Phi_1^{-1}  \Phi_2 \right\|_2 + \left\| \Phi_1^{-1}  G_0 \left[ G_0^\top  \Phi_1^{-1}  G_0 \right]^{-1} G_0^\top \Phi_1^{-1} \Phi_2   \right\|_2  \\
  & + \left\| \Phi_1^{-1}  G_0 \left[ G_0^\top  \Phi_1^{-1}  G_0 \right]^{-1} \Gamma  \right\|_2 \\
  & \leq \left\| \Phi_1^{-1}\right\|_{op} \left\|\Phi_2 \right\|_2 +  \left\| \Phi_1^{-1/2}\right\|_{op}^2 \left\|\Phi_2 \right\|_2 + \left\| \Phi_1^{-1}\right\|_{op} \left\| G_0 \right\|_{op} \left\| \left[ G_0^\top  \Phi_1^{-1}  G_0 \right]^{-1}\right\|_{op} \left\|\Gamma \right\|_2 \\
  & = O(1). 
\end{align*}
This concludes the proof.

\subsection{Proof of Lemma \ref{lem:strong_approx}}
Recall 
	\begin{align*}
		s_{ij} = & \begin{pmatrix}
			& n\left(Y_{ij} - \int \Lambda(R_{ij}^\top \theta_{0,n}+n^{-1/2}v)\pi_{ij}(v)dv\right)H_{ij} \\
			& \gamma_n(D_{ij},0,\theta_{0,n}) -  \mathbb{E}\gamma_n(D_{ij},0, \theta_{0,n}) 
		\end{pmatrix} \equiv (s_{Y,ij} H_{ij}^\top,s_{\gamma,ij})^\top, 
	\end{align*}
	$\overline s_{ij} = \mathbb{E}(s_{ij}|X_i,A_i,W_j,B_j)$,  $\overline s_{1i}^a = \mathbb{E}(\overline s_{ij}|X_i,A_i)$, $\overline s_{1j}^p = \mathbb{E}(\overline s_{ij}|W_j,B_j)$, $\Sigma_{n}^a = \mathbb{E} \overline s_{1i}^a(\overline s_{1i}^a)^\top$, $\Sigma_{n}^p = \mathbb{E} \overline s_{1j}^p(\overline s_{1j}^p)^\top$, and $\Sigma_{2n} = \mathbb{E}(s_{ij} - \overline s_{ij})(s_{ij} - \overline s_{ij})^\top$. 

 Following \cite{graham:2022}, we have
 \begin{align*}
     \frac{\sqrt{n} }{NM}\sum_{i=1}^N \sum_{j = 1}^M s_{ij} & = \frac{\sqrt{n} }{N}\sum_{i=1}^N\overline{s}^a_{1i} +  \frac{\sqrt{n} }{M}\sum_{j=1}^M\overline{s}^p_{1j} +       \frac{\sqrt{n} }{NM}\sum_{i=1}^N \sum_{j = 1}^M (s_{ij} - \overline{s}_{ij}) +     \frac{  \sqrt{n}}{NM}\sum_{i=1}^N \sum_{j = 1}^M (\overline{s}_{ij} -  \overline{s}^a_{1i} -\overline{s}^p_{1j}) \\
     & \equiv  \sum_{l=1}^4 U_{l,n},
 \end{align*}
where 
\begin{align*}
U_{1,n} & = \sum_{i \in [N]} \frac{\sqrt{n}}{N}\overline{s}^a_{1i}, \\
U_{2,n} & = \sum_{j \in [M]} \frac{\sqrt{n}}{M}\overline{s}^p_{1j}, \\
U_{3,n} & =   \frac{\sqrt{n} }{NM}\sum_{i=1}^N \sum_{j = 1}^M (s_{ij} - \overline{s}_{ij}), \\
U_{4,n} & = \frac{  \sqrt{n}}{NM}\sum_{i=1}^N \sum_{j = 1}^M (\overline{s}_{ij} -  \overline{s}^a_{1i} -\overline{s}^p_{1j}).  
\end{align*}
Let $\mathcal F_n$ be the sigma field generated by $\{X_i,A_i\}_{i \in [N]}$ and $\{W_j,B_j\}_{j \in [M]}$. Then, we note that $(U_{1,n},U_{2,n})$ belong to $\mathcal F_n$, and conditionally on $\mathcal F_n$, $U_{3,n} = (U_{3,n}^{*,\top},0)^\top$, 
\begin{align*}
 U_{3,n}^* = \frac{\sqrt{n}}{NM}\sum_{i=1}^N \sum_{j=1}^M  n\left(Y_{ij} - \mathbb E (Y_{i,j}|\mathcal F_n)\right)  H_{ij} = \frac{\sqrt{n}}{NM}\sum_{i=1}^N \sum_{j=1}^M  n\left(Y_{ij} - p_{ij}\right)  H_{ij}
\end{align*}
is a sequence of independent random variables. Therefore, by Yurinskii’s Coupling (\citet[Theorem 10]{P02}), there exists $g_{3,n} = (g_{3,n}^{*,\top},0)^\top$ such that, for $g^*_{3,n}  \stackrel{d}{=}\N(0,I_{k_n})$, $g^*_{3,n}$ is independent of $\mathcal F_n$ and for any $\mathcal F_n$-measurable random variable $\Upsilon_n$ with $\Upsilon_n>0$, we have 
\begin{align*}
 \mathbb P \left\{ || U^*_{3,n} - \hat \Omega_{3,n}^{1/2}g^*_{3,n}||_2 >3 \Upsilon_n \mid \mathcal F_n \right\} \leq C_0 \hat B_{3,n}\left(1 + \frac{|\log (1/\hat B_{3,n})|}{k_n} \right),
\end{align*}
where $C_0$ is a universal constant, 
\begin{align*}
\hat    \Omega_{3,n} = \frac{1}{NM}\sum_{i = 1}^N \sum_{j=1}^M \frac{n^3 p_{ij}(1-p_{ij})H_{ij}H_{ij}^\top}{NM} ,
\end{align*}
\begin{align*}
\hat     B_{3,n} = \hat D_{3,n} k_n \Upsilon_n^{-3}, \quad \text{and} \quad \hat D_{3,n} = \frac{n^{3/2}}{(NM)^3}\sum_{i \in [N]} \sum_{j \in [M]} \mathbb E\left(  ||s_{ij} - \overline{s}_{ij}||_2^3 \mid \mathcal F_n \right).
\end{align*}
Consequently, by choosing $\Upsilon_n = \hat D_{3,n}^{1/3}\delta_n^{-1/3}$ for some deterministic sequence $\delta_n>0$ which will be specified later, we have 
\begin{align*}
    \hat B_{3,n} = k_n \delta_n
\end{align*}
and
\begin{align*}
 \mathbb P\left\{ || U^*_{3,n} - \hat \Omega_{3,n}^{1/2}g^*_{3,n}||_2 >3  D_{3,n}^{1/3}\delta_n^{-1/3} \mid \mathcal F_n \right\} \leq C_0 k_n \delta_n \left(1 + \frac{|\log (1/(k_n \delta_n))|}{k_n} \right).
\end{align*}
We further note that $\hat D_{3,n} = O_P\left(\sqrt{\frac{k_n^3}{n}}\right)$, and $\hat    \Omega_{3,n}$ and $\hat     B_{3,n}$ belong to $\mathcal F_n$. In addition, let $\Omega_{3,n} = \mathbb E \hat \Omega_{3,n}$. Then, we have $\left\|\Omega_{3,n} - \hat \Omega_{3,n}\right\|_{op} = O_P\left(  \sqrt{\frac{k_n^2 \zeta_n^2}{n}} \right)$ by \eqref{eq:V_bd}, and thus, 
\begin{align*}
\left\| \left(   \hat \Omega_{3,n}^{1/2} - \Omega_{3,n}^{1/2} \right) g_{3,n}^*\right\|_2 \leq \left\|   \hat \Omega_{3,n}^{1/2} - \Omega_{3,n}^{1/2}\right\|_{op} ||g^*_{3,n}||_2 = O_P\left(  \sqrt{\frac{ k^3_n \zeta_n^2}{n}} \right),
\end{align*}
where we use the fact that $\lambda_{\min}\left(\Omega_{3,n} \right) \geq c>0$ for some constant $c$ and $||g^*_{3,n}|| = O_P(\sqrt{k_n})$.

Similarly, we can strongly approximate $U_{1,n}$ and $U_{2,n}$ by $\Omega_{1,n}^{1/2} g_{1,n}$ and $\Omega_{2,n}^{1/2} g_{2,n}$, respectively, where $g_{1,n}$ and $g_{2,n}$ are $(k_n+1)$-dimensional standard normal random variables, 
\begin{align*}
    \Omega_{1,n} & = \frac{n}{N^2}\sum_{i=1}^N \mathbb E (\overline{s}^a_{1i})^2, \quad       \Omega_{2,n}  = \frac{n}{M^2}\sum_{j=1}^M \mathbb E (\overline{s}^p_{1j})^2.
\end{align*}
More specifically, for $l = 1,2$, we have
\begin{align*}
    \mathbb P\left(\left\|U_{l,n} - \Omega_{l,n}^{1/2}g_{l,n}\right\|_2 > 3 D_{l,n}^{1/3} \delta_n \right) \leq C_0 k_n\delta_n\left(1 + \frac{|\log (1/(k_n \delta_n))|}{k_n} \right),
\end{align*}
where 
\begin{align*}
    D_{1,n} =  \frac{n^{3/2}}{N^3}\sum_{i \in [N]} \mathbb E\left(  ||\overline s^a_{1i} ||_2^3 \right) = O\left(\sqrt{ \frac{k_n^3}{n}} \right) \quad \text{and} \quad  D_{2,n} =  \frac{n^{3/2}}{M^3}\sum_{j \in [M]} \mathbb E\left(  ||\overline s^p_{1j} ||_2^3 \right) = O\left(\sqrt{ \frac{k_n^3}{n}} \right).
\end{align*}
In addition, because $(X_i,A_i)_{i \in [N]}$ and $(W_j,B_j)_{j \in [M]}$ are independent, we have $g_{1,n} \perp g_{2,n}$. Last, because $g_{3,n}^*$ is independent of $\mathcal F_n$, we have $(g_{1,n}, g_{2,n}, g_{3,n}^*)$ are independent. 

Next, we turn to $U_{4,n}$. We note that 
\begin{align*}
    ||U_{4,n}||_2 \leq \sqrt{k_n} ||U_{4,n}||_\infty. 
\end{align*}
In addition, by \citet[Corollary 3]{chiang/kato/sasaki:2023} with their $k=2$, $q=2$, and $e = \{1,1\}$, we have
\begin{align*}
  \frac{NM}{\sqrt{n}} \left(  \mathbb E \left[ ||U_{4,n}||_\infty^2 \right]\right)^{1/2} & \leq C (\log(k_n)) \sqrt{NM} \left(\mathbb E ||\overline{s}_{ij} -  \overline{s}^a_{1i} -\overline{s}^p_{1j}||_\infty^2 \right)^{1/2}  \\
  & \leq C (\log(k_n)) \sqrt{NM} \left(\mathbb E ||\overline{s}_{ij} -  \overline{s}^a_{1i} -\overline{s}^p_{1j}||^2_2 \right)^{1/2} \\
  & \leq C \log(k_n) \sqrt{k_n} \sqrt{NM}.
  \end{align*}
This implies 
\begin{align*}
||U_{4,n}||_2 = O_P\left(\sqrt{\frac{\log(n) k_n^2}{n}}\right).
\end{align*}

Therefore, by letting $\delta_n = (Ck_n)^{-1}$ for a large constant $C>0$, we have 
\begin{align*}
\left\|   \frac{\sqrt{n} }{NM}\sum_{i=1}^N \sum_{j = 1}^M s_{ij} - \left( \Omega_{1,n}^{1/2}g_{1,n} + \Omega_{2,n}^{1/2}g_{2,n} + \begin{pmatrix}
       \Omega_{3,n}^{*,1/2}g_{3,n}^* \\
       0
   \end{pmatrix}\right) \right\| & = O_P\left(\sqrt{\frac{\log(n) k_n^2}{n}} +  \sqrt{\frac{ k^3_n \zeta_n^2}{n}} + \left(\frac{k_n^3}{n}\right)^{1/6}k_n^{1/3}\right) \\
   & = o_P(1).
\end{align*}
In addition, we have 
\begin{align*}
g_n \equiv \Omega_{1,n}^{1/2}g_{1,n} + \Omega_{2,n}^{1/2}g_{2,n} + \begin{pmatrix}
       \Omega_{3,n}^{*,1/2}g_{3,n}^* \\
       0
   \end{pmatrix} \sim \N\left(0, \Omega_{1,n} + \Omega_{2,n} + \begin{pmatrix}
       \Omega_{3,n}^* & 0 \\
       0 & 0
   \end{pmatrix}\right) \stackrel{d}{=} \N(0,\Sigma_n).
\end{align*}
This concludes the proof. 

\section{Additional Lemmas}
\begin{lemma}\label{lem:Sigma^a_YY}
    Suppose conditions in Theorem \ref{theo:size} holds. Then, we have
\begin{align*}
    || \hat \Sigma^a_{YY,n} -  \Sigma^a_{YY,n}||_{op} = O_P\left(  \sqrt{\frac{k_n^2 \zeta_n^2}{n} } \right)
\end{align*}
and 
\begin{align*}
    || \hat \Sigma^p_{YY,n} -  \Sigma^p_{YY,n}||_{op} = O_P\left(  \sqrt{\frac{k_n^2 \zeta_n^2}{n} } \right).
\end{align*}
\end{lemma}
\begin{proof}
We have
\begin{align}\label{eq:Sigma_YY}
& \left| \hat \Sigma^a_{YY,n} - \mathbb E \left\{ [\mathbb E (s_{Y,ij}H_{ij}|X_i,A_i)][\mathbb E (s_{Y,ij}H_{ij}^\top|X_i,A_i)]	\right\} \right| \notag \\
& \leq  \left\| \frac{1}{N}\sum_{i=1}^N \frac{2}{M(M-1)}\sum_{j=1}^{M-1} \sum_{k=j+1}^M(\hat s_{Y,ij}\hat s_{Y,ik} - s_{Y,ij} s_{Y,ik}) H_{ik} H_{ij}^\top \right\|_{op} \notag \\
& + \left\| \frac{1}{N}\sum_{i=1}^N \frac{2}{M(M-1)}\sum_{j=1}^{M-1} \sum_{k=j+1}^M s_{Y,ij} s_{Y,ik} H_{ij} H_{ij}^\top  - \mathbb E \left\{ [\mathbb E (s_{Y,ij}H_{ij}|X_i,A_i)] [\mathbb E (s_{Y,ij}H_{ij}^\top |X_i,A_i)]	\right\} \right\|_{op}\notag \\
& \leq \left\| \frac{1}{N}\sum_{i=1}^N \frac{2}{M(M-1)}\sum_{j=1}^{M-1} \sum_{k=j+1}^M(\hat s_{Y,ij}\hat s_{Y,ik} - s_{Y,ij} s_{Y,ik}) H_{ik} H_{ij}^\top \right\|_{op}  \notag \\
& + \left\|  \frac{2}{M(M-1)}\sum_{j=1}^{M-1} \sum_{k=j+1}^M \frac{1}{N}\sum_{i=1}^N\left\{s_{Y,ij} s_{Y,ik} H_{ij}H_{ik}^\top  -\left\{ [\mathbb E (s_{Y,ij}H_{ij}|X_i,A_i)] [\mathbb E (s_{Y,ij}H_{ij}^\top|X_i,A_i)] \right\}	\right\} \right\|_{op} \notag \\
& + \left\| \frac{1}{N}\sum_{i=1}^N  \left\{ [\mathbb E (s_{Y,ij}H_{ij}|X_i,A_i)] [\mathbb E (s_{Y,ij}H_{ij}^\top |X_i,A_i)] \right\}	-  \mathbb E \left\{ [\mathbb E (s_{Y,ij}H_{ij}|X_i,A_i)] [\mathbb E (s_{Y,ij}H_{ij}^\top |X_i,A_i)]	\right\} \right\|_{op} \notag \\
& \equiv ||I||_{op} + ||II||_{op} + ||III||_{op}. 
\end{align}

\textbf{Bound for $||I||_{op}$.}
We will establish that 
\begin{align}\label{eq:I_final}
    ||I||_{op} \leq C \left\| \frac{1}{N}\sum_{i=1}^N \frac{2}{M(M-1)}\sum_{j=1}^{M-1} \sum_{k=j+1}^M H_{ik}H_{ij}\right\|_{op} ||\hat \theta_{\text{initial}}   - \theta_{0,n}||_2 = O_P(n^{-1/2}).
\end{align}
First, we note that 
\begin{align}
    ||I||_{op} & = \sup_{u,v \in \Re^{k_n}, ||u||_2 = ||v||_2 = 1}\left| \frac{1}{N}\sum_{i=1}^N \frac{2}{M(M-1)}\sum_{j=1}^{M-1} \sum_{k=j+1}^M(\hat s_{Y,ij} - s_{Y,ij}) \hat s_{Y,ik} u^\top H_{ik} H_{ij}^\top v \right| \notag \\
    & + \sup_{u,v \in \Re^{k_n}, ||u||_2 = ||v||_2 = 1}\left| \frac{1}{N}\sum_{i=1}^N \frac{2}{M(M-1)}\sum_{j=1}^{M-1} \sum_{k=j+1}^M(\hat s_{Y,ij} - s_{Y,ij}) s_{Y,ik} u^\top H_{ik} H_{ij}^\top v \right| \notag \\
    & \leq \frac{C}{N}\sum_{i=1}^N \frac{2}{M(M-1)}\sum_{j=1}^{M-1} \sum_{k=j+1}^M n(Y_{ik}+\Lambda(R^\top_{ik}\hat \theta_{\text{initial}}) ||H_{ik}||_2 ||H_{ij}||_2 ||\hat \theta_{\text{initial}} - \theta_{0,n} ||_2  \notag \\
    & + \frac{C}{N}\sum_{i=1}^N \frac{2}{M(M-1)}\sum_{j=1}^{M-1} \sum_{k=j+1}^M n(Y_{ik}+\Lambda(R^\top_{i,k} \theta_{0,n}) ||H_{ik}||_2 ||H_{ij}||_2 ||\hat \theta_{\text{initial}} - \theta_{0,n} ||_2 \notag \\
    & \leq \frac{C}{N}\sum_{i=1}^N \frac{2}{M(M-1)}\sum_{j=1}^{M-1} \sum_{k=j+1}^M n(Y_{ik} - p_{ik}) ||H_{ik}||_2 ||H_{ij}||_2 \times O_P(n^{-1/2}) \notag \\
    & + \frac{1}{N}\sum_{i=1}^N \frac{2}{M(M-1)}\sum_{j=1}^{M-1} \sum_{k=j+1}^M ||H_{ik}||_2 ||H_{ij}||_2 \times O_P(n^{-1/2}),    
\end{align}

where we use the fact that $n p_{ij}$ is bounded, $ \max_{i,j} n \Lambda(R_{ij}^\top \hat \theta_{\text{initial}}) = O_P(1)$, and 
\begin{align*}
 |\hat s_{Y,ij} - S_{Y,ij}| = n |\Lambda(R_{ij}^\top \hat \theta_{\text{initial}}) - \Lambda(R_{ij}^\top \theta_{0,n})| \leq C||\hat \theta_{\text{initial}}   - \theta_{0,n}||_2  . 
\end{align*}

Then, we have
\begin{align*}
& \mathbb E    \left| \frac{1}{N}\sum_{i=1}^N \frac{2}{M(M-1)}\sum_{j=1}^{M-1} \sum_{k=j+1}^M n(Y_{ik} - p_{ik}) ||H_{ik}||_2 ||H_{ij}||_2\right|^2 \\
& \leq \mathbb E \left| \frac{1}{NM(M-1)}\sum_{i=1}^N \sum_{j=1}^{M-1}  \left[n (Y_{ij}-p_{ij}) ||H_{ij}||_2\right]  \left[ \sum_{k=j+1}^M ||H_{ik}||_2 \right]\right|^2 \\
 & = \frac{1}{N^2M^2(M-1)^2}\sum_{i=1}^N \sum_{j=1}^{M-1}  \mathbb E  \left[n (Y_{ij}-p_{ij}) ||H_{ij,l}||_2\right]^2  \left[ \sum_{k=j+1}^M ||H_{ik}||_2 \right]^2  \\
 & \leq \frac{C n M^2 k_n \zeta_n^2}{N^2M^2(M-1)^2} \sum_{i=1}^N \sum_{j=1}^{M-1}  ||H_{ij}||_2^2 = O_P\left( \frac{k_n^2 \zeta_n^2}{n} \right). 
\end{align*}
This implies
\begin{align*}
    \frac{1}{N}\sum_{i=1}^N \frac{2}{M(M-1)}\sum_{j=1}^{M-1} \sum_{k=j+1}^M n(Y_{ik} - p_{ik}) ||H_{ik}||_2 ||H_{ij}||_2 =  O_P\left(\sqrt{ \frac{k_n^2 \zeta_n^2}{n}} \right).
\end{align*}
In addition, we have 
\begin{align*}
     \frac{1}{N}\sum_{i=1}^N \frac{2}{M(M-1)}\sum_{j=1}^{M-1} \sum_{k=j+1}^M ||H_{ik}||_2 ||H_{ij}||_2  \leq     \frac{ \sqrt{k_n} \zeta_n }{NM}\sum_{i=1}^N \sum_{j=1}^{M-1} ||H_{ij}||_2 \leq O_P(k_n \zeta_n). 
\end{align*}

Therefore, given that $||\hat \theta_{\text{initial}}   - \theta_{0,n}||_2 = O_P(n^{-1/2})$, we have $||I||_{op} = O_P\left( \sqrt{\frac{k_n^2 \zeta_n^2}{n}} \right)$. 

\textbf{Bound for $||II||_{op}$.} We first define $\mathcal D_j = \{Y_{1j},\cdots,Y_{Nj},W_j,B_j\}$ and make our argument conditional on $\mathcal X_N $ and $\mathcal A_N$. In this view, 
\begin{align*}
  &  \frac{2}{M(M-1)}\sum_{j=1}^{M-1} \sum_{k=j+1}^M \frac{1}{N}\sum_{i=1}^N\left\{ s_{Y,ij} s_{Y,ik} H_{ij}H_{ik}^\top  -\left\{ [\mathbb E (s_{Y,ij}H_{ij}|X_i,A_i)] [\mathbb E (s_{Y,ij}H_{ij}^\top |X_i,A_i)] \right\} \right\} \\
  & \equiv \frac{1}{M(M-1)} \sum_{j,k \in [M], j \neq k}f(\mathcal{D}_j, \mathcal D_k),
\end{align*}
where 
\begin{align*}
f(\mathcal{D}_j, \mathcal D_k) =    \frac{1}{N}\sum_{i=1}^N \left\{s_{Y,ij} s_{Y,ik} H_{ij}H_{ik}^\top  - [\mathbb E (s_{Y,ij}H_{ij}|X_i,A_i)][\mathbb E (s_{Y,ij}H_{ij}^\top|X_i,A_i)]\right\}. 
\end{align*}

By Hoeffding decomposition, we have
\begin{align}\label{eq:hoeffding}
   II & = \frac{1}{M(M-1)} \sum_{j,k \in [M], j \neq k}f(\mathcal{D}_j, \mathcal D_k) \notag \\
    & = \frac{2}{M} \sum_{j =1}^M \mathbb E\left[ f(\mathcal{D}_j, \mathcal D_k)|\mathcal D_j, \mathcal X_N, \mathcal A_N \right] \notag \\
    & + \frac{1}{M(M-1)} \sum_{j,k \in [M], j \neq k}\left\{f(\mathcal{D}_j, \mathcal D_k) - \mathbb E\left[ f(\mathcal{D}_j, \mathcal D_k;u)|\mathcal D_k,\mathcal X_N, \mathcal A_N\right] - \mathbb E\left[ f(\mathcal{D}_j, \mathcal D_k;u)|\mathcal D_j, \mathcal X_N, \mathcal A_N\right]\right\} \notag \\
    & \equiv 2II_1 + II_2.
\end{align}

\textbf{Bound for $ ||II_1||_{op}$.} For the linear term $II_1$ on the RHS of \eqref{eq:hoeffding}, we have
\begin{align}\label{eq:II_1}
   II_1 & = \frac{1}{M} \sum_{j =1}^M \mathbb E\left[ f(\mathcal{D}_j, \mathcal D_k)|\mathcal D_j, \mathcal X_N, \mathcal A_N \right] \notag \\
   & =  \frac{1}{M N} \sum_{j =1}^M\sum_{i=1}^N \left[ s_{Y,ij} H_{ij}  - \mathbb E (s_{Y,ij} H_{ij}| X_i,A_i)\right]\mathbb E (s_{Y,ij} H_{ij}^\top | X_i,A_i) \notag \\
   & =  \frac{1}{M N} \sum_{j =1}^M\sum_{i=1}^N ( s_{Y,ij} - \overline s_{Y,ij}) H_{ij} \mathbb E (s_{Y,ij} H_{ij}^\top | X_i,A_i) \notag \\
   & +  \frac{1}{M N} \sum_{j =1}^M\sum_{i=1}^N \left[ \overline s_{Y,ij} H_{ij} - \mathbb E (s_{Y,ij} H_{ij}| X_i,A_i)\right]\mathbb E (s_{Y,ij} H_{ij}^\top | X_i,A_i) \notag \\
   & \equiv II_{1,1} + II_{1,2}.
\end{align}

Conditionally on $(\mathcal X_N, \mathcal A_N, \mathcal W_M, \mathcal B_M)$, $\{( s_{Y,ij} - \overline s_{Y,ij}) H_{ij}  \mathbb E (s_{Y,ij} H_{ij}^\top| X_i,A_i)\}_{i \in [N], j \in [M]}$ is a sequence of independent and mean-zero random variables. Therefore, we have
\begin{align*}
&    \mathbb E \left(||II_{1,1}||_F^2 \mid \mathcal X_N, \mathcal A_N, \mathcal W_M, \mathcal B_M\right) \\
& \leq \frac{1}{(MN)^2}\sum_{(\ell,\ell') \in [k_n] \times [k_n]} \sum_{i=1}^N \sum_{j=1}^M n^2 p_{ij}(1-p_{ij}) H_{ij,\ell}^2  \left[\mathbb E (s_{Y,ij} H_{ij,\ell'}| X_i,A_i) \right]^2 \\
& \leq C \frac{k_n n \zeta_n^2}{N M} \left[\sum_{\ell=1}^{k_n} \sum_{i=1}^N \sum_{j=1}^M \frac{1}{NM} H_{ij,\ell}^2\right] = O_P\left(\frac{k_n^2 \zeta_n^2}{n}\right),
\end{align*}
which implies 
\begin{align*}
    || II_{1,1}||_{op} =  O_P\left(\sqrt{\frac{k_n^2 \zeta_n^2}{n}}\right).
\end{align*}

For $II_{1,2}$, we have
\begin{align*}
    & \mathbb E \left( ||II_{1,2}||_F^2 \mid \mathcal X_N, \mathcal A_N \right) \\
    & =  \frac{1}{M^2}\sum_{(\ell,\ell') \in [k_n] \times [k_n]}  \sum_{j=1}^M \mathbb E \left(\left\{\frac{1}{N}\sum_{i=1}^N \left[ \overline s_{Y,ij} H_{ij,\ell} - \mathbb E (s_{Y,ij} H_{ij,\ell}| X_i,A_i)\right] \left[\mathbb E (s_{Y,ij} H_{ij,\ell'} | X_i,A_i) \right] \right\}^2 \mid \mathcal X_N, \mathcal A_N \right) \\
    & \leq \frac{1}{M^2}\sum_{(\ell,\ell') \in [k_n] \times [k_n]}  \sum_{j=1}^M \mathbb E \left(\frac{1}{N}\sum_{i=1}^N \left\{\left[ \overline s_{Y,ij} H_{ij,\ell} - \mathbb E (s_{Y,ij} H_{ij,\ell}| X_i,A_i)\right] \left[\mathbb E (s_{Y,ij} H_{ij,\ell'} | X_i,A_i) \right] \right\}^2 \mid \mathcal X_N, \mathcal A_N \right),
\end{align*}
and thus, 
\begin{align*}
    & \mathbb E \left( ||II_{1,2}||_F^2 \right) \\
  & \leq \frac{1}{N M^2}\sum_{(\ell,\ell') \in [k_n] \times [k_n]} \sum_{i=1}^N \sum_{j=1}^M \mathbb E  \left\{\left[ \overline s_{Y,ij} H_{ij,\ell} - \mathbb E (s_{Y,ij} H_{ij,\ell}| X_i,A_i)\right] \left[\mathbb E (s_{Y,ij} H_{ij,\ell'} | X_i,A_i) \right] \right\}^2 \\   
    & =O_P\left(\frac{k_n^2 \zeta_n^2}{n} \right),
\end{align*}
where use the fact that $|\overline s_{Y,ij}| \leq C<\infty$.

This implies $||II_{1,2}||_{op} =  O_P\left(\sqrt{\frac{k_n^2 \zeta_n^2}{n}}\right)||$, and thus, 
\begin{align}\label{eq:II1_final}
   || II_1||_{op} = O_P\left(  \sqrt{\frac{k_n^2 \zeta_n^2}{n}}\right).
\end{align}

\textbf{Bound for $||II_2||_{op}$.} Next, we turn to $II_2(u)$ in \eqref{eq:hoeffding}. We note that, conditionally on $(\mathcal X_N, \mathcal A_N)$, $II_2$ is a second-order degenerate U-statistic (based on observations $(\mathcal D_j,\mathcal D_k)$) which has an increasing dimension of $k_n \times k_n$. Let 
\begin{align*}
    f_{\ell,\ell'}(\mathcal{D}_j, \mathcal D_k) =    \frac{1}{N}\sum_{i=1}^N \left\{s_{Y,ij} s_{Y,ik} H_{ij,\ell}H_{ik,\ell'}  - [\mathbb E (s_{Y,ij}H_{ij,\ell}|X_i,A_i)][\mathbb E (s_{Y,ij}H_{ij,\ell'}|X_i,A_i)]\right\}. 
\end{align*} 
We have
\begin{align*}
\mathbb E \left( ||II_2||_F^2 \mid \mathcal X_N, \mathcal A_N\right) & \leq \frac{1}{M^4} \sum_{(\ell,\ell') \in [k_n] \times [k_n]}\sum_{(j,k) \in [M] \times [M], j\neq k} \mathbb E \left( f_{\ell,\ell'}^2(\mathcal D_j, \mathcal D_k)\mid  \mathcal X_N, \mathcal A_N \right). 
\end{align*}

Further note that $ \{s_{Y,ij} s_{Y,ik} H_{ij,\ell}H_{ik,\ell'}\}_{i \in [N]}$ are independent conditionally on $(\mathcal X_N, \mathcal A_N, \mathcal W_M, \mathcal B_M)$ and 
\begin{align*}
    \mathbb E\left[ s_{Y,ij} s_{Y,ik} H_{ij,\ell}H_{ik,\ell'} \mid \mathcal X_N, \mathcal A_N, \mathcal W_M, \mathcal B_M\right] = n^2 p_{ij}p_{ik}H_{ij,\ell}H_{ik,\ell'}.
\end{align*}

Then, we have
\begin{align*}
     f_{\ell,\ell'}^2(\mathcal{D}_j, \mathcal D_k) & \leq 2\left[  \frac{1}{N}\sum_{i=1}^N(s_{Y,ij} s_{Y,ik} - n^2p_{ij}p_{ik} )H_{ij,\ell}H_{ik,\ell'}\right]^2 \\
     & + 2 \left[  \frac{1}{N}\sum_{i=1}^Nn^2p_{ij}p_{ik}H_{ij,\ell}H_{ik,\ell'} -  [\mathbb E (s_{Y,ij}H_{ij,\ell}|X_i,A_i)][\mathbb E (s_{Y,ij}H_{ij,\ell'}|X_i,A_i)]\right]^2 \\
     & \leq  C\left[  \frac{1}{N}\sum_{i=1}^N(s_{Y,ij} s_{Y,ik} - n^2p_{ij}p_{ik} )H_{ij,\ell}H_{ik,\ell'}\right]^2 + C \zeta_n^4.
\end{align*}
and thus, 
\begin{align*}
& \mathbb E \left( ||II_2||_F^2 \right) \\
& \leq \frac{C}{M^4} \sum_{(\ell,\ell') \in [k_n] \times [k_n]}\sum_{(j,k) \in [M] \times [M], j\neq k} \mathbb E \left( f_{\ell,\ell'}^2(\mathcal D_j, \mathcal D_k) \right) \\
& \leq \frac{C}{M^4} \sum_{(\ell,\ell') \in [k_n] \times [k_n]}\sum_{(j,k) \in [M] \times [M], j\neq k} \left\{ \mathbb E \left[  \frac{1}{N}\sum_{i=1}^N(s_{Y,ij} s_{Y,ik} - n^2p_{ij}p_{ik} )H_{ij,\ell}H_{ik,\ell'}\right]^2 + \zeta_n^4  \right\} \\
& = \frac{1}{M^4} \sum_{(\ell,\ell') \in [k_n] \times [k_n]}\sum_{(j,k) \in [M] \times [M], j\neq k} \mathbb E \left\{  \mathbb E\left( \left[  \frac{1}{N}\sum_{i=1}^N(s_{Y,ij} s_{Y,ik} - n^2p_{ij}p_{ik} )H_{ij,\ell}H_{ik,\ell'}\right]^2 \mid  \mathcal X_N, \mathcal A_N, \mathcal W_M, \mathcal B_M \right) \right\}  \\
& + \frac{C k_n^2  \zeta_n^4}{n^2} \\
& = \frac{C}{M^4N^2} \sum_{(\ell,\ell') \in [k_n] \times [k_n]}\sum_{(j,k) \in [M] \times [M], j\neq k}   \mathbb E\left[(s_{Y,ij} s_{Y,ik} - n^2p_{ij}p_{ik} )H_{ij,\ell}H_{ik,\ell'}\right]^2 + \frac{C k_n^2  \zeta_n^4}{n^2} \\
& \leq \frac{n^2 k_n \zeta_n^2}{N^2 M^4} \sum_{(\ell,\ell') \in [k_n] \times [k_n]}\sum_{(j,k) \in [M] \times [M], j\neq k} \sum_{i =1}^N \mathbb E H_{ij,\ell}^2 + \frac{C k_n^2  \zeta_n^4}{n^2} \\
& = O\left( \frac{k_n^2 \zeta_n^2}{n} \right).
\end{align*}
This implies 
\begin{align}\label{eq:II2_final}
    ||II_2||_{op} \leq ||II_2||_{F} = O_P\left( \sqrt{\frac{k_n^2 \zeta_n^2}{n}} \right).
\end{align}
Combining \eqref{eq:II2_final} with \eqref{eq:II1_final}, we have
\begin{align*}
    ||II||_{op}  = O_P\left( \sqrt{\frac{k_n^2 \zeta_n^2}{n}} \right).
\end{align*}

\textbf{Bound for $||III||_{op}$.} By a similar argument to the analysis for $II_1$, we have 
\begin{align}\label{eq:III_final}
 ||III||_{op} =  O_P\left(  \sqrt{\frac{k_n^2 \zeta_n^2}{n} } \right).
\end{align}

Combining \eqref{eq:I_final}, \eqref{eq:II2_final}, and \eqref{eq:III_final}, we have the desired result that 
\begin{align*}
    || \hat \Sigma^a_{YY,n} -  \Sigma^a_{YY,n}||_{op} = O_P\left(  \sqrt{\frac{k_n^2 \zeta_n^2}{n} } \right).
\end{align*}    
The other result can be established in the same manner. 
\end{proof}

\newpage 
\section{Additional Tables}
%--------------------------------------------------------------------
% Additional Empirical Application 
%--------------------------------------------------------------------
\subsection{Empirical Application}
\begin{table}[ht!]
\centering
\sisetup{round-mode=places}
\begin{threeparttable}
\caption{Results Robust Estimator with $\overline{\sigma}^2 = 1$}
\label{table:robust_appx_sigma_1}
\begin{tabular}{ l S[round-precision=3] S[round-precision=3] S[round-precision=3]
    >{{[}} % Add square bracket before column
    S[round-precision=3, table-format = -1.3,table-space-text-pre={[}]
    @{,\,} % Add comma and thin-space between the columns
    S[round-precision=3, table-format = -1.3,table-space-text-post={]}]
    <{{]}} % Add square bracket after column
    S[round-precision=3]
    S[round-precision=3] 
  }    
\toprule
                    & {coeff.} & {s.e.} & {max rMSE} & \multicolumn{2}{c}{{[conf. int.]}} & {length} & {adjm} \\
\midrule
\multicolumn{8}{c}{Homophily parameters} \\
jelcode       &  1.611004 &  0.043573 &   5.402529 &  1.443462 &  1.778546 &  0.335084 &  0.104304 \\
citations     &  0.390156 &  0.014674 &   5.075579 &  0.275282 &  0.505030 &  0.229748 &  0.086539 \\
gender        &  5.895072 &  0.180907 &  35.255671 &  4.961321 &  6.828824 &  1.867503 &  0.438419 \\
junior\_senior & -3.638509 &  0.129350 &  13.234874 & -4.079521 & -3.197497 &  0.882025 & -0.590508 \\
constant      & -0.831788 &  0.077303 &   8.158071 & -1.100510 & -0.563067 &  0.537443 & -0.066701 \\
\multicolumn{8}{c}{Average out-degree} \\
Psi           &  1.745033 &  0.031788 &   1.896677 &  1.675307 &  1.814759 &  0.139452 & -0.026189 \\
\bottomrule
\end{tabular}    
\begin{tablenotes}
\item[$^{1}$] \footnotesize{Total sample includes $N=1776$ authors and $M=1600$ articles.}
\end{tablenotes}
\end{threeparttable}
\end{table}

\begin{table}[ht!]
\centering
\sisetup{round-mode=places}
\begin{threeparttable}
\caption{Results Robust Estimator with $\overline{\sigma}^2 = 2$}
\label{table:robust_appx_sigma_2}
\begin{tabular}{ l S[round-precision=3] S[round-precision=3] S[round-precision=3]
    >{{[}} % Add square bracket before column
    S[round-precision=3, table-format = -1.3,table-space-text-pre={[}]
    @{,\,} % Add comma and thin-space between the columns
    S[round-precision=3, table-format = -1.3,table-space-text-post={]}]
    <{{]}} % Add square bracket after column
    S[round-precision=3]
    S[round-precision=3] 
  }    
\toprule
                    & {coeff.} & {s.e.} & {max rMSE} & \multicolumn{2}{c}{{[conf. int.]}} & {length} & {adjm} \\
\midrule
\multicolumn{8}{c}{Homophily parameters} \\
jelcode       &  1.644416 &  0.047269 &   8.015726 &  1.422165 &  1.866667 &  0.444503 &  0.137716 \\
citations     &  0.401171 &  0.013523 &   7.405562 &  0.247931 &  0.554412 &  0.306481 &  0.097555 \\
gender        &  5.969021 &  0.178334 &  50.036111 &  4.777003 &  7.161039 &  2.384035 &  0.512368 \\
junior\_senior & -3.669652 &  0.132803 &  17.742208 & -4.204906 & -3.134397 &  1.070509 & -0.621650 \\
constant      & -0.847101 &  0.086964 &  12.232393 & -1.209275 & -0.484927 &  0.724348 & -0.082013 \\
\multicolumn{8}{c}{Average out-degree} \\
Psi           &  1.757651 &  0.031369 &   2.004229 &  1.681822 &  1.833480 &  0.151658 & -0.013571 \\
\bottomrule
\end{tabular}    
\begin{tablenotes}
\item[$^{1}$] \footnotesize{Total sample includes $N=1776$ authors and $M=1600$ articles.}
\end{tablenotes}
\end{threeparttable}
\end{table}

\begin{table}[ht!]
\centering
\sisetup{round-mode=places}
\begin{threeparttable}
\caption{Results Robust Estimator with $\overline{\sigma}^2 = 3$}
\label{table:robust_appx_sigma_3}
\begin{tabular}{ l S[round-precision=3] S[round-precision=3] S[round-precision=3]
    >{{[}} % Add square bracket before column
    S[round-precision=3, table-format = -1.3,table-space-text-pre={[}]
    @{,\,} % Add comma and thin-space between the columns
    S[round-precision=3, table-format = -1.3,table-space-text-post={]}]
    <{{]}} % Add square bracket after column
    S[round-precision=3]
    S[round-precision=3] 
  }    
\toprule
                    & {coeff.} & {s.e.} & {max rMSE} & \multicolumn{2}{c}{{[conf. int.]}} & {length} & {adjm} \\
\midrule
\multicolumn{8}{c}{Homophily parameters} \\
jelcode       &  1.649930 &  0.045266 &   9.086640 &  1.411517 &  1.888343 &  0.476826 &  0.143229 \\
citations     &  0.401535 &  0.013583 &   8.874086 &  0.222789 &  0.580281 &  0.357492 &  0.097918 \\
gender        &  5.979310 &  0.178768 &  60.135444 &  4.609513 &  7.349107 &  2.739594 &  0.522656 \\
junior\_senior & -3.677594 &  0.131637 &  20.771380 & -4.267970 & -3.087218 &  1.180752 & -0.629592 \\
constant      & -0.847979 &  0.083174 &  13.888473 & -1.235091 & -0.460867 &  0.774224 & -0.082891 \\
\multicolumn{8}{c}{Average out-degree} \\
Psi           &  1.778314 &  0.031203 &   2.348697 &  1.691460 &  1.865168 &  0.173709 &  0.007091 \\
\bottomrule
\end{tabular}    
\begin{tablenotes}
\item[$^{1}$] \footnotesize{Total sample includes $N=1776$ authors and $M=1600$ articles.}
\end{tablenotes}
\end{threeparttable}
\end{table}

%-------------------------------------------------------------------------------------------------
% Additional Simulations
%-------------------------------------------------------------------------------------------------
\subsection{Additional DGPs for Local Misspecification}
%--------------------------------------------------------------------%
% DGP: Functional 
\begin{landscape}
\begin{sidewaystable}
\centering    
\sisetup{round-mode=places}
\caption{$\hat \beta_n$ Robust Estimator}
\label{table:beta_robust_DGP_functional}
\resizebox{\textwidth}{!}{%
\begin{tabular}{ l S[round-precision=3] S[round-precision=3] S[round-precision=3] S[round-precision=3] S[round-precision=3] S[round-precision=3]
>{{[}} % Add square bracket before column
S[round-precision=3, table-format = -1.3,table-space-text-pre={[}]
@{,\,} % Add comma and thin-space between the columns
S[round-precision=3, table-format = -1.3,table-space-text-post={]}]
<{{]}} % Add square bracket after column
S[round-precision=3] S[round-precision=3] S[round-precision=3]
}  
\toprule
$\overline{\sigma}^2$ & {True} &  {coeff.} & {$\sqrt{n}$-Bias} & {s.e.} & {s.e./sd} & {rMSE} & \multicolumn{2}{c}{{[conf. int.]}} & {length} & {95\% CP} & {Degree} \\
\midrule
        \multicolumn{12}{c}{$n=200$} \\
0.0    &  1.386294 &  1.384514 &  -0.024659 &  0.146408 &  0.877630 &  2.146430 &  0.657083 &  2.112019 &   1.454936 &  0.9990 &  0.010340 \\
1.0    &  1.386294 &  1.686582 &   4.216354 &  0.144855 &  0.875421 &  2.158414 &  0.810777 &  2.558094 &   1.747316 &  1.0000 &  0.011391 \\
2.0    &  1.386294 &  1.773942 &   5.428478 &  0.144617 &  0.887350 &  2.172111 &  0.841574 &  2.698718 &   1.857144 &  1.0000 &  0.012113 \\
3.0    &  1.386294 &  1.841166 &   6.395499 &  0.144579 &  0.885337 &  2.200464 &  0.858749 &  2.818300 &   1.959551 &  1.0000 &  0.012988 \\
4.0    &  1.386294 &  1.888153 &   7.037935 &  0.146760 &  0.900354 &  2.251561 &  0.861326 &  2.906577 &   2.045252 &  0.9995 &  0.014025 \\
        \multicolumn{12}{c}{$n=400$} \\
\bottomrule
\end{tabular}
}
\begin{flushleft}
\footnotesize{$^{1}$ Number of Monte Carlo simulations is $2,000$. $^{2}$ Local misspecification: Functional misspecification. $^{3}$ Sieves dimension $k_n=2$.}    
\end{flushleft}
\end{sidewaystable}
\end{landscape}

\begin{landscape}
\begin{sidewaystable}
\centering
\sisetup{round-mode=places}
\caption{$\hat \beta_n$ Logistic Estimator}
\label{table:beta_logistic_DGP_functional}
\resizebox{\textwidth}{!}{%
\begin{tabular}{ l S[round-precision=3] S[round-precision=3] S[round-precision=3] S[round-precision=3] S[round-precision=3] S[round-precision=3]
>{{[}} % Add square bracket before column
S[round-precision=3, table-format = -1.3,table-space-text-pre={[}]
@{,\,} % Add comma and thin-space between the columns
S[round-precision=3, table-format = -1.3,table-space-text-post={]}]
<{{]}} % Add square bracket after column
S[round-precision=3] S[round-precision=3] S[round-precision=3]
}  
\toprule
$\overline{\sigma}^2$ & {True} &  {coeff.} & {$\sqrt{n}$-Bias} & {s.e.} & {s.e./sd} & {rMSE} & \multicolumn{2}{c}{{[conf. int.]}} & {length} & {95\% CP} & {Degree} \\
\midrule
        \multicolumn{12}{c}{$n=200$} \\
        0.0    &  1.386294 &  1.418720 &   0.475257 &  0.159111 &  0.894407 &  2.343308 &  0.665750 &  2.174050 &   1.508300 &  0.9990 &  0.010340 \\
1.0    &  1.386294 &  1.751088 &   5.159177 &  0.167114 &  0.895463 &  2.498323 &  0.832968 &  2.669238 &   1.836270 &  1.0000 &  0.011391 \\
2.0    &  1.386294 &  1.854566 &   6.583693 &  0.170836 &  0.907586 &  2.565695 &  0.871814 &  2.831849 &   1.960035 &  0.9995 &  0.012113 \\
3.0    &  1.386294 &  1.939792 &   7.770508 &  0.174234 &  0.906700 &  2.632011 &  0.899618 &  2.971886 &   2.072268 &  0.9995 &  0.012988 \\
4.0    &  1.386294 &  2.000181 &   8.618930 &  0.177039 &  0.905496 &  2.703243 &  0.914384 &  3.077105 &   2.162721 &  0.9995 &  0.014025 \\
        \multicolumn{12}{c}{$n=400$} \\
\bottomrule
\end{tabular}
}
\begin{flushleft}
\footnotesize{$^{1}$ Number of Monte Carlo simulations is $2,000$. $^{2}$ Local misspecification: Functional misspecification.}    
\end{flushleft}
\end{sidewaystable}
\end{landscape}

\begin{landscape}
\begin{sidewaystable}
\centering
\sisetup{round-mode=places}
\caption{$\hat \beta_n$ Poisson Estimator}
\label{table:beta_poisson_DGP_functional}
\resizebox{\textwidth}{!}{%
\begin{tabular}{ l S[round-precision=3] S[round-precision=3] S[round-precision=3] S[round-precision=3] S[round-precision=3] S[round-precision=3]
>{{[}} % Add square bracket before column
S[round-precision=3, table-format = -1.3,table-space-text-pre={[}]
@{,\,} % Add comma and thin-space between the columns
S[round-precision=3, table-format = -1.3,table-space-text-post={]}]
<{{]}} % Add square bracket after column
S[round-precision=3] S[round-precision=3] S[round-precision=3]
}  
\toprule
$\overline{\sigma}^2$ & {True} &  {coeff.} & {$\sqrt{n}$-Bias} & {s.e.} & {s.e./sd} & {rMSE} & \multicolumn{2}{c}{{[conf. int.]}} & {length} & {95\% CP} & {Degree} \\
\midrule
        \multicolumn{12}{c}{$n=200$} \\
0.0    &  1.386294 &  1.383727 &  -0.041355 &  0.149168 &  0.897599 &  2.190993 &  0.665004 &  2.101736 &   1.436732 &  0.9990 &  0.010340 \\
1.0    &  1.386294 &  1.682705 &   4.157707 &  0.147297 &  0.901461 &  2.191474 &  0.843855 &  2.516722 &   1.672866 &  1.0000 &  0.011391 \\
2.0    &  1.386294 &  1.767815 &   5.340462 &  0.144888 &  0.904331 &  2.188585 &  0.889544 &  2.638300 &   1.748756 &  1.0000 &  0.012113 \\
3.0    &  1.386294 &  1.832488 &   6.271199 &  0.143504 &  0.897060 &  2.187704 &  0.924163 &  2.735308 &   1.811145 &  0.9995 &  0.012988 \\
4.0    &  1.386294 &  1.875418 &   6.872929 &  0.142311 &  0.894594 &  2.202161 &  0.944376 &  2.800192 &   1.855816 &  0.9995 &  0.014025 \\
        \multicolumn{12}{c}{$n=400$} \\
\bottomrule
\end{tabular}
}
\begin{flushleft}
\footnotesize{$^{1}$ Number of Monte Carlo simulations is $2,000$. $^{2}$ Local misspecification: Functional misspecification.} 
\end{flushleft}
\end{sidewaystable}
\end{landscape}

% Ratios Beta DGP Functional
\begin{table}[ht!]
\centering
\sisetup{round-mode=places}
\begin{threeparttable}
\caption{$\hat \beta_n$ Ratios}
\label{table:beta_ratios_DGP_functional}
\begin{tabular}{ l 
S[round-precision=3] S[round-precision=3] S[round-precision=3] S[round-precision=3] 
S[round-precision=3] S[round-precision=3] S[round-precision=3] S[round-precision=3]
}  
\toprule
& \multicolumn{4}{c}{Logit/Robust} & \multicolumn{4}{c}{Poisson/Robust} \\
\cline{2-5}\cline{6-9}
$\overline{\sigma}^2$ & {s.e} & {bias} &  {rMSE} &  {length} & {s.e} & {bias} & {rMSE} &  {length} \\
\midrule
        \multicolumn{9}{c}{$n=200$} \\
        0.0 &  1.086761 &   19.273246 &  1.091724 &  1.036678 &  1.018852 &  1.677067 &  1.020762 &  0.987488 \\
1.0 &  1.059730 &    0.767770 &  1.055119 &  1.023802 &  1.010142 &  1.005767 &  1.009210 &  0.989885 \\
2.0 &  1.053060 &    0.811450 &  1.051472 &  1.022565 &  1.005981 &  1.004494 &  1.008262 &  0.990432 \\
3.0 &  1.050598 &    0.835726 &  1.048493 &  1.021487 &  1.005377 &  1.003797 &  1.007396 &  0.990717 \\
4.0 &  1.047882 &    0.852260 &  1.046374 &  1.020736 &  1.005836 &  1.003336 &  1.006822 &  0.990968 \\
\bottomrule
\end{tabular}
\begin{tablenotes}
\item[$^{1}$] \footnotesize{Number of Monte Carlo simulations is $2,000$.} 
\item[$^{2}$] \footnotesize{DGP Local misspecification: Functional misspecification.} 
\item[$^{3}$] \footnotesize{Sieves dimension $k_n=2$.}     
\end{tablenotes}
\end{threeparttable}
\end{table}

% Set DGP: Normal
\begin{landscape}
\begin{sidewaystable}
\centering    
\sisetup{round-mode=places}
\caption{$\hat \beta_n$ Robust Estimator}
\label{table:beta_robust_DGP_normal}
\resizebox{\textwidth}{!}{%
\begin{tabular}{ l S[round-precision=3] S[round-precision=3] S[round-precision=3] S[round-precision=3] S[round-precision=3] S[round-precision=3]
>{{[}} % Add square bracket before column
S[round-precision=3, table-format = -1.3,table-space-text-pre={[}]
@{,\,} % Add comma and thin-space between the columns
S[round-precision=3, table-format = -1.3,table-space-text-post={]}]
<{{]}} % Add square bracket after column
S[round-precision=3] S[round-precision=3] S[round-precision=3]
}  
\toprule
$\overline{\sigma}^2$ & {True} &  {coeff.} & {$\sqrt{n}$-Bias} & {s.e.} & {s.e./sd} & {rMSE} & \multicolumn{2}{c}{{[conf. int.]}} & {length} & {95\% CP} & {Degree} \\
\midrule
        \multicolumn{12}{c}{$n=200$} \\
        0.0    &  1.386294 &  1.384514 &  -0.024659 &  0.146408 &  0.877630 &  2.146430 &  0.657083 &  2.112019 &   1.454936 &   0.999 &  0.010340 \\
1.0    &  1.386294 &  1.529299 &   1.990244 &  0.145564 &  0.881393 &  2.140863 &  0.738379 &  2.315673 &   1.577295 &   1.000 &  0.010706 \\
2.0    &  1.386294 &  1.555891 &   2.387240 &  0.146038 &  0.883692 &  2.144603 &  0.752400 &  2.357796 &   1.605396 &   1.000 &  0.010816 \\
3.0    &  1.386294 &  1.579843 &   2.715649 &  0.145861 &  0.879377 &  2.145365 &  0.763842 &  2.392797 &   1.628955 &   1.000 &  0.010910 \\
4.0    &  1.386294 &  1.594941 &   2.974162 &  0.145793 &  0.877799 &  2.143638 &  0.773122 &  2.420077 &   1.646955 &   1.000 &  0.010997 \\
        \multicolumn{12}{c}{$n=400$} \\
\bottomrule
\end{tabular}
}
\begin{flushleft}
\footnotesize{$^{1}$ Number of Monte Carlo simulations is $2,000$. $^{2}$ Local misspecification: Semiparametric distribution. $^{3}$ Sieves dimension $k_n=2$.}    
\end{flushleft}
\end{sidewaystable}
\end{landscape}

\begin{landscape}
\begin{sidewaystable}
\centering
\sisetup{round-mode=places}
\caption{$\hat \beta_n$ Logistic Estimator}
\label{table:beta_logistic_DGP_normal}
\resizebox{\textwidth}{!}{%
\begin{tabular}{ l S[round-precision=3] S[round-precision=3] S[round-precision=3] S[round-precision=3] S[round-precision=3] S[round-precision=3]
>{{[}} % Add square bracket before column
S[round-precision=3, table-format = -1.3,table-space-text-pre={[}]
@{,\,} % Add comma and thin-space between the columns
S[round-precision=3, table-format = -1.3,table-space-text-post={]}]
<{{]}} % Add square bracket after column
S[round-precision=3] S[round-precision=3] S[round-precision=3]
}  
\toprule
$\overline{\sigma}^2$ & {True} &  {coeff.} & {$\sqrt{n}$-Bias} & {s.e.} & {s.e./sd} & {rMSE} & \multicolumn{2}{c}{{[conf. int.]}} & {length} & {95\% CP} & {Degree} \\
\midrule
        \multicolumn{12}{c}{$n=200$} \\
        0.0    &  1.386294 &  1.418720 &   0.475257 &  0.159111 &  0.894407 &  2.343308 &  0.665750 &  2.174050 &   1.508300 &   0.999 &  0.010340 \\
1.0    &  1.386294 &  1.577123 &   2.660858 &  0.161544 &  0.897687 &  2.391206 &  0.752479 &  2.396412 &   1.643933 &   1.000 &  0.010706 \\
2.0    &  1.386294 &  1.602933 &   3.099542 &  0.162182 &  0.896040 &  2.407406 &  0.767967 &  2.442964 &   1.674997 &   1.000 &  0.010816 \\
3.0    &  1.386294 &  1.630843 &   3.464868 &  0.162809 &  0.891708 &  2.419836 &  0.780605 &  2.481990 &   1.701386 &   1.000 &  0.010910 \\
4.0    &  1.386294 &  1.648476 &   3.753968 &  0.162546 &  0.886108 &  2.426273 &  0.791040 &  2.512440 &   1.721401 &   1.000 &  0.010997 \\
        \multicolumn{12}{c}{$n=400$} \\
\bottomrule
\end{tabular}
}
\begin{flushleft}
\footnotesize{$^{1}$ Number of Monte Carlo simulations is $2,000$. $^{2}$ Local misspecification: Semiparametric distribution.}    
\end{flushleft}
\end{sidewaystable}
\end{landscape}

\begin{landscape}
\begin{sidewaystable}
\centering
\sisetup{round-mode=places}
\caption{$\hat \beta_n$ Poisson Estimator}
\label{table:beta_poisson_DGP_normal}
\resizebox{\textwidth}{!}{%
\begin{tabular}{ l S[round-precision=3] S[round-precision=3] S[round-precision=3] S[round-precision=3] S[round-precision=3] S[round-precision=3]
>{{[}} % Add square bracket before column
S[round-precision=3, table-format = -1.3,table-space-text-pre={[}]
@{,\,} % Add comma and thin-space between the columns
S[round-precision=3, table-format = -1.3,table-space-text-post={]}]
<{{]}} % Add square bracket after column
S[round-precision=3] S[round-precision=3] S[round-precision=3]
}  
\toprule
$\overline{\sigma}^2$ & {True} &  {coeff.} & {$\sqrt{n}$-Bias} & {s.e.} & {s.e./sd} & {rMSE} & \multicolumn{2}{c}{{[conf. int.]}} & {length} & {95\% CP} & {Degree} \\
\midrule
        \multicolumn{12}{c}{$n=200$} \\
        0.0    &  1.386294 &  1.383727 &  -0.041355 &  0.149168 &  0.897599 &  2.190993 &  0.665004 &  2.101736 &   1.436732 &   0.999 &  0.010340 \\
1.0    &  1.386294 &  1.527289 &   1.960685 &  0.147879 &  0.901424 &  2.182735 &  0.754973 &  2.294899 &   1.539926 &   1.000 &  0.010706 \\
2.0    &  1.386294 &  1.553730 &   2.353922 &  0.148148 &  0.903224 &  2.184778 &  0.771421 &  2.334062 &   1.562641 &   1.000 &  0.010816 \\
3.0    &  1.386294 &  1.576709 &   2.678669 &  0.148581 &  0.903324 &  2.183876 &  0.785089 &  2.366321 &   1.581232 &   1.000 &  0.010910 \\
4.0    &  1.386294 &  1.592090 &   2.934166 &  0.147771 &  0.897667 &  2.180503 &  0.796227 &  2.391315 &   1.595088 &   1.000 &  0.010997 \\
        \multicolumn{12}{c}{$n=400$} \\
\bottomrule
\end{tabular}
}
\begin{flushleft}
\footnotesize{$^{1}$ Number of Monte Carlo simulations is $2,000$. $^{2}$ Local misspecification: Semiparametric distribution.}    
\end{flushleft}
\end{sidewaystable}
\end{landscape}

% Ratios Beta DGP Normal
\begin{table}[ht!]
\centering
\sisetup{round-mode=places}
\begin{threeparttable}
\caption{$\hat \beta_n$ Ratios}
\label{table:beta_ratios_DGP_normal}
\begin{tabular}{ l 
S[round-precision=3] S[round-precision=3] S[round-precision=3] S[round-precision=3] 
S[round-precision=3] S[round-precision=3] S[round-precision=3] S[round-precision=3]
}  
\toprule
& \multicolumn{4}{c}{Logit/Robust} & \multicolumn{4}{c}{Poisson/Robust} \\
\cline{2-5}\cline{6-9}
$\overline{\sigma}^2$ & {s.e} & {bias} &  {rMSE} &  {length} & {s.e} & {bias} & {rMSE} &  {length} \\
\midrule
        \multicolumn{9}{c}{$n=200$} \\
        0.0 &  1.086761 &   19.273246 &  1.091724 &  1.036678 &  1.018852 &  1.677067 &  1.020762 &  0.987488 \\
1.0 &  1.109778 &    1.336951 &  1.116936 &  1.042249 &  1.015905 &  0.985148 &  1.019558 &  0.976309 \\
2.0 &  1.110550 &    1.298379 &  1.122542 &  1.043354 &  1.014446 &  0.986043 &  1.018733 &  0.973368 \\
3.0 &  1.116196 &    1.275889 &  1.127937 &  1.044465 &  1.018648 &  0.986382 &  1.017951 &  0.970704 \\
4.0 &  1.114909 &    1.262193 &  1.131848 &  1.045202 &  1.013570 &  0.986552 &  1.017197 &  0.968507 \\
\bottomrule
\end{tabular}
\begin{tablenotes}
\item[$^{1}$] \footnotesize{Number of Monte Carlo simulations is $2,000$.} 
\item[$^{2}$] \footnotesize{DGP Local misspecification: Semiparametric distribution.} 
\item[$^{3}$] \footnotesize{Sieves dimension $k_n=2$.}     
\end{tablenotes}
\end{threeparttable}
\end{table}

%-----------------------------------------------------------------------------------#
% Psi Estimates Normal DGP
%-----------------------------------------------------------------------------------#
% Psi Estimates 
\begin{landscape}    
\begin{sidewaystable}
\centering
\sisetup{round-mode=places}
\caption{$\hat \Psi_n$ Robust Estimator}
\label{table:psi_robust_DGP_functional}
\resizebox{\textwidth}{!}{%
\begin{tabular}{ l S[round-precision=3] S[round-precision=3] S[round-precision=3] S[round-precision=3] S[round-precision=3] S[round-precision=3]
>{{[}} % Add square bracket before column
S[round-precision=3, table-format = -1.3,table-space-text-pre={[}]
@{,\,} % Add comma and thin-space between the columns
S[round-precision=3, table-format = -1.3,table-space-text-post={]}]
<{{]}} % Add square bracket after column
S[round-precision=3] S[round-precision=3] S[round-precision=3]
}  
\toprule
$\overline{\sigma}^2$ & {True} &  {coeff.} & {$\sqrt{n}$-Bias} & {s.e.} & {s.e./sd} & {rMSE} & \multicolumn{2}{c}{{[conf. int.]}} & {length} & {95\% CP} & {Degree} \\
\midrule
        \multicolumn{12}{c}{$n=200$} \\
0.0    &  1.02580 &  1.024305 &   0.132071 &  0.163248 &  0.949301 &  2.406937 &  0.699013 &  1.371265 &   0.672253 &  0.9415 &  0.010374 \\
1.0    &  0.97676 &  0.977692 &   0.104540 &  0.148536 &  0.937129 &  2.161932 &  0.678150 &  1.290154 &   0.612004 &  0.9390 &  0.009876 \\
2.0    &  0.96880 &  0.969474 &   0.093490 &  0.146104 &  0.935266 &  2.125009 &  0.673949 &  1.276873 &   0.602924 &  0.9415 &  0.009791 \\
3.0    &  0.96278 &  0.962588 &   0.086207 &  0.144595 &  0.934822 &  2.098333 &  0.670760 &  1.266991 &   0.596231 &  0.9420 &  0.009729 \\
4.0    &  0.95826 &  0.956332 &   0.077994 &  0.143081 &  0.929289 &  2.079314 &  0.668074 &  1.259477 &   0.591403 &  0.9420 &  0.009680 \\
        \multicolumn{12}{c}{$n=400$} \\
\bottomrule
\end{tabular}
}
\begin{flushleft}
\footnotesize{$^{1}$ Number of Monte Carlo simulations is $2,000$. $^{2}$ Local misspecification: Functional misspecification. $^{3}$  Sieves dimension $k_n=3$.}    
\end{flushleft}
\end{sidewaystable}
\end{landscape}

\begin{landscape}
\begin{sidewaystable}
\centering
\sisetup{round-mode=places}
\caption{$\hat \Psi_n$ Logistic Estimator}
\label{table:psi_logistic_DGP_functional}
\resizebox{\textwidth}{!}{%
\begin{tabular}{ l S[round-precision=3] S[round-precision=3] S[round-precision=3] S[round-precision=3] S[round-precision=3] S[round-precision=3]
>{{[}} % Add square bracket before column
S[round-precision=3, table-format = -1.3,table-space-text-pre={[}]
@{,\,} % Add comma and thin-space between the columns
S[round-precision=3, table-format = -1.3,table-space-text-post={]}]
<{{]}} % Add square bracket after column
S[round-precision=3] S[round-precision=3] S[round-precision=3]
}  
\toprule
$\overline{\sigma}^2$ & {True} &  {coeff.} & {$\sqrt{n}$-Bias} & {s.e.} & {s.e./sd} & {rMSE} & \multicolumn{2}{c}{{[conf. int.]}} & {length} & {95\% CP} & {Degree} \\
\midrule
        \multicolumn{12}{c}{$n=200$} \\
        0.0    &  1.02580 &  1.055048 &   0.606573 &  0.173670 &  0.937740 &  2.619422 &  0.716395 &  1.420987 &   0.704591 &  0.9470 &  0.010374 \\
1.0    &  0.97676 &  1.003766 &   0.486234 &  0.158969 &  0.951981 &  2.329999 &  0.694037 &  1.328246 &   0.634209 &  0.9445 &  0.009876 \\
2.0    &  0.96880 &  0.994521 &   0.464780 &  0.156339 &  0.951632 &  2.288903 &  0.689922 &  1.313408 &   0.623486 &  0.9450 &  0.009791 \\
3.0    &  0.96278 &  0.988562 &   0.450621 &  0.154586 &  0.951816 &  2.259379 &  0.686762 &  1.302526 &   0.615764 &  0.9460 &  0.009729 \\
4.0    &  0.95826 &  0.980112 &   0.437270 &  0.153031 &  0.947485 &  2.237606 &  0.684198 &  1.294162 &   0.609964 &  0.9485 &  0.009680 \\
        \multicolumn{12}{c}{$n=400$} \\
\bottomrule
\end{tabular}
}
\begin{flushleft}
\footnotesize{$^{1}$ Number of Monte Carlo simulations is $2,000$. $^{2}$ Local misspecification: Functional misspecification.}    
\end{flushleft}
\end{sidewaystable}
\end{landscape}

\begin{landscape}
\begin{sidewaystable}
\centering
\sisetup{round-mode=places}
\caption{$\hat \Psi_n$ Poisson Estimator}
\label{table:psi_poisson_DGP_functional}
\resizebox{\textwidth}{!}{%
\begin{tabular}{ l S[round-precision=3] S[round-precision=3] S[round-precision=3] S[round-precision=3] S[round-precision=3] S[round-precision=3]
>{{[}} % Add square bracket before column
S[round-precision=3, table-format = -1.3,table-space-text-pre={[}]
@{,\,} % Add comma and thin-space between the columns
S[round-precision=3, table-format = -1.3,table-space-text-post={]}]
<{{]}} % Add square bracket after column
S[round-precision=3] S[round-precision=3] S[round-precision=3]
}  
\toprule
$\overline{\sigma}^2$ & {True} &  {coeff.} & {$\sqrt{n}$-Bias} & {s.e.} & {s.e./sd} & {rMSE} & \multicolumn{2}{c}{{[conf. int.]}} & {length} & {95\% CP} & {Degree} \\
\midrule
        \multicolumn{12}{c}{$n=200$} \\
0.0    &  1.02580 &  1.03 &   0.163695 &  0.164767 &  0.960152 &  2.415425 &  0.708694 &  1.366056 &   0.657363 &  0.9395 &  0.010374 \\
1.0    &  0.97676 &  0.98 &   0.152735 &  0.152648 &  0.963569 &  2.218442 &  0.684632 &  1.290488 &   0.605856 &  0.9400 &  0.009876 \\
2.0    &  0.96880 &  0.97 &   0.146159 &  0.150162 &  0.961136 &  2.186629 &  0.680456 &  1.277814 &   0.597357 &  0.9430 &  0.009791 \\
3.0    &  0.96278 &  0.97 &   0.142553 &  0.148934 &  0.962710 &  2.163421 &  0.677275 &  1.268445 &   0.591169 &  0.9415 &  0.009729 \\
4.0    &  0.95826 &  0.96 &   0.137037 &  0.147844 &  0.960039 &  2.146083 &  0.674708 &  1.261192 &   0.586483 &  0.9405 &  0.009680 \\
        \multicolumn{12}{c}{$n=400$} \\
\bottomrule
\end{tabular}
}
\begin{flushleft}
\footnotesize{$^{1}$ Number of Monte Carlo simulations is $2,000$. $^{2}$ Local misspecification: Functional misspecification.}    
\end{flushleft}
\end{sidewaystable}
\end{landscape}

% Ratios Psi DGP Functional
\begin{table}[ht!]
\centering
\sisetup{round-mode=places}
\begin{threeparttable}
\caption{$\hat \Psi_n$ Ratios}
\label{table:psi_ratios_DGP_functional}
\begin{tabular}{ l 
S[round-precision=3] S[round-precision=3] S[round-precision=3] S[round-precision=3] 
S[round-precision=3] S[round-precision=3] S[round-precision=3] S[round-precision=3]
}  
\toprule
& \multicolumn{4}{c}{Logit/Robust} & \multicolumn{4}{c}{Poisson/Robust} \\
\cline{2-5}\cline{6-9}
$\overline{\sigma}^2$ & {s.e} & {bias} &  {rMSE} &  {length} & {s.e} & {bias} & {rMSE} &  {length} \\
\midrule
        \multicolumn{9}{c}{$n=200$} \\
        0.0 &  1.063836 &  4.592779 &  1.088280 &  1.048105 &  1.009300 &  1.239448 &  1.003526 &  0.977851 \\
1.0 &  1.070242 &  4.651152 &  1.077739 &  1.036282 &  1.027687 &  1.461014 &  1.026139 &  0.989954 \\
2.0 &  1.070049 &  4.971460 &  1.077126 &  1.034104 &  1.027769 &  1.563371 &  1.028997 &  0.990766 \\
3.0 &  1.069095 &  5.227203 &  1.076750 &  1.032762 &  1.030003 &  1.653610 &  1.031019 &  0.991511 \\
4.0 &  1.069541 &  5.606453 &  1.076127 &  1.031385 &  1.033289 &  1.757021 &  1.032111 &  0.991682 \\
\bottomrule
\end{tabular}
\begin{tablenotes}
\item[$^{1}$] \footnotesize{Number of Monte Carlo simulations is $2,000$.} 
\item[$^{2}$] \footnotesize{DGP Local misspecification: Functional misspecification.} 
\item[$^{3}$] \footnotesize{Sieves dimension $k_n=3$.}     
\end{tablenotes}
\end{threeparttable}
\end{table}

\begin{landscape}    
\begin{sidewaystable}
\centering
\sisetup{round-mode=places}
\caption{$\hat \Psi_n$ Robust Estimator}
\label{table:psi_robust_DGP_normal}
\resizebox{\textwidth}{!}{%
\begin{tabular}{ l S[round-precision=3] S[round-precision=3] S[round-precision=3] S[round-precision=3] S[round-precision=3] S[round-precision=3]
>{{[}} % Add square bracket before column
S[round-precision=3, table-format = -1.3,table-space-text-pre={[}]
@{,\,} % Add comma and thin-space between the columns
S[round-precision=3, table-format = -1.3,table-space-text-post={]}]
<{{]}} % Add square bracket after column
S[round-precision=3] S[round-precision=3] S[round-precision=3]
}  
\toprule
$\overline{\sigma}^2$ & {True} &  {coeff.} & {$\sqrt{n}$-Bias} & {s.e.} & {s.e./sd} & {rMSE} & \multicolumn{2}{c}{{[conf. int.]}} & {length} & {95\% CP} & {Degree} \\
\midrule
        \multicolumn{12}{c}{$n=200$} \\
        0.0    &  1.02580 &  1.024305 &   0.132071 &  0.163248 &  0.949301 &  2.406937 &  0.699013 &  1.371265 &   0.672253 &  0.9415 &  0.010374 \\
1.0    &  1.06094 &  1.060609 &   0.153901 &  0.177371 &  0.933111 &  2.653598 &  0.703727 &  1.439918 &   0.736192 &  0.9355 &  0.010738 \\
2.0    &  1.07124 &  1.070688 &   0.157947 &  0.181170 &  0.934164 &  2.715725 &  0.706379 &  1.458438 &   0.752059 &  0.9355 &  0.010843 \\
3.0    &  1.08122 &  1.080256 &   0.148388 &  0.184393 &  0.932964 &  2.768020 &  0.708894 &  1.474532 &   0.765638 &  0.9350 &  0.010935 \\
4.0    &  1.08852 &  1.089941 &   0.170047 &  0.187916 &  0.935488 &  2.812857 &  0.711822 &  1.489266 &   0.777443 &  0.9330 &  0.011023 \\
        \multicolumn{12}{c}{$n=400$} \\
\bottomrule
\end{tabular}
}
\begin{flushleft}
\footnotesize{$^{1}$ Number of Monte Carlo simulations is $2,000$. $^{2}$ Local misspecification: Semiparametric distribution. $^{3}$ Sieves dimension $k_n=3$.}    
\end{flushleft}
\end{sidewaystable}
\end{landscape}

\begin{landscape}
\begin{sidewaystable}
\centering
\sisetup{round-mode=places}
\caption{$\hat \Psi_n$ Logistic Estimator}
\label{table:psi_logistic_DGP_normal}
\resizebox{\textwidth}{!}{%
\begin{tabular}{ l S[round-precision=3] S[round-precision=3] S[round-precision=3] S[round-precision=3] S[round-precision=3] S[round-precision=3]
>{{[}} % Add square bracket before column
S[round-precision=3, table-format = -1.3,table-space-text-pre={[}]
@{,\,} % Add comma and thin-space between the columns
S[round-precision=3, table-format = -1.3,table-space-text-post={]}]
<{{]}} % Add square bracket after column
S[round-precision=3] S[round-precision=3] S[round-precision=3]
}  
\toprule
$\overline{\sigma}^2$ & {True} &  {coeff.} & {$\sqrt{n}$-Bias} & {s.e.} & {s.e./sd} & {rMSE} & \multicolumn{2}{c}{{[conf. int.]}} & {length} & {95\% CP} & {Degree} \\
\midrule
        \multicolumn{12}{c}{$n=200$} \\
        0.0    &  1.02580 &  1.055048 &   0.606573 &  0.173670 &  0.937740 &  2.619422 &  0.716395 &  1.420987 &   0.704591 &  0.9470 &  0.010374 \\
1.0    &  1.06094 &  1.101049 &   0.774626 &  0.193050 &  0.913115 &  3.029738 &  0.716511 &  1.514918 &   0.798407 &  0.9475 &  0.010738 \\
2.0    &  1.07124 &  1.114486 &   0.819577 &  0.197774 &  0.908496 &  3.148862 &  0.717427 &  1.540959 &   0.823531 &  0.9500 &  0.010843 \\
3.0    &  1.08122 &  1.123907 &   0.846271 &  0.201319 &  0.900416 &  3.252440 &  0.718425 &  1.563696 &   0.845272 &  0.9490 &  0.010935 \\
4.0    &  1.08852 &  1.132507 &   0.900410 &  0.206268 &  0.901123 &  3.347883 &  0.719825 &  1.584552 &   0.864727 &  0.9485 &  0.011023 \\
        \multicolumn{12}{c}{$n=400$} \\
\bottomrule
\end{tabular}
}
\begin{flushleft}
\footnotesize{$^{1}$ Number of Monte Carlo simulations is $2,000$. $^{2}$ Local misspecification: Semiparametric distribution.}    
\end{flushleft}
\end{sidewaystable}
\end{landscape}

\begin{landscape}
\begin{sidewaystable}
\centering
\sisetup{round-mode=places}
\caption{$\hat \Psi_n$ Poisson Estimator}
\label{table:psi_poisson_DGP_normal}
\resizebox{\textwidth}{!}{%
\begin{tabular}{ l S[round-precision=3] S[round-precision=3] S[round-precision=3] S[round-precision=3] S[round-precision=3] S[round-precision=3]
>{{[}} % Add square bracket before column
S[round-precision=3, table-format = -1.3,table-space-text-pre={[}]
@{,\,} % Add comma and thin-space between the columns
S[round-precision=3, table-format = -1.3,table-space-text-post={]}]
<{{]}} % Add square bracket after column
S[round-precision=3] S[round-precision=3] S[round-precision=3]
}  
\toprule
$\overline{\sigma}^2$ & {True} &  {coeff.} & {$\sqrt{n}$-Bias} & {s.e.} & {s.e./sd} & {rMSE} & \multicolumn{2}{c}{{[conf. int.]}} & {length} & {95\% CP} & {Degree} \\
\midrule
        \multicolumn{12}{c}{$n=200$} \\
        0.0    &  1.02580 &  1.03 &   0.163695 &  0.164767 &  0.960152 &  2.415425 &  0.708694 &  1.366056 &   0.657363 &  0.9395 &  0.010374 \\
1.0    &  1.06094 &  1.06 &   0.181797 &  0.179094 &  0.944126 &  2.658139 &  0.713592 &  1.433998 &   0.720406 &  0.9335 &  0.010738 \\
2.0    &  1.07124 &  1.07 &   0.184555 &  0.182730 &  0.944255 &  2.719426 &  0.716159 &  1.452421 &   0.736261 &  0.9325 &  0.010843 \\
3.0    &  1.08122 &  1.08 &   0.174302 &  0.186006 &  0.943075 &  2.771282 &  0.718680 &  1.468410 &   0.749730 &  0.9315 &  0.010935 \\
4.0    &  1.08852 &  1.09 &   0.194949 &  0.189415 &  0.944817 &  2.817461 &  0.721446 &  1.483164 &   0.761718 &  0.9305 &  0.011023 \\
        \multicolumn{12}{c}{$n=400$} \\
\bottomrule
\end{tabular}
}
\begin{flushleft}
\footnotesize{$^{1}$ Number of Monte Carlo simulations is $2,000$. $^{2}$ Local misspecification: Semiparametric distribution.}    
\end{flushleft}
\end{sidewaystable}
\end{landscape}

% Ratios Psi DGP Normal
\begin{table}[ht!]
\centering
\sisetup{round-mode=places}
\begin{threeparttable}
\caption{$\hat \Psi_n$ Ratios}
\label{table:psi_ratios_DGP_normal}
\begin{tabular}{ l 
S[round-precision=3] S[round-precision=3] S[round-precision=3] S[round-precision=3] 
S[round-precision=3] S[round-precision=3] S[round-precision=3] S[round-precision=3]
}  
\toprule
& \multicolumn{4}{c}{Logit/Robust} & \multicolumn{4}{c}{Poisson/Robust} \\
\cline{2-5}\cline{6-9}
$\overline{\sigma}^2$ & {s.e} & {bias} &  {rMSE} &  {length} & {s.e} & {bias} & {rMSE} &  {length} \\
\midrule
        \multicolumn{9}{c}{$n=200$} \\
        0.0 &  1.063836 &  4.592779 &  1.088280 &  1.048105 &  1.009300 &  1.239448 &  1.003526 &  0.977851 \\
1.0 &  1.088397 &  5.033291 &  1.141747 &  1.084509 &  1.009719 &  1.181264 &  1.001711 &  0.978557 \\
2.0 &  1.091647 &  5.188930 &  1.159492 &  1.095035 &  1.008612 &  1.168459 &  1.001363 &  0.978994 \\
3.0 &  1.091790 &  5.703093 &  1.175006 &  1.104009 &  1.008745 &  1.174635 &  1.001178 &  0.979223 \\
4.0 &  1.097662 &  5.295076 &  1.190207 &  1.112270 &  1.007981 &  1.146447 &  1.001637 &  0.979773 \\
\bottomrule
\end{tabular}
\begin{tablenotes}
\item[$^{1}$] \footnotesize{Number of Monte Carlo simulations is $2,000$.} 
\item[$^{2}$] \footnotesize{DGP Local misspecification: Semiparametric distribution.} 
\item[$^{3}$] \footnotesize{Sieves dimension $k_n=3$.}     
\end{tablenotes}
\end{threeparttable}
\end{table}

\subsection{Additional Simulations with Initial Poisson Estimator}\label{sec:appx:sim_poisson_init}
%-------------------------------------------------------------------------------------------------%

\begin{landscape}
\begin{sidewaystable}
\centering    
\sisetup{round-mode=places}
\caption{$\hat \beta_n$ Robust Estimator with $\hat\theta_{\text{initial}}$ Poisson}
\resizebox{\textwidth}{!}{%
\begin{tabular}{ l S[round-precision=3] S[round-precision=3] S[round-precision=3] S[round-precision=3] S[round-precision=3] S[round-precision=3]
>{{[}} % Add square bracket before column
S[round-precision=3, table-format = -1.3,table-space-text-pre={[}]
@{,\,} % Add comma and thin-space between the columns
S[round-precision=3, table-format = -1.3,table-space-text-post={]}]
<{{]}} % Add square bracket after column
S[round-precision=3] S[round-precision=3] S[round-precision=3]
}  
\toprule
$\overline{\sigma}^2$ & {True} &  {coeff.} & {$\sqrt{n}$-Bias} & {s.e.} & {s.e./sd} & {rMSE} & \multicolumn{2}{c}{{[conf. int.]}} & {length} & {95\% CP} & {Degree} \\
\midrule
        \multicolumn{12}{c}{$n=200$} \\
        0.0    &  1.386294 &  1.384936 &  -0.064024 &  0.146513 &  0.884916 &  2.164864 &  0.651498 &  2.112036 &   1.460538 &  0.9995 &  0.010359 \\
1.0    &  1.386294 &  1.382533 &  -0.074874 &  0.146392 &  0.884879 &  2.165600 &  0.650950 &  2.111050 &   1.460100 &  0.9995 &  0.010359 \\
2.0    &  1.386294 &  1.381577 &  -0.073572 &  0.146885 &  0.884407 &  2.166627 &  0.650806 &  2.111378 &   1.460572 &  0.9995 &  0.010359 \\
3.0    &  1.386294 &  1.380808 &  -0.070893 &  0.146731 &  0.883814 &  2.162847 &  0.651332 &  2.111231 &   1.459898 &  1.0000 &  0.010371 \\
4.0    &  1.386294 &  1.380891 &  -0.066203 &  0.146241 &  0.882415 &  2.155220 &  0.652574 &  2.110652 &   1.458078 &  1.0000 &  0.010416 \\
        \multicolumn{12}{c}{$n=400$} \\
\bottomrule
\end{tabular}
}
\begin{flushleft}
\footnotesize{$^{1}$ Number of Monte Carlo simulations is $2,000$. $^{2}$ Local misspecification: Latent Homophily. $^{3}$ Sieves dimension $k_n=2$.}    
\end{flushleft}
\end{sidewaystable}
\end{landscape}

\begin{landscape}
\begin{sidewaystable}
\centering    
\sisetup{round-mode=places}
\caption{$\hat \beta_n$ Robust Estimator with $\hat\theta_{\text{initial}}$ Poisson}
\resizebox{\textwidth}{!}{%
\begin{tabular}{ l S[round-precision=3] S[round-precision=3] S[round-precision=3] S[round-precision=3] S[round-precision=3] S[round-precision=3]
>{{[}} % Add square bracket before column
S[round-precision=3, table-format = -1.3,table-space-text-pre={[}]
@{,\,} % Add comma and thin-space between the columns
S[round-precision=3, table-format = -1.3,table-space-text-post={]}]
<{{]}} % Add square bracket after column
S[round-precision=3] S[round-precision=3] S[round-precision=3]
}  
\toprule
$\overline{\sigma}^2$ & {True} &  {coeff.} & {$\sqrt{n}$-Bias} & {s.e.} & {s.e./sd} & {rMSE} & \multicolumn{2}{c}{{[conf. int.]}} & {length} & {95\% CP} & {Degree} \\
\midrule
        \multicolumn{12}{c}{$n=200$} \\
        0.0    &  1.386294 &  1.384514 &  -0.024659 &  0.146408 &  0.877630 &  2.146430 &  0.657083 &  2.112019 &   1.454936 &  0.9990 &  0.010340 \\
1.0    &  1.386294 &  1.686582 &   4.216354 &  0.144855 &  0.875421 &  2.158414 &  0.810777 &  2.558094 &   1.747316 &  1.0000 &  0.011391 \\
2.0    &  1.386294 &  1.773942 &   5.428478 &  0.144617 &  0.887350 &  2.172111 &  0.841574 &  2.698718 &   1.857144 &  1.0000 &  0.012113 \\
3.0    &  1.386294 &  1.841166 &   6.395499 &  0.144579 &  0.885337 &  2.200464 &  0.858749 &  2.818300 &   1.959551 &  1.0000 &  0.012988 \\
4.0    &  1.386294 &  1.888153 &   7.037935 &  0.146760 &  0.900354 &  2.251561 &  0.861326 &  2.906577 &   2.045252 &  0.9995 &  0.014025 \\
        \multicolumn{12}{c}{$n=400$} \\
\bottomrule
\end{tabular}
}
\begin{flushleft}
\footnotesize{$^{1}$ Number of Monte Carlo simulations is $2,000$. $^{2}$ Local misspecification: Functional misspecification. $^{3}$ Sieves dimension $k_n=2$.}    
\end{flushleft}
\end{sidewaystable}
\end{landscape}

\begin{landscape}
\begin{sidewaystable}
\centering    
\sisetup{round-mode=places}
\caption{$\hat \beta_n$ Robust Estimator with $\hat\theta_{\text{initial}}$ Poisson}
\resizebox{\textwidth}{!}{%
\begin{tabular}{ l S[round-precision=3] S[round-precision=3] S[round-precision=3] S[round-precision=3] S[round-precision=3] S[round-precision=3]
>{{[}} % Add square bracket before column
S[round-precision=3, table-format = -1.3,table-space-text-pre={[}]
@{,\,} % Add comma and thin-space between the columns
S[round-precision=3, table-format = -1.3,table-space-text-post={]}]
<{{]}} % Add square bracket after column
S[round-precision=3] S[round-precision=3] S[round-precision=3]
}  
\toprule
$\overline{\sigma}^2$ & {True} &  {coeff.} & {$\sqrt{n}$-Bias} & {s.e.} & {s.e./sd} & {rMSE} & \multicolumn{2}{c}{{[conf. int.]}} & {length} & {95\% CP} & {Degree} \\
\midrule
        \multicolumn{12}{c}{$n=200$} \\
        0.0    &  1.386294 &  1.384514 &  -0.024659 &  0.146408 &  0.877630 &  2.146430 &  0.657083 &  2.112019 &   1.454936 &   0.999 &  0.010340 \\
1.0    &  1.386294 &  1.529299 &   1.990244 &  0.145564 &  0.881393 &  2.140863 &  0.738379 &  2.315673 &   1.577295 &   1.000 &  0.010706 \\
2.0    &  1.386294 &  1.555891 &   2.387240 &  0.146038 &  0.883692 &  2.144603 &  0.752400 &  2.357796 &   1.605396 &   1.000 &  0.010816 \\
3.0    &  1.386294 &  1.579843 &   2.715649 &  0.145861 &  0.879377 &  2.145365 &  0.763842 &  2.392797 &   1.628955 &   1.000 &  0.010910 \\
4.0    &  1.386294 &  1.594941 &   2.974162 &  0.145793 &  0.877799 &  2.143638 &  0.773122 &  2.420077 &   1.646955 &   1.000 &  0.010997 \\
        \multicolumn{12}{c}{$n=400$} \\
\bottomrule
\end{tabular}
}
\begin{flushleft}
\footnotesize{$^{1}$ Number of Monte Carlo simulations is $2,000$. $^{2}$ Local misspecification: Semiparametric distribution. $^{3}$ Sieves dimension $k_n=2$.}    
\end{flushleft}
\end{sidewaystable}
\end{landscape}

\begin{landscape}    
\begin{sidewaystable}
\centering
\sisetup{round-mode=places}
\caption{$\hat \Psi_n$ Robust Estimator with $\hat\theta_{\text{initial}}$ Poisson}
\resizebox{\textwidth}{!}{%
\begin{tabular}{ l S[round-precision=3] S[round-precision=3] S[round-precision=3] S[round-precision=3] S[round-precision=3] S[round-precision=3]
>{{[}} % Add square bracket before column
S[round-precision=3, table-format = -1.3,table-space-text-pre={[}]
@{,\,} % Add comma and thin-space between the columns
S[round-precision=3, table-format = -1.3,table-space-text-post={]}]
<{{]}} % Add square bracket after column
S[round-precision=3] S[round-precision=3] S[round-precision=3]
}  
\toprule
$\overline{\sigma}^2$ & {True} &  {coeff.} & {$\sqrt{n}$-Bias} & {s.e.} & {s.e./sd} & {rMSE} & \multicolumn{2}{c}{{[conf. int.]}} & {length} & {95\% CP} & {Degree} \\
\midrule
        \multicolumn{12}{c}{$n=200$} \\
        0.0    &  1.04078 &  1.021123 &  -0.118292 &  0.161829 &  0.950031 &  2.405740 &  0.696576 &  1.368255 &   0.671678 &  0.9310 &  0.010348 \\
1.0    &  1.04130 &  1.020989 &  -0.133929 &  0.161802 &  0.949840 &  2.405466 &  0.695999 &  1.367660 &   0.671661 &  0.9295 &  0.010342 \\
2.0    &  1.04112 &  1.020364 &  -0.126456 &  0.161972 &  0.950020 &  2.406518 &  0.696266 &  1.368090 &   0.671824 &  0.9280 &  0.010346 \\
3.0    &  1.04286 &  1.020130 &  -0.134546 &  0.162343 &  0.954913 &  2.408571 &  0.697140 &  1.369552 &   0.672412 &  0.9270 &  0.010357 \\
4.0    &  1.04840 &  1.025599 &  -0.162622 &  0.162867 &  0.956951 &  2.413504 &  0.700003 &  1.373799 &   0.673797 &  0.9255 &  0.010393 \\
        \multicolumn{12}{c}{$n=400$} \\
\bottomrule
\end{tabular}
}
\begin{flushleft}
\footnotesize{$^{1}$ Number of Monte Carlo simulations is $2,000$. $^{2}$ Local misspecification: Latent Homophily. $^{3}$  Sieves dimension $k_n=3$.}    
\end{flushleft}
\end{sidewaystable}
\end{landscape}

\begin{landscape}    
\begin{sidewaystable}
\centering
\sisetup{round-mode=places}
\caption{$\hat \Psi_n$ Robust Estimator with $\hat\theta_{\text{initial}}$ Poisson}
\resizebox{\textwidth}{!}{%
\begin{tabular}{ l S[round-precision=3] S[round-precision=3] S[round-precision=3] S[round-precision=3] S[round-precision=3] S[round-precision=3]
>{{[}} % Add square bracket before column
S[round-precision=3, table-format = -1.3,table-space-text-pre={[}]
@{,\,} % Add comma and thin-space between the columns
S[round-precision=3, table-format = -1.3,table-space-text-post={]}]
<{{]}} % Add square bracket after column
S[round-precision=3] S[round-precision=3] S[round-precision=3]
}  
\toprule
$\overline{\sigma}^2$ & {True} &  {coeff.} & {$\sqrt{n}$-Bias} & {s.e.} & {s.e./sd} & {rMSE} & \multicolumn{2}{c}{{[conf. int.]}} & {length} & {95\% CP} & {Degree} \\
\midrule
        \multicolumn{12}{c}{$n=200$} \\
        0.0    &  1.02580 &  1.024305 &   0.132071 &  0.163248 &  0.949301 &  2.406937 &  0.699013 &  1.371265 &   0.672253 &  0.9415 &  0.010374 \\
1.0    &  0.97676 &  0.977692 &   0.104540 &  0.148536 &  0.937129 &  2.161932 &  0.678150 &  1.290154 &   0.612004 &  0.9390 &  0.009876 \\
2.0    &  0.96880 &  0.969474 &   0.093490 &  0.146104 &  0.935266 &  2.125009 &  0.673949 &  1.276873 &   0.602924 &  0.9415 &  0.009791 \\
3.0    &  0.96278 &  0.962588 &   0.086207 &  0.144595 &  0.934822 &  2.098333 &  0.670760 &  1.266991 &   0.596231 &  0.9420 &  0.009729 \\
4.0    &  0.95826 &  0.956332 &   0.077994 &  0.143081 &  0.929289 &  2.079314 &  0.668074 &  1.259477 &   0.591403 &  0.9420 &  0.009680 \\
        \multicolumn{12}{c}{$n=400$} \\
\bottomrule
\end{tabular}
}
\begin{flushleft}
\footnotesize{$^{1}$ Number of Monte Carlo simulations is $2,000$. $^{2}$ Local misspecification: Functional misspecification. $^{3}$  Sieves dimension $k_n=3$.}    
\end{flushleft}
\end{sidewaystable}
\end{landscape}

\begin{landscape}    
\begin{sidewaystable}
\centering
\sisetup{round-mode=places}
\caption{$\hat \Psi_n$ Robust Estimator with $\hat\theta_{\text{initial}}$ Poisson}
\input{Table_Psi_R02_DGP_Normal_Hn3.tex}
\begin{flushleft}
\footnotesize{$^{1}$ Number of Monte Carlo simulations is $2,000$. $^{2}$ Local misspecification: Semiparametric distribution. $^{3}$ Sieves dimension $k_n=3$.}    
\end{flushleft}
\end{sidewaystable}
\end{landscape}

\subsection{Additional Simulations with Larger Sieves Dimension}\label{sec:appx:sim_larger_sieves}
%-------------------------------------------------------------------------------------------------%
% Beta DGP Homophily
\begin{landscape}    
\begin{sidewaystable}
\centering
\sisetup{round-mode=places}
\caption{$\hat \beta_n$ Robust Estimator}
\label{table:beta_robust_DGP_homophily_kn=4}
\input{Table_Beta_R01_DGP_Homophily_Hn4.tex}
\begin{flushleft}
\footnotesize{$^{1}$ Number of Monte Carlo simulations is $2,000$. $^{2}$ Local misspecification: Latent Homophily. $^{3}$  Sieves dimension $k_n=4$.}    
\end{flushleft}
\end{sidewaystable}
\end{landscape}

\begin{landscape}    
\begin{sidewaystable}
\centering
\sisetup{round-mode=places}
\caption{$\hat \beta_n$ Robust Estimator}
\label{table:beta_robust_DGP_homophily_kn=5}
\input{Table_Beta_R01_DGP_Homophily_Hn5.tex}
\begin{flushleft}
\footnotesize{$^{1}$ Number of Monte Carlo simulations is $2,000$. $^{2}$ Local misspecification: Latent Homophily. $^{3}$  Sieves dimension $k_n=5$.}    
\end{flushleft}
\end{sidewaystable}
\end{landscape}

% Psi
%----------------------------------------------------------------%
\begin{landscape}    
\begin{sidewaystable}
\centering
\sisetup{round-mode=places}
\caption{$\hat \Psi_n$ Robust Estimator}
\label{table:psi_robust_DGP_homophily_kn=4}
\input{Table_Psi_R01_DGP_Homophily_Hn4.tex}
\begin{flushleft}
\footnotesize{$^{1}$ Number of Monte Carlo simulations is $2,000$. $^{2}$ Local misspecification: Latent Homophily. $^{3}$  Sieves dimension $k_n=4$.}    
\end{flushleft}
\end{sidewaystable}
\end{landscape}

\begin{landscape}    
\begin{sidewaystable}
\centering
\sisetup{round-mode=places}
\caption{$\hat \Psi_n$ Robust Estimator}
\label{table:psi_robust_DGP_homophily_kn=5}
\input{Table_Psi_R01_DGP_Homophily_Hn5.tex}
\begin{flushleft}
\footnotesize{$^{1}$ Number of Monte Carlo simulations is $2,000$. $^{2}$ Local misspecification: Latent Homophily. $^{3}$  Sieves dimension $k_n=5$.}    
\end{flushleft}
\end{sidewaystable}
\end{landscape}

%-----------------------------------------------------------------------------------
\clearpage
\newpage
\singlespacing
\bibliography{refero}

\end{document}